\documentclass[final]{siamart0216}

\usepackage{amsmath}
\usepackage{amssymb}
\usepackage{mathrsfs}
\usepackage{mathtools} 

\usepackage{enumitem} 
\setlist{noitemsep,topsep=0pt,parsep=0pt,partopsep=0pt,listparindent=\parindent} 

\usepackage[margin=1cm]{caption} 
\usepackage{subfig}

\usepackage{rotating}
\usepackage{multirow} 

\usepackage{tikz}
\usetikzlibrary{arrows,backgrounds,calc,fit,decorations.pathreplacing,decorations.markings,shapes.geometric}

\tikzset{every fit/.append style=text badly centered}

{\par\medskip}

\def\borderColor{blue!60}
\def\scale{0.6}
\def\nodeDist{1.4cm}
\tikzstyle{internal} = [draw, fill, shape=circle]
\tikzstyle{external} = [shape=circle]
\tikzstyle{square}   = [draw, fill, rectangle]
\tikzstyle{triangle} = [draw, fill, regular polygon, regular polygon sides=3, inner sep=2.5pt] 

\newcommand{\TheTitle}{A Complete Dichotomy Rises from the Capture of Vanishing Signatures} 
\newcommand{\TheAuthors}{J.-Y. Cai, H. Guo, and T. Williams}

\headers{Complex Holant Dichotomy}{\TheAuthors}

\title{{\TheTitle}\thanks{An extended abstract of this paper appeared in STOC 2013 \cite{CGW13}. 
All authors were supported by NSF CCF-0914969 and NSF CCF-1217549.
The work was done while the second and the third authors were graduate students in UW-Madison.}}

\author{
  Jin-Yi Cai\thanks{University of Wisconsin-Madison.
    \email{jyc@cs.wisc.edu}}
  \and
  Heng Guo\thanks{Queen Mary, Univerisity of London, UK. 
    \email{h.guo@qmul.ac.uk}}
  \and
  Tyson Williams\thanks{University of Wisconsin-Madison.
    \email{tdw@cs.wisc.edu}}
}

\ifpdf
\hypersetup{
  pdftitle={\TheTitle},
  pdfauthor={\TheAuthors}
}
\fi

\newcommand{\Holant}{\operatorname{Holant}}
\newcommand{\holant}[2]{\Holant\left(#1\mid #2\right)}
\newcommand{\CSP}{\operatorname{\#CSP}}
\newcommand{\vd}{\operatorname{vd}}
\newcommand{\rd}{\operatorname{rd}}
\newcommand{\arity}{\operatorname{arity}}
\newcommand{\Sym}{\operatorname{Sym}}
\newcommand{\RM}{\operatorname{RM}}
\newcommand{\StabA}{\operatorname{Stab}(\mathscr{A})}
\newcommand{\LStabA}{\operatorname{LStab}(\mathscr{A})}
\newcommand{\RStabA}{\operatorname{RStab}(\mathscr{A})}
\newcommand{\StabP}{\operatorname{Stab}(\mathscr{P})}

\newcommand{\binaryA}{\mathscr{A}^{2 \times 2}}

\newcommand{\numP}{{\rm \#P}}
\newcommand{\trans}[4]{\ensuremath{\left[\begin{smallmatrix} #1 & #2 \\ #3 & #4 \end{smallmatrix}\right]}}

\newcommand{\SHARPP}{{\rm \#P}}

\newcommand{\ceil}[1]{\left\lceil #1 \right\rceil}
\newcommand{\floor}[1]{\left\lfloor #1 \right\rfloor}

\newcommand{\union}{\cup}

\newcommand{\intersect}{\cap}

\newcommand{\st}{\mid}

\newcommand{\tbmatrix}[4]{\left[\begin{smallmatrix} #1 & #2 \\ #3 & #4 \end{smallmatrix}\right]}

\newcommand{\tbcolvec}[2]{\left[\begin{smallmatrix} #1 \\ #2 \end{smallmatrix}\right]}

\begin{document}

\maketitle

\begin{abstract}
 We prove a complexity dichotomy theorem for Holant problems over an arbitrary set of complex-valued symmetric constraint functions $\mathcal{F}$ on Boolean variables.
 This extends and unifies all previous dichotomies for Holant problems on symmetric constraint functions
 (taking values without a finite modulus).
 We define and characterize all symmetric \emph{vanishing} signatures.
 They turned out to be essential to the complete classification of Holant problems.
 The dichotomy theorem has an explicit tractability criterion expressible in terms of holographic transformations.
 A Holant problem defined by a set of constraint functions $\mathcal{F}$ is solvable in polynomial time if it satisfies this tractability criterion,
 and is \#P-hard otherwise.
 The tractability criterion can be intuitively stated as follows:
 A set $\mathcal{F}$ is tractable if
 (1) every function in $\mathcal{F}$ has arity at most two, or
 (2) $\mathcal{F}$ is transformable to an affine type, or
 (3) $\mathcal{F}$ is transformable to a product type, or
 (4) $\mathcal{F}$ is vanishing, combined with the right type of binary functions, or
 (5) $\mathcal{F}$ belongs to a special category of vanishing type Fibonacci gates.
 The proof of this theorem utilizes many previous dichotomy theorems on Holant problems and Boolean \#CSP.
 Holographic transformations play an indispensable role as both a proof technique and in the statement of the tractability criterion.
\end{abstract}

\begin{keywords}
  Computational complexity, \#P, Counting problems, Dichotomy theorem, Holographic algorithm
\end{keywords}

\begin{AMS}
  68Q25, 68Q17
\end{AMS}

\section{Introduction} \label{sec:intro}

In the study of counting problems,
several interesting frameworks of increasing generality have been proposed.
One is called $H$-coloring or Graph Homomorphism~\cite{Lov67, HN90, DG00, BG05, DGP07, CC10, GGJT10, CCL13}.
Another is called Constraint Satisfaction Problems (\#CSP)~\cite{BD07, Bul13, BG05, DGJ09, BDGJR09, CLX09a, CCL11, CHL12, DR10, GHLX11, CK12, CC12}.
Recently,
inspired by Valiant's holographic algorithms~\cite{Val08, Val06},
a further refined framework called Holant problems~\cite{CLX13a, CLX12, CLX09a, CLX11d} was proposed.
They all describe classes of counting problems that can be expressed as a sum-of-product computation,
specified by a set of local constraint functions $\mathcal{F}$,
also called signatures.
They differ mainly in what $\mathcal{F}$ can be and what is assumed to be present in $\mathcal{F}$ by default.
Such frameworks are interesting because the language is \emph{expressive} enough so that they contain many natural counting problems,
while \emph{specific} enough so that it is possible to prove \emph{dichotomy theorems}.
Such theorems completely classify every problem in a class to be either in P or $\numP$-hard~\cite{Sch78, CH96, FV98, CKS01}.

The goal is to understand which counting problems are computable in polynomial time (called tractable) and which are not (called intractable).
We aim for a characterization in terms of $\mathcal{F}$.
An ideal outcome is to classify,
within a broad class of functions,
\emph{every} function set $\mathcal{F}$ according to whether it defines a tractable counting problem or a $\numP$-hard one.
We note that,
by an analogue of Ladner's theorem~\cite{Lad75},
such a dichotomy is \emph{false} for the whole of $\numP$,
unless P = \#P.

We give a brief description of the Holant framework here~\cite{CLX13a, CLX12, CLX09a, CLX11d}.
A \emph{signature grid} $\Omega = (G, \mathcal{F}, \pi)$ is a tuple,
where $G = (V,E)$ is a graph,
$\pi$ labels each $v \in V$ with a function $f_v \in \mathcal{F}$,
and $f_v$ maps $\{0,1\}^{\deg(v)}$ to $\mathbb{C}$.
We consider all 0-1 edge assignments.
An assignment $\sigma$ for every $e \in E$ gives an evaluation $\prod_{v \in V} f_v(\sigma \mid_{E(v)})$,
where $E(v)$ denotes the incident edges of $v$ and $\sigma \mid_{E(v)}$ denotes the restriction of $\sigma$ to $E(v)$.
The counting problem on the instance $\Omega$ is to compute
\begin{equation} \label{eqn:holant-sum}
 \Holant_\Omega = \sum_{\sigma : E \to \{0,1\}} \prod_{v \in V} f_v\left(\sigma \mid_{E(v)}\right).
\end{equation}
For example,
consider the problem of counting \textsc{Perfect Matching} on $G$.
This problem corresponds to attaching the \textsc{Exact-One} function at every vertex of $G$.

The Holant framework can be defined for general domain $[q]$;
in this paper we restrict to the Boolean case $q = 2$.
The \#CSP problems are the special case of Holant problems where all \textsc{Equality} functions
(with any number of inputs) are assumed to be included in $\mathcal{F}$.
Graph Homomorphism is the further special case of \#CSP where $\mathcal{F}$ consists of a single binary function (in addition to all \textsc{Equality} functions).
Similar or essentially the same notions as Holant have been studied as tensor networks~\cite{Jos95, MS08} in physics as well
as Forney graphs and sum-product algorithms of factor graphs~\cite{For01, Loe04}
in artificial intelligence, coding theory, and signal processing.

Consider the following constraint function $f : \{0, 1\}^4 \to \mathbb{C}$.
Let the input $(x_1, x_2, x_3, x_4)$ have Hamming weight $w$,
then $f(x_1, x_2, x_3, x_4) = 3, 0, 1, 0, 3$,
if $w = 0, 1, 2, 3, 4$,
respectively.
We denote this function by $f = [3,0,1,0,3]$.
What is the counting problem defined by the Holant sum in~(\ref{eqn:holant-sum}) on 4-regular graphs $G$ when $\mathcal{F} = \{f\}$?
By definition,
this is a sum over all 0-1 edge assignments of products of local evaluations.
We only sum over assignments which assign an even number of 1's to the incident edges of each vertex,
since $f=0$ for $w=1$ and $3$.
Then each vertex contributes a factor~3 if the~4 incident edges are assigned all~0 or all~1,
and contributes a factor~1 if exactly two incident edges are assigned~1.
\emph{Before anyone thinks that this problem is artificial},
let us consider a holographic transformation.
Consider the edge-vertex incidence graph $H = (E(G), V(G), \{(e,v) \mid \text{$v$ is incident to $e$ in $G$} \})$ of $G$.
This Holant problem can be expressed in the bipartite form $\holant{{=}_2}{f}$ on $H$,
where $=_2$ is the binary \textsc{Equality} function.
Thus,
every $e \in E(G)$ is assigned $=_2$,
and every $v \in V(G)$ is assigned $f$.
We can write $=_2$ by its truth table $(1,0,0,1)$ indexed by $\{0,1\}^2$.
If we apply the holographic transformation $Z = \tfrac{1}{\sqrt{2}} \left[\begin{smallmatrix} 1 & 1 \\ i & -i \end{smallmatrix}\right]$,
then Valiant's Holant Theorem~\cite{Val08} tells us that $\holant{{=}_2}{f}$ is exactly the same as $\holant{(=_2) Z^{\otimes 2}}{(Z^{-1})^{\otimes 4} f}$.
Here $(=_2) Z^{\otimes 2}$ is a row vector indexed by $\{0,1\}^2$ denoting the transformed function under $Z$ from $(=_2) = (1,0,0,1)$,
and $(Z^{-1})^{\otimes 4} f$ is the column vector indexed by $\{0,1\}^4$ denoting the transformed function under $Z^{-1}$ from $f$.
Let $\hat{f}$ be the \textsc{Exact-Two} function on $\{0,1\}^4$.
We can write its truth table as a column vector indexed by $\{0,1\}^4$,
which has a value~$1$ at entries of Hamming weight~$2$ and~$0$ elsewhere.
In symmetric signature notation,
$\hat{f} = [0,0,1,0,0]$.
Then we have
\begin{align*}
 &Z^{\otimes 4} \hat{f}
 =
 Z^{\otimes 4} \{
          \left[\begin{smallmatrix} 1 \\ 0 \end{smallmatrix}\right]
 \otimes  \left[\begin{smallmatrix} 1 \\ 0 \end{smallmatrix}\right]
 \otimes  \left[\begin{smallmatrix} 0 \\ 1 \end{smallmatrix}\right]
 \otimes  \left[\begin{smallmatrix} 0 \\ 1 \end{smallmatrix}\right]
 +
          \left[\begin{smallmatrix} 1 \\ 0 \end{smallmatrix}\right]
 \otimes  \left[\begin{smallmatrix} 0 \\ 1 \end{smallmatrix}\right]
 \otimes  \left[\begin{smallmatrix} 1 \\ 0 \end{smallmatrix}\right]
 \otimes  \left[\begin{smallmatrix} 0 \\ 1 \end{smallmatrix}\right]
 +
          \left[\begin{smallmatrix} 1 \\ 0 \end{smallmatrix}\right]
 \otimes  \left[\begin{smallmatrix} 0 \\ 1 \end{smallmatrix}\right]
 \otimes  \left[\begin{smallmatrix} 0 \\ 1 \end{smallmatrix}\right]
 \otimes  \left[\begin{smallmatrix} 1 \\ 0 \end{smallmatrix}\right]\\
 &\hphantom{{}={} Z^{\otimes 4} \{}\mathllap{{}+{}}
         {\left[\begin{smallmatrix} 0 \\ 1 \end{smallmatrix}\right]
 \otimes  \left[\begin{smallmatrix} 1 \\ 0 \end{smallmatrix}\right]
 \otimes  \left[\begin{smallmatrix} 1 \\ 0 \end{smallmatrix}\right]
 \otimes  \left[\begin{smallmatrix} 0 \\ 1 \end{smallmatrix}\right]
 +
          \left[\begin{smallmatrix} 0 \\ 1 \end{smallmatrix}\right]
 \otimes  \left[\begin{smallmatrix} 1 \\ 0 \end{smallmatrix}\right]
 \otimes  \left[\begin{smallmatrix} 0 \\ 1 \end{smallmatrix}\right]
 \otimes  \left[\begin{smallmatrix} 1 \\ 0 \end{smallmatrix}\right]
 +
          \left[\begin{smallmatrix} 0 \\ 1 \end{smallmatrix}\right]
 \otimes  \left[\begin{smallmatrix} 0 \\ 1 \end{smallmatrix}\right]
 \otimes  \left[\begin{smallmatrix} 1 \\ 0 \end{smallmatrix}\right]
 \otimes  \left[\begin{smallmatrix} 1 \\ 0 \end{smallmatrix}\right]} \}\\
 &=
 \tfrac{1}{4} \{
          \left[\begin{smallmatrix} 1 \\  i \end{smallmatrix}\right]
 \otimes  \left[\begin{smallmatrix} 1 \\  i \end{smallmatrix}\right]
 \otimes  \left[\begin{smallmatrix} 1 \\ -i \end{smallmatrix}\right]
 \otimes  \left[\begin{smallmatrix} 1 \\ -i \end{smallmatrix}\right]
 +
          \left[\begin{smallmatrix} 1 \\  i \end{smallmatrix}\right]
 \otimes  \left[\begin{smallmatrix} 1 \\ -i \end{smallmatrix}\right]
 \otimes  \left[\begin{smallmatrix} 1 \\  i \end{smallmatrix}\right]
 \otimes  \left[\begin{smallmatrix} 1 \\ -i \end{smallmatrix}\right]
 +
          \left[\begin{smallmatrix} 1 \\  i \end{smallmatrix}\right]
 \otimes  \left[\begin{smallmatrix} 1 \\ -i \end{smallmatrix}\right]
 \otimes  \left[\begin{smallmatrix} 1 \\ -i \end{smallmatrix}\right]
 \otimes  \left[\begin{smallmatrix} 1 \\  i \end{smallmatrix}\right]\\
 &\hphantom{{}={} \tfrac{1}{4} \{}\mathllap{{}+{}}
         {\left[\begin{smallmatrix} 1 \\ -i \end{smallmatrix}\right]
 \otimes  \left[\begin{smallmatrix} 1 \\  i \end{smallmatrix}\right]
 \otimes  \left[\begin{smallmatrix} 1 \\  i \end{smallmatrix}\right]
 \otimes  \left[\begin{smallmatrix} 1 \\ -i \end{smallmatrix}\right]
 +
          \left[\begin{smallmatrix} 1 \\ -i \end{smallmatrix}\right]
 \otimes  \left[\begin{smallmatrix} 1 \\  i \end{smallmatrix}\right]
 \otimes  \left[\begin{smallmatrix} 1 \\ -i \end{smallmatrix}\right]
 \otimes  \left[\begin{smallmatrix} 1 \\  i \end{smallmatrix}\right]
 +
          \left[\begin{smallmatrix} 1 \\ -i \end{smallmatrix}\right]
 \otimes  \left[\begin{smallmatrix} 1 \\ -i \end{smallmatrix}\right]
 \otimes  \left[\begin{smallmatrix} 1 \\  i \end{smallmatrix}\right]
 \otimes  \left[\begin{smallmatrix} 1 \\  i \end{smallmatrix}\right]} \}\\
 &= \tfrac{1}{2} [3,0,1,0,3] = \tfrac{1}{2} f;
\end{align*}
hence $(Z^{-1})^{\otimes 4} f = 2 \hat{f}$.
(Here we use the elementary fact that $(A \otimes B)(u \otimes v) = A u \otimes B v$ for tensor products of matrices and vectors.)
Meanwhile,
$Z$ transforms $=_2$ to the binary \textsc{Disequality} function $\neq_2$:
\begin{align*}
 (=_2) Z^{\otimes 2}
 & = \left(\begin{smallmatrix} 1 & 0 & 0 & 1 \end{smallmatrix}\right) Z^{\otimes 2}
 = \left\{ \left(\begin{smallmatrix} 1 & 0 \end{smallmatrix}\right)^{\otimes 2} + \left(\begin{smallmatrix} 0 &  1 \end{smallmatrix}\right)^{\otimes 2} \right\} Z^{\otimes 2}
 = \tfrac{1}{2} \left\{ \left(\begin{smallmatrix} 1 & 1 \end{smallmatrix}\right)^{\otimes 2} + \left(\begin{smallmatrix} i & -i \end{smallmatrix}\right)^{\otimes 2} \right\}\\
 & = [0,1,0]
 = (\neq_2).
\end{align*}
Hence,
up to a global constant factor of $2^n$ on a graph with $n$ vertices,
the Holant problem with $[3,0,1,0,3]$ is exactly the same as $\holant{{\neq}_2}{[0,0,1,0,0]}$.
A moment's reflection shows that this latter problem is counting Eulerian orientations over 4-regular graphs,
an eminently natural problem!
Thus holographic transformations can reveal the fact that completely different-looking problems are really the same problem,
and there is no objective criterion on one problem being more ``natural'' than another.
Hence we would like to classify all Holant problems given by such signatures.

An interesting observation is that $\holant{{\neq}_2}{[0,0,1,0,0]}$ has exactly the same value as $\holant{{\neq}_2}{[a,b,1,0,0]}$ on any signature grid,
for any $a,b \in \mathbb{C}$.
This is because on a bipartite graph,
$\neq_2$ demands that exactly half of the edges are~0 and the other half are~1,
while on the other side,
any use of the value $a$ or $b$ results in strictly less than half of the edges being~1.
This is related to a phenomenon we call \emph{vanishing}.
Vanishing signatures are constraint functions,
that when applied to any signature grid,
produce a zero Holant value.
A simple example is a tensor product of $\begin{pmatrix} 1 & i \end{pmatrix}$,
i.e.,
a constraint function of the form $\begin{pmatrix} 1 & i \end{pmatrix}^{\otimes k}$ on $k$ variables.
This function on a vertex (of degree $k$) can be replaced by $k$ copies of the unary function $\begin{pmatrix} 1 & i \end{pmatrix}$ on $k$ new vertices,
each connected to an incident edge.
Whenever two copies of $\begin{pmatrix} 1 & i \end{pmatrix}$ meet in the evaluation of $\Holant_\Omega$ in~(\ref{eqn:holant-sum}),
they annihilate each other since they give the value $\begin{pmatrix} 1 & i \end{pmatrix} \cdot \begin{pmatrix} 1 & i \end{pmatrix} = 0$.
These ghostly constraint functions are like the elusive dark matter.
They do not actually contribute any value to the Holant sum.
However in order to give a complete dichotomy for Holant problems,
it turns out to be essential that we capture these vanishing signatures.
There is another similarity with dark matter.
Their contribution to the Holant sum is not directly observed.
Yet in terms of the dimension of the algebraic variety they constitute,
they make up the vast majority of the tractable symmetric signatures.
Furthermore,
when combined with others,
they provide a large substrate to produce non-vanishing and tractable signatures.
In \#CSP,
they are invisible due to the presumed inclusion of all the \textsc{Equality} functions;
and they lurk beneath the surface when one only considers real-valued Holant problems.

The existence of vanishing signatures have influenced previous dichotomy results,
although this influence was not fully recognized at the time.
In the dichotomy theorems in~\cite{CLX09a} and~\cite{CHL12},
almost all tractable signatures can be transformed into a tractable \#CSP problem,
except for one special category.
The tractability proof for this category used the fact that they are a special case of generalized Fibonacci signatures~\cite{CLX13a}.
However,
what went completely unnoticed is that for every input instance using such signatures alone,
the Holant value is always zero!

The most significant previous encounter with vanishing signatures was in the parity setting~\cite{GLV13}.
The authors noticed that a large fraction of signatures always induce an even Holant value,
which is vanishing in $\mathbb{Z}_2$.
However,
the parity dichotomy was achieved using an existential argument without obtaining a complete characterization of the vanishing signatures.
Consequently,
the dichotomy criterion is non-constructive and is currently not known to be decidable.
Nevertheless,
this work is important because it was the first to discover nontrivial vanishing signatures in the parity setting
and to obtain a dichotomy that was \emph{completed} by vanishing signatures.

To complement our characterization of vanishing signatures,
we also obtain a characterization of signatures \emph{transformable} to the \#CSP tractable \emph{Affine} type $\mathscr{A}$ or \emph{Product} type $\mathscr{P}$,
after an orthogonal holographic transformation.
An orthogonal transformation is natural since the binary \textsc{Equality} $=_2$ is unchanged under such holographic transformations.
With explicit characterizations of these tractable signatures,
a complete dichotomy theorem becomes possible.

We first prove a dichotomy for a single signature,
and then we extend it to an arbitrary set of signatures.
The most difficult part is to prove a dichotomy for a single signature of arity~4.
The proof involves a demanding interpolation step and an approximation argument,
both of which use asymmetric signatures.
We found that in order to prove a dichotomy for symmetric signatures,
we must go through asymmetric signatures.

With this dichotomy,
we come to a conclusion on a long series of dichotomies on Holant problems~\cite{CLX12, CLX09a, CLX11b, Kow09, KC10, CK13, CK12, CHL12, HL12},
including the dichotomy theorems for the $\Holant^c$ and $\Holant^*$ frameworks with symmetric signatures.
They all become special cases of this dichotomy.
However,
the proof of this theorem is logically dependent on some of these previous dichotomies.
In particular,
this dichotomy extends the dichotomy in~\cite{HL12} that covers all real-valued symmetric signatures.
While we do not rely on their real-valued dichotomy itself,
we do make important use of two results in~\cite{HL12}.
One is the \#P-hardness of counting Eulerian orientations over $4$-regular graphs;
the other is a dichotomy for $\CSP^d$,
where every variable appears a multiple of $d$ times.

\paragraph{Acknowledgements}

We benefited greatly from the comments and suggestions of the anonymous referees, to whom we are grateful.
We thank Avi Wigderson for the invitation to present this work at the IAS,
and to Peter B\"{u}rgisser, Leslie Ann Goldberg, Mark Jerrum, and Pascal Koiran
for the invitation to present this work at the Dagstuhl seminar on computational counting.
We also thank all those at the Dagstuhl seminar for their interest.
We especially thank Mingji Xia and Les Valiant for their insightful comments.

\section{Preliminaries}

\subsection{Problems and Definitions}

The framework of Holant problems is defined for functions mapping any $[q]^k \to \mathbb{F}$ for a finite $q$ and some field $\mathbb{F}$.
In this paper,
we investigate complex-weighted Boolean $\Holant$ problems,
that is,
all functions are $[2]^k \to \mathbb{C}$.
Strictly speaking,
for consideration of models of computation,
functions take complex algebraic numbers.

A \emph{signature grid} $\Omega = (G, \mathcal{F}, \pi)$ consists of a graph $G = (V,E)$,
where each vertex is labeled by a function $f_v \in \mathcal{F}$, and $\pi : V \to \mathcal{F}$ is the labeling.
The Holant problem on instance $\Omega$ is to evaluate $\Holant_\Omega = \sum_{\sigma} \prod_{v \in V} f_v(\sigma \mid_{E(v)})$,
a sum over all edge assignments $\sigma: E \to \{0,1\}$.

A function $f_v$ can be represented by listing its values in lexicographical order as in a truth table,
which is a vector in $\mathbb{C}^{2^{\deg(v)}}$,
or as a tensor in $(\mathbb{C}^{2})^{\otimes \deg(v)}$.
We also use $f^\alpha$ to denote the value $f(\alpha)$,
where $\alpha$ is a binary string.
A function $f \in \mathcal{F}$ is also called a \emph{signature}.
A symmetric signature $f$ on $k$ Boolean variables can be expressed as $[f_0,f_1,\dotsc,f_k]$,
where $f_w$ is the value of $f$ on inputs of Hamming weight $w$.
In this paper, we consider symmetric signatures.
Sometimes we represent a signature of arity $k$ by a labeled vertex with $k$ ordered dangling edges corresponding to its input.

A Holant problem is parametrized by a set of signatures.

\begin{definition}
 Given a set of signatures $\mathcal{F}$,
 we define the counting problem $\Holant$ $(\mathcal{F})$ as:

 Input: A \emph{signature grid} $\Omega = (G, \mathcal{F}, \pi)$;

 Output: $\Holant_\Omega$.
\end{definition}

The following family Holant$^*$ of Holant problems were investigated previously~\cite{CLX09a, CLX10}.
This is the class of Holant problems in which all unary signatures are freely available.

\begin{definition}
 Given a set of signatures $\mathcal{F}$,
 $\Holant^*(\mathcal{F})$ denotes $\Holant(\mathcal{F} \cup \mathcal{U})$,
 where $\mathcal{U}$ is the set of all unary signatures.
\end{definition}

The family $\Holant^c$ of Holant problems (on Boolean variables) are defined analogously.
The $c$ stands for \emph{constants} and refers to the signatures that can fix a variable to a constant of the domain.

\begin{definition}
 Given a set of signatures $\mathcal{F}$,
 $\Holant^c(\mathcal{F})$ denotes $\Holant(\mathcal{F} \cup \{[0,1],[1,0]\})$.
\end{definition}

A signature $f$ of arity $n$ is \emph{degenerate} if there exist unary signatures $u_j \in \mathbb{C}^2$ ($1 \le j \le n$)
such that $f = u_1 \otimes \cdots \otimes u_n$.
A symmetric degenerate signature has the form $u^{\otimes n}$.
For such signatures,
it is equivalent to replace it by $n$ copies of the corresponding unary signature.
Replacing a signature $f \in \mathcal{F}$ by a constant multiple $c f$,
where $c \ne 0$,
does not change the complexity of $\Holant(\mathcal{F})$.
It introduces a global nonzero factor to $\Holant_\Omega$.
Hence, for two signatures $f,g$ of the same arity,
we use $f \neq g$ to mean that these signatures are not equal in the projective space sense,
i.e.~not equal up to any nonzero constant multiple.

We say a signature set $\mathcal{F}$ is tractable (resp.~$\SHARPP$-hard)
if the corresponding counting problem $\Holant(\mathcal{F})$ is tractable (resp.~$\SHARPP$-hard).
Similarly for a signature $f$, we say $f$ is tractable (resp.~$\SHARPP$-hard) if $\{f\}$ is.
We follow the usual conventions about polynomial time Turing reduction $\le_T$ and polynomial time Turing equivalence $\equiv_T$.

\subsection{Holographic Reduction}

To introduce the idea of holographic reductions,
it is convenient to consider bipartite graphs.
For a general graph,
we can always transform it into a bipartite graph while preserving the Holant value,
as follows.
For each edge in the graph,
we replace it by a path of length two.
(This operation is called the \emph{2-stretch} of the graph and yields the edge-vertex incidence graph.)
Each new vertex is assigned the binary \textsc{Equality} signature $(=_2) = [1,0,1]$.

We use $\holant{\mathcal{R}}{\mathcal{G}}$ to denote the Holant problem over bipartite graphs $H = (U,V,E)$,
where each vertex in $U$ or $V$ is assigned a signature in $\mathcal{R}$ or $\mathcal{G}$,
respectively.
An input instance for this bipartite Holant problem is a bipartite signature grid and is denoted by $\Omega = (H;\ \mathcal{R} \mid \mathcal{G};\ \pi)$.
Signatures in $\mathcal{R}$ are considered as row vectors (or covariant tensors);
signatures in $\mathcal{G}$ are considered as column vectors (or contravariant tensors)~\cite{DP91}.

For a 2-by-2 matrix $T$ and a signature set $\mathcal{F}$,
define $T \mathcal{F} = \{g \mid \exists f \in \mathcal{F}$ of arity $n,~g = T^{\otimes n} f\}$, similarly for $\mathcal{F} T$.
Whenever we write $T^{\otimes n} f$ or $T \mathcal{F}$,
we view the signatures as column vectors;
similarly for $f T^{\otimes n} $ or $\mathcal{F} T$ as row vectors.

Let $T$ be an invertible 2-by-2 matrix.
The holographic transformation defined by $T$ is the following operation:
given a signature grid $\Omega = (H;\ \mathcal{R} \mid \mathcal{G};\ \pi)$,
for the same bipartite graph $H$,
we get a new grid $\Omega' = (H;\ \mathcal{R} T \mid T^{-1} \mathcal{G};\ \pi')$ by replacing each signature in
$\mathcal{R}$ or $\mathcal{G}$ with the corresponding signature in $\mathcal{R} T$ or $T^{-1} \mathcal{G}$.

\begin{theorem}[Valiant's Holant Theorem~\cite{Val08}]
 If there is a holographic transformation mapping signature grid $\Omega$ to $\Omega'$,
 then $\Holant_\Omega = \Holant_{\Omega'}$.
\end{theorem}

Therefore,
an invertible holographic transformation does not change the complexity of the Holant problem in the bipartite setting.
Furthermore,
there is a special kind of holographic transformation,
the orthogonal transformation,
that preserves the binary equality and thus can be used freely in the standard setting.

\begin{theorem}[Theorem~2.2 in~\cite{CLX09a}] \label{thm:orthogonal}
 Suppose $T$ is a 2-by-2 orthogonal matrix $(T T^\intercal = I_2)$ and let $\Omega = (H, \mathcal{F}, \pi)$ be a signature grid.
 Under a holographic transformation by $T$, we get a new grid $\Omega' = (H, T \mathcal F, \pi')$ and $\Holant_\Omega = \Holant_{\Omega'}$.
\end{theorem}

Since the complexity of a signature is equivalent up to a nonzero constant factor,
we also call a transformation $T$ such that $T T^\intercal = \lambda I$ for some $\lambda \neq 0$ an orthogonal transformation.
Such transformations do not change the complexity of a problem.

\subsection{Realization}

One basic notion used throughout the paper is realization.
We say a signature $f$ is \emph{realizable} or \emph{constructible} from a signature set $\mathcal{F}$
if there is a gadget with some dangling edges such that each vertex is assigned a signature from $\mathcal{F}$,
and the resulting graph, when viewed as a black-box signature with inputs on the dangling edges, is exactly $f$.
We will only construct polynomial-sized gadget in this paper.
Hence if $f$ is realizable from a set $\mathcal{F}$,
then we can freely add $f$ into $\mathcal{F}$ while preserving the complexity.

Formally,
such a notion is defined by an $\mathcal{F}$-gate~\cite{CLX09a, CLX10}.
An $\mathcal{F}$-gate is similar to a signature grid $(H, \mathcal{F}, \pi)$ except that $H = (V,E,D)$ is a graph with some dangling edges $D$.
The dangling edges define external variables for the $\mathcal{F}$-gate.
(See Figure~\ref{fig:Fgate} for an example.)
We denote the regular edges in $E$ by $1, 2, \dotsc, m$ and the dangling edges in $D$ by $m+1, \dotsc, m+n$.
Then we can define a function $\Gamma$ for this $\mathcal{F}$-gate as
\[
 \Gamma(y_1, \dotsc, y_n) = \sum_{x_1, \dotsc, x_m \in \{0, 1\}} H(x_1, \dotsc, x_m, y_1, \dotsc, y_n),
\]
where $(y_1, \dotsc, y_n) \in \{0, 1\}^n$ is an assignment on the dangling edges
and $H(x_1, \dotsc, x_m,$ $y_1, \dotsc, y_n)$ is the value of the signature grid on an assignment of all edges,
which is the product of evaluations at all internal vertices.
We also call this function $\Gamma$ the signature of the $\mathcal{F}$-gate.
An $\mathcal{F}$-gate can be used in a signature grid as if it is just a single vertex with the particular signature.

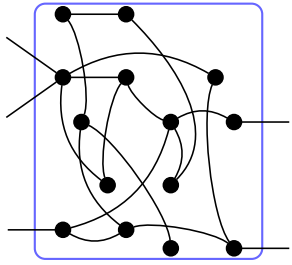
\begin{figure}[t]
 \centering
 \begin{tikzpicture}[scale=\scale,transform shape,node distance=\nodeDist,semithick]
  \node[external]  (0)                     {};
  \node[internal]  (1) [below right of=0]  {};
  \node[external]  (2) [below left  of=1]  {};
  \node[internal]  (3) [above       of=1]  {};
  \node[internal]  (4) [right       of=3]  {};
  \node[internal]  (5) [below       of=4]  {};
  \node[internal]  (6) [below right of=5]  {};
  \node[internal]  (7) [right       of=6]  {};
  \node[internal]  (8) [below       of=6]  {};
  \node[internal]  (9) [below       of=8]  {};
  \node[internal] (10) [right       of=9]  {};
  \node[internal] (11) [above right of=6]  {};
  \node[internal] (12) [below left  of=8]  {};
  \node[internal] (13) [left        of=8]  {};
  \node[internal] (14) [below left  of=13] {};
  \node[external] (15) [left        of=14] {};
  \node[internal] (16) [below left  of=5]  {};
  \path let
         \p1 = (15),
         \p2 = (0)
        in
         node[external] (17) at (\x1, \y2) {};
  \path let
         \p1 = (15),
         \p2 = (2)
        in
         node[external] (18) at (\x1, \y2) {};
  \node[external] (19) [right of=7]  {};
  \node[external] (20) [right of=10] {};
  \path (1) edge                             (5)
            edge[bend left]                 (11)
            edge[bend right]                (13)
            edge node[near start] (e1) {}   (17)
            edge node[near start] (e2) {}   (18)
        (3) edge                             (4)
        (4) edge[out=-45,in=45]              (8)
        (5) edge[bend right, looseness=0.5] (13)
            edge[bend right, looseness=0.5]  (6)
        (6) edge[bend left]                  (8)
            edge[bend left]                  (7)
            edge[bend left]                 (14)
        (7) edge node[near start] (e3) {}   (19)
       (10) edge[bend right, looseness=0.5] (12)
            edge[bend left,  looseness=0.5] (11)
            edge node[near start] (e4) {}   (20)
       (12) edge[bend left]                 (16)
       (14) edge node[near start] (e5) {}   (15)
            edge[bend right]                (12)
       (16) edge[bend left,  looseness=0.5]  (9)
            edge[bend right, looseness=0.5]  (3);
  \begin{pgfonlayer}{background}
   \node[inner sep=1pt,transform shape=false,draw=\borderColor,thick,rounded corners,fit = (3) (4) (9) (e1) (e2) (e3) (e4) (e5)] {};
  \end{pgfonlayer}
 \end{tikzpicture}
 \caption{An $\mathcal{F}$-gate with 5 dangling edges.}
 \label{fig:Fgate}
\end{figure}

Using the idea of $\mathcal{F}$-gates,
we can reduce one Holant problem to another.
Suppose $g$ is the signature of some $\mathcal{F}$-gate.
Then $\Holant(\mathcal{F} \cup \{g\}) \leq_T \Holant(\mathcal{F})$.
The reduction is simple.
Given an instance of $\Holant(\mathcal{F} \cup \{g\})$,
by replacing every appearance of $g$ by the $\mathcal{F}$-gate,
we get an instance of $\Holant(\mathcal{F})$.
Since the signature of the $\mathcal{F}$-gate is $g$,
the Holant values for these two signature grids are identical.

Although our main result is about symmetric signatures,
some of our proofs utilize asymmetric signatures.
When a gadget has an asymmetric signature,
we place a diamond on the edge corresponding to the most significant index bit.
The remaining index bits are in order of decreasing significance as one travels counterclockwise around the vertex.
(See Figure~\ref{fig:rotate_asymmetric_signature} for an example.)
Some of our gadget constructions are bipartite graphs.
To highlight this structure, we use vertices of different shapes.
Any time a gadget has a square vertex, it is assigned $[0,1,0]$.
(See Figure~\ref{fig:gadget:special_vanishing_case:glue} for an example.)

We note that even for a very simple signature set $\mathcal{F}$,
the signatures for all $\mathcal{F}$-gates can be quite complicated and expressive.

\subsection{\texorpdfstring{\#}{Count-}CSP and Its Tractable Signatures} \label{subsec:CSP-Tractable}

An instance of $\CSP(\mathcal{F})$ has the following bipartite view.
Create a node for each variable and each constraint.
Connect a variable node to a constraint node if the variable appears in the constraint function.
This bipartite graph is also known as the \emph{constraint graph}.
Under this view,
we can see that
\[
 \CSP(\mathcal{F}) \equiv_T \holant{\mathcal{F}}{\mathcal{EQ}} \equiv_T \Holant(\mathcal{F} \cup \mathcal{EQ}),
\]
where $\mathcal{EQ} = \{=_1, =_2, =_3, \dotsc\}$ is the set of equality signatures of all arities.

For a positive integer $d$,
the problem $\CSP^d(\mathcal{F})$ is similar to $\CSP(\mathcal{F})$ except that every variable has to appear a multiple of $d$ times.
Therefore, we have
\[
 \CSP^d(\mathcal{F}) \equiv_T \holant{\mathcal{F}}{\mathcal{EQ}_d},
\]
where $\mathcal{EQ}_d = \{=_d, =_{2 d}, =_{3 d}, \dotsc\}$ is the set of equality signatures of arities that are a multiple of $d$.

For the \#CSP framework,
the following two signature sets are tractable~\cite{CLX09a}.

\begin{definition}
 A $k$-ary function $f(x_1, \dotsc, x_k)$ is \emph{affine} if it has the form
 \[\lambda \chi_{Ax = 0} \cdot \sqrt{-1}^{\sum_{j=1}^n \langle \alpha_j, x \rangle},\]
 where $\lambda \in \mathbb{C}$,
 $x = (x_1, x_2, \dotsc, x_k, 1)^\intercal$,
 $A$ is a matrix over $\mathbb{F}_2$, $\alpha_j$ is a vector over $\mathbb{F}_2$,
 and $\chi$ is a 0-1 indicator function such that $\chi_{Ax = 0}$ is 1 iff $A x = 0$.
 Note that the dot product $\langle \alpha_j, x \rangle$ is calculated over $\mathbb{F}_2$,
 while the summation $\sum_{j=1}^n$ on the exponent of $i = \sqrt{-1}$ is evaluated as a sum mod 4 of 0-1 terms.
 We use $\mathscr{A}$ to denote the set of all affine functions.
\end{definition}

Notice that there is no restriction on the number of rows in the matrix $A$.
The trivial case is when $A$ is the zero matrix so that $\chi_{A x = 0} = 1$ holds for all $x$.

\begin{definition}
 A function is of \emph{product type} if it can be expressed as a product of unary functions,
 binary equality functions $([1,0,1])$, and binary disequality functions $([0,1,0])$.
 We use $\mathscr{P}$ to denote the set of product-type functions.
\end{definition}

An alternate definition for $\mathscr{P}$,
implicit in~\cite{CLX11a},
is the tensor closure of signatures with support on two entries of complement indices.

It is easy to see
(cf.~Lemma~2.2 in~\cite{HL13},
the full version of~\cite{HL12})
that if $f$ is a symmetric signature in $\mathscr{P}$,
then $f$ is either degenerate, binary disequality, or generalized equality (i.e.~$[a,0,\dotsc,0,b]$ for $a, b \in \mathbb{C}$).
It is known that the set of non-degenerate symmetric signatures in $\mathscr{A}$ is precisely the nonzero signatures
($\lambda \neq 0$) in $\mathscr{F}_1 \union \mathscr{F}_2 \union \mathscr{F}_3$ with arity at least two,
where $\mathscr{F}_1$, $\mathscr{F}_2$, and $\mathscr{F}_3$ are three families of signatures defined as
\begin{align*}
 \mathscr{F}_1 &= \left\{\lambda \left([1,0]^{\otimes k} + i^r [0, 1]^{\otimes k}\right) \st \lambda \in \mathbb{C}, k = 1, 2, \dotsc, r = 0, 1, 2, 3\right\},\\
 \mathscr{F}_2 &= \left\{\lambda \left([1,1]^{\otimes k} + i^r [1,-1]^{\otimes k}\right) \st \lambda \in \mathbb{C}, k = 1, 2, \dotsc, r = 0, 1, 2, 3\right\}, \text{ and}\\
 \mathscr{F}_3 &= \left\{\lambda \left([1,i]^{\otimes k} + i^r [1,-i]^{\otimes k}\right) \st \lambda \in \mathbb{C}, k = 1, 2, \dotsc, r = 0, 1, 2, 3\right\}.
\end{align*}
Let $\mathscr{F}_{123} = \mathscr{F}_1 \union \mathscr{F}_2 \union \mathscr{F}_3$ be the union of these three sets of signatures.
We explicitly list all the signatures in $\mathscr{F}_{123}$ up to an arbitrary constant multiple from $\mathbb{C}$:\\
\parbox{0.61\textwidth}{
 \begin{enumerate}
  \item $[1, 0, \dotsc, 0, \pm 1]$; \hfill $(\mathscr{F}_1, r=0,2)$
  \item $[1, 0, \dotsc, 0, \pm i]$; \hfill $(\mathscr{F}_1, r=1,3)$
  \item $[1,  0, 1,  0, \dotsc,   0  \text{ or } 1]$; \hfill $(\mathscr{F}_2, r=0)$
  \item $[1, -i, 1, -i, \dotsc, (-i) \text{ or } 1]$; \hfill $(\mathscr{F}_2, r=1)$
  \item $[0,  1, 0,  1, \dotsc,   0  \text{ or } 1]$; \hfill $(\mathscr{F}_2, r=2)$
  \item $[1,  i, 1,  i, \dotsc,   i  \text{ or } 1]$; \hfill $(\mathscr{F}_2, r=3)$
  \item $[1,  0, -1,  0, 1,  0, -1,  0, \dotsc, 0 \text{ or } 1 \text{ or } (-1)]$; \hfill $(\mathscr{F}_3, r=0)$
  \item $[1,  1, -1, -1, 1,  1, -1, -1, \dotsc,               1 \text{ or } (-1)]$; \hfill $(\mathscr{F}_3, r=1)$
  \item $[0,  1,  0, -1, 0,  1,  0, -1, \dotsc, 0 \text{ or } 1 \text{ or } (-1)]$; \hfill $(\mathscr{F}_3, r=2)$
  \item $[1, -1, -1,  1, 1, -1, -1,  1, \dotsc,               1 \text{ or } (-1)]$. \hfill $(\mathscr{F}_3, r=3)$
 \end{enumerate}}

In the Holant framework, there are two corresponding signature sets that are tractable.
A signature $f$ (resp.~a signature set $\mathcal{F}$) is $\mathscr{A}$-transformable if there exists a holographic transformation $T$
such that $f \in T \mathscr{A}$ (resp.~$\mathcal{F} \subseteq T \mathscr{A}$) and $[1,0,1] T^{\otimes 2} \in \mathscr{A}$.
Similarly, a signature $f$ (resp.~a signature set $\mathcal{F}$) is $\mathscr{P}$-transformable if there exists a holographic transformation $T$
such that $f \in T \mathscr{P}$ (resp.~$\mathcal{F}\subseteq T \mathscr{P}$) and $[1,0,1] T^{\otimes 2} \in \mathscr{P}$.
These two families are tractable because after a transformation by $T$, it is a tractable \#CSP instance.

\subsection{Some Known Dichotomies}

Here we list several known dichotomies.
Our main dichotomy theorem generalizes all of them.
In order to clearly see this, we state the previous dichotomies using the language of this paper.
In particular, some previous classifications are now presented differently using our new understanding.

The dichotomy for a single symmetric ternary signature is an important base case of our proof.

\begin{theorem}[Theorem~3 in~\cite{CHL12}] \label{thm:arity3:singleton}
 If $f = [f_0,f_1,f_2,f_3]$ is a non-degenerate, complex-valued signature,
 then $\Holant(f)$ is $\SHARPP$-hard unless $f$ satisfies one of the following conditions,
 in which case the problem is computable in polynomial time:
 \begin{enumerate}
  \item $f$ is $\mathscr{A}$- or $\mathscr{P}$-transformable;
  \item For $\alpha \in \{2 i, -2 i\}$, $f_2= \alpha f_1 + f_0$ and $f_3 = \alpha f_2 + f_1$. \label{case:arity3:exceptional}
 \end{enumerate}
\end{theorem}

We also use the following theorem about edge-weighted signatures on $k$-regular graphs.

\begin{theorem}[Theorem~3 in~\cite{CK12}] \label{thm:k-reg_homomorphism}
 Let $k \ge 3$ be an integer and suppose $f$ is a non-degenerate, symmetric, complex-valued binary signature.
 Then $\holant{f}{{=}_k}$ is $\SHARPP$-hard unless there exists a holographic transformation $T$
 such that $f T^{\otimes 2} = [1,0,1]$ and $\left(({T^{-1}})^{\otimes k} (=_k)\right)$ is $\mathscr{A}$- or $\mathscr{P}$-transformable,
 in which case the problem is computable in polynomial time.
\end{theorem}

While Theorem~\ref{thm:k-reg_homomorphism} is conceptual,
the original statement in Theorem~\ref{thm:k-reg_homomorphism}$'$ is directly applicable.

\newtheorem*{specialtheorem}{Theorem {\thetheorem}$'$}

\begin{specialtheorem}[Theorem~3 in~\cite{CK12}]
 Let $k \ge 3$ be an integer.
 Then $\holant{[f_0, f_1, f_2]}{({=}_k)}$ is $\SHARPP$-hard unless one of the following conditions hold,
 in which case the problem is computable in polynomial time:
 \begin{enumerate}
  \item $f_0 f_2 = f_1^2$;
  \item $f_0 = f_2 = 0$;
  \item $f_1 = 0$;
  \item $f_0 f_2 = -f_1^2$ and $f_0^{2 k} = f_2^{2 k}$.
 \end{enumerate}
\end{specialtheorem}

The next theorem is a generalization of the Boolean \#CSP dichotomy (where $d=1$).
\begin{theorem}[Theorem~IV.1 in~\cite{HL12}] \label{thm:CSPd}
 Let $\mathcal{T}_d = \left\{\left[\begin{smallmatrix} 1 & 0 \\ 0 & \omega \end{smallmatrix}\right] \in \mathbb{C}^{2 \times 2} \st \omega^d = 1\right\}$,
 $d \ge 1$ be an integer, and $\mathcal{F}$ be any set of symmetric, complex-valued signatures in Boolean variables.
 Then $\CSP^d(\mathcal{F})$ is $\SHARPP$-hard unless there exists a $T \in \mathcal{T}_{4 d}$
 such that $T \mathcal{F} \subseteq \mathscr{P}$ or $T \mathcal{F} \subseteq \mathscr{A}$,
 in which case the problem is computable in polynomial time.
\end{theorem}

The following three dichotomies are not directly used in this paper.
We list them for comparison.
First is the real-valued Holant dichotomy.
Our results have no dependence on this dichotomy.

\begin{theorem}[Theorem~III.2 in~\cite{HL12}] \label{thm:real_holant}
 Let $\mathcal{F}$ be any set of symmetric, real-valued signatures in Boolean variables.
 Then $\Holant(\mathcal{F})$ is $\SHARPP$-hard unless $\mathcal{F}$ satisfies one of the following conditions,
 in which case the problem is computable in polynomial time:
 \begin{enumerate}
  \item Any non-degenerate signature in $\mathcal{F}$ is of arity at most~$2$;
  \item $\mathcal{F}$ is $\mathscr{A}$- or $\mathscr{P}$-transformable.
 \end{enumerate}
\end{theorem}

The other two dichotomy theorems are for complex-valued $\Holant^*$ and $\Holant^c$.
We do not directly apply these two theorems,
but our results depend on some intermediate results such as Theorems~\ref{thm:arity3:singleton}, \ref{thm:k-reg_homomorphism}, and~\ref{thm:CSPd}.

\begin{theorem}[Theorem~3.1 in~\cite{CLX09a}] \label{thm:holant_star}
 Let $\mathcal{F}$ be any set of non-degenerate, symmetric, complex-valued signatures in Boolean variables.
 Then $\Holant^*(\mathcal{F})$ is $\SHARPP$-hard unless $\mathcal{F}$ satisfies one of the following conditions,
 in which case the problem is computable in polynomial time:
 \begin{enumerate}
  \item Any signature in $\mathcal{F}$ is of arity at most~$2$;
  \item $\mathcal{F}$ is $\mathscr{P}$-transformable;
  \item There exists $\alpha \in \{2 i, -2 i\}$, such that for any signature $f\in\mathcal{F}$ of arity $n$, for $0 \le k \le n-2$,
  we have $f_{k+2} = \alpha f_{k+1} + f_k$. \label{case:holant_star:exceptional}
 \end{enumerate}
\end{theorem}

\begin{theorem}[Theorem~6 in~\cite{CHL12}] \label{thm:holant_c}
 Let $\mathcal{F}$ be any set of symmetric, complex-valued signatures in Boolean variables.
 Then $\Holant^c(\mathcal{F})$ is $\SHARPP$-hard unless $\mathcal{F}$ satisfies one of the following conditions,
 in which case the problem is computable in polynomial time:
 \begin{enumerate}
  \item Any non-degenerate signature in $\mathcal{F}$ is of arity at most~$2$;
  \item $\mathcal{F}$ is $\mathscr{P}$-transformable;
  \item $\mathcal{F} \cup \{[1,0],[0,1]\}$ is $\mathscr{A}$-transformable;
  \item There exists $\alpha \in \{2 i, -2 i\}$, such that for any non-degenerate signature $f\in\mathcal{F}$ of arity $n$,
  for $0 \le k \le n-2$, we have $f_{k+2} = \alpha f_{k+1} + f_k$. \label{case:holant_c:exceptional}
 \end{enumerate}
\end{theorem}

\def\problemSpace{4pt}

\section{A Sampling of Problems}

We illustrate the scope of our dichotomy theorem by several concrete problems.
Some problems are naturally expressed with real weights,
but they are linked inextricably to other problems that use complex weights.
Sometimes the inherent link between two real-weighted problems is provided by a transformation through $\mathbb{C}$.

\vspace*{\problemSpace}
\textbf{Problem:} \#\textsc{VertexCover}

\textbf{Input:} An undirected graph $G$.

\textbf{Output:} The number of vertex covers in $G$.
\vspace*{\problemSpace}

This classic problem is most naturally expressed as the real-weighted bipartite Holant problem $\holant{[0,1,1]}{\mathcal{EQ}}$.
A vertex assigned an equality signature forces all its incident edges to be assigned the same value;
this is equivalent to these vertices being assigned a value themselves.
The degree two vertices assigned the binary $\textsc{Or} = [0,1,1]$ should be thought of as an edge between its neighboring vertices.
These edge-like vertices force at least one of its neighbors to be selected.
The number of assignments satisfying these requirements is exactly the number of vertex covers.

To apply our dichotomy theorem,
we perform a holographic transformation by $T = \left[\begin{smallmatrix} 0 & -i \\ 1 & i \end{smallmatrix}\right]$.
To understand why we choose this particular $T$,
let us express $[0,1,1]$ as
\begin{align*}
 [0,1,1]
 &= (0\ 1\ 1\ 1)
  = \left\{[1,1]^{\otimes 2} + [i, 0]^{\otimes 2}\right\}
  = \left\{[1,0]^{\otimes 2} + [0, 1]^{\otimes 2}\right\} \begin{bmatrix} 1 & 1 \\ i & 0 \end{bmatrix}^{\otimes 2}\\
 &= (1\ 0\ 0\ 1) (T^{-1})^{\otimes 2}
  = (=_2) (T^{-1})^{\otimes 2}.
\end{align*}
Thus, a holographic transformation by $T$ yields
\begin{align*}
 \holant{[0,1,1]}{\mathcal{EQ}}
 &\equiv_T \holant{[0,1,1] T^{\otimes 2}}{T^{-1} \mathcal{EQ}}\\
 &\equiv_T \holant{{=}_2}{T^{-1} \mathcal{EQ}}\\
 &\equiv_T \Holant(T^{-1} \mathcal{EQ}).
\end{align*}
The equality signature of arity $k$ in $\mathcal{EQ}$,
a column vector denoted by $=_k$, is transformed by $T^{-1}$ to
\begin{align*}
 f_{(k)}
 &= (T^{-1})^{\otimes k} (=_k)
  = \begin{bmatrix} 1 & 1 \\ i & 0 \end{bmatrix}^{\otimes k}
    \left\{\begin{bmatrix}  1 \\ 0 \end{bmatrix}^{\otimes k} + \begin{bmatrix} 0 \\ 1 \end{bmatrix}^{\otimes k}\right\}\\
 &= \begin{bmatrix}  1 \\ i \end{bmatrix}^{\otimes k} +  \begin{bmatrix} 1 \\ 0 \end{bmatrix}^{\otimes k}
  = [2,i,-1,-i,1,i,-1,-i,1,i,\dotsc]
\end{align*}
of length $k+1$.
By our main dichotomy, Theorem~\ref{thm:main}, $\Holant(T^{-1} \mathcal{EQ})$ is $\SHARPP$-hard.
Indeed, even $\Holant(f_{(k)})$, the restriction of this problem to $k$-regular graphs
is $\SHARPP$-hard for $k \ge 3$ by our single signature dichotomy, Theorem~\ref{thm:dic:single}.

\vspace*{\problemSpace}
\textbf{Problem:} \#\textsc{$\lambda$-VertexCover}

\textbf{Input:} An undirected graph $G$.

\textbf{Output:} $\displaystyle \sum_{C \in \mathcal{C}(G)} \lambda^{e(C)},$\\
where $\mathcal{C}(G)$ denotes the set of all vertex covers of $G$,
and $e(C)$ is the number of edges with both endpoints in the vertex cover $C$.
\vspace*{\problemSpace}

Our dichotomy also easily handles this edge-weighted vertex cover problem that is denoted by $\holant{[0,1,\lambda]}{\mathcal{EQ}}$.
Suppose $\lambda \not = 0$.
On regular graphs,
this problem is equivalent to the so-called \emph{hardcore gas model},
which is the vertex-weighted problem denoted by $\holant{[1,1,0]}{\mathcal{F}}$,
where $\mathcal{F}$ consists of signatures of the form $[1,0,\dotsc,0,\mu]$.
By flipping~0 and~1, this is the same as $\holant{[0,1,1]}{\mathcal{F}'}$
with $\mathcal{F}'$ containing $[\mu,0,\dotsc,0,1]$.
For $k$-regular graphs, we consider the diagonal transformation
$T = \left[\begin{smallmatrix} 1 & 0 \\ 0 & \frac{1}{\lambda} \end{smallmatrix}\right]$,
where $\lambda = 1 / \mu^{1/k}$;
\begin{align*}
 \holant{[0,1,\lambda]}{{=}_k}
 &\equiv_T \holant{[0,1,\lambda] T^{\otimes 2}}{(T^{-1})^{\otimes k} (=_k)}\\
 &\equiv_T \holant{\tfrac{1}{\lambda}[0,1,1]}{[1,0,\cdots,0,\lambda^k]}\\
 &\equiv_T \holant{[0,1,1]}{[\mu,0,\cdots,0,1]}.
\end{align*}
This problem, denoted by \#\textsc{$k$-$\lambda$-VertexCover}, is also $\SHARPP$-hard for $k \ge 3$.
To see this,
apply the holographic transformation $T = \left[\begin{smallmatrix} 0 & -i \lambda \\ 1 & i \end{smallmatrix}\right]$ to the edge-weighted form of the problem.
Then $[0,1,\lambda]$ is transformed to $\lambda (=_2)$ and $=_k$ is transformed to $g_{(\lambda,k)} = \frac{1}{\lambda^k} [\lambda^k + 1,i,-1,-i,1,\dotsc]$.
Since $\Holant(g_{(\lambda,k)})$ is $\SHARPP$-hard by Theorem~\ref{thm:dic:single},
we conclude that \#\textsc{$k$-$\lambda$-VertexCover} is also $\SHARPP$-hard.

If $\lambda = 0$, then the above problem is $\holant{[0,1,0]}{\mathcal{EQ}}$, which is tractable.
However, the transformation $T$ above is singular in this case.
We can in fact apply another transformation
$T' = \left[\begin{smallmatrix} 1 -\tfrac{\lambda}{2} & -\left(1+\tfrac{\lambda}{2}\right)i \\ 1 & i \end{smallmatrix}\right]$
such that it transforms the problem $\holant{[0,1,\lambda]}{=_k}$ into $\Holant(h_{(\lambda, k)})$
for some $h_{(\lambda, k)}$ regardless of whether $\lambda = 0$ or not.
Then by applying Theorem~\ref{thm:dic:single},
we reach the same conclusion that \#\textsc{$\lambda$-VertexCover} is $\SHARPP$-hard on $k$-regular graphs when $\lambda \neq 0$.
We note that when $\lambda = 0$,  $T'= \left[\begin{smallmatrix} 1 & -i \\ 1 & i \end{smallmatrix} \right] = \sqrt{2} Z^{-1}$,
where $Z = \frac{1}{\sqrt{2}} \left[ \begin{smallmatrix} 1 & 1 \\  i & -i\end{smallmatrix}\right]$ was used in Section~\ref{sec:intro}.

We now consider some orientation problems.

\vspace*{\problemSpace}
\textbf{Problem:} \#\textsc{NoSinkOrientation}

\textbf{Input:} An undirected graph $G$.

\textbf{Output:} The number of orientations of $G$ such that each vertex has at least one outgoing edge.
\vspace*{\problemSpace}

This problem is denoted by $\holant{[0,1,0]}{\mathcal{F}}$,
where $\mathcal{F}$ consists of $f_{(k)} = [0,1,\dotsc,1,1]$ for any arity $k$.
Each degree two vertex on the left side of the bipartite graph must have its incident edges assigned different values.
We associate an oriented edge between the neighbors of such vertices with the head on the side assigned~$0$ and the tail on the side assigned~$1$.
This problem is $\SHARPP$-hard even over $k$-regular graphs provided $k \ge 3$.
Just as with the bipartite form of the vertex cover problem,
we do a holographic transformation to apply our dichotomy theorem.
This time, we pick $T = \frac{1}{2} \left[\begin{smallmatrix} 1 & -i \\ 1 & i \end{smallmatrix}\right] = \frac{1}{\sqrt{2}} Z^{-1}$,
with $T^{-1} = \sqrt{2} Z = \left[\begin{smallmatrix} 1 & 1 \\ i & -i \end{smallmatrix}\right]$ and get
\begin{align*}
 \holant{[0,1,0]}{f_{(k)}}
 &\equiv_T \holant{[0,1,0] {T}^{\otimes 2}}{({T}^{-1})^{\otimes k} f_{(k)}}\\
 &\equiv_T \holant{\tfrac{1}{2}[1,0,1]}{\hat{f}_{(k)}} \\
 &\equiv_T \Holant(\hat{f}_{(k)}),
\end{align*}
where $\hat{f}_{(k)} = [2^k-1, -i, 1, i, -1, \dotsc]$.
This is actually a special case (consider $-\hat{f}_{(k)}$) of the \#\textsc{$k$-$\lambda$-VertexCover} problem with $\lambda = 2 e^{\pi i / k}$.
Therefore, this problem is $\SHARPP$-hard.
However, if we consider this problem modulo $2^k$,
$\hat{f}_{(k)}$ becomes $[-1, -i, 1, i, -1, \dotsc]$,
and belongs to one of the tractable cases in our dichotomy.
Thus, \#\textsc{NoSinkOrientation} is tractable modulo $2^t$,
where $t$ is the minimal degree of the input graph.

\vspace*{\problemSpace}
\textbf{Problem:} \#\textsc{NoSinkNoSourceOrientation}

\textbf{Input:} An undirected graph $G$.

\textbf{Output:} The number of orientations of $G$ such that each vertex has at least one incoming and one outgoing edge.
\vspace*{\problemSpace}

This problem is denoted by $\holant{[0,1,0]}{\mathcal{F}}$,
where $\mathcal{F}$ consists of $f_{(k)} = [0,1,\dotsc,1,0]$ for any arity $k$.
This problem is also $\SHARPP$-hard on $k$-regular graphs for $k \ge 3$.
We pick the same $T$ as in the previous problem and get
\begin{align*}
 \holant{[0,1,0]}{f_{(k)}}
 &\equiv_T \holant{[0,1,0] {T}^{\otimes 2}}{({T}^{-1})^{\otimes k} f_{(k)}}\\
 &\equiv_T \holant{\tfrac{1}{2}[1,0,1]}{\hat{f}_{(k)}} \\
 &\equiv_T \Holant(\hat{f}_{(k)}),
\end{align*}
where $\hat{f}_{(k)} = [2^{k} - 2, 0 , 2 , 0, -2, \dotsc]$.
Here we transform from one real-weighted Holant problem to another real-weighted Holant problem via a complex-weighted transformation.
The hardness follows from Theorem~\ref{thm:dic:single}.
Like the previous problem,
\#\textsc{NoSinkNoSourceOrientation} is tractable modulo $2^t$,
where $t$ is the minimal degree of the input graph.

Our dichotomy theorem also applies to a set of signatures,
that is, different vertices may have different constraints.

\vspace*{\problemSpace}
\textbf{Problem:} \#\textsc{1In-Or-1Out-Orientation}

\textbf{Input:} An undirected graph $G$ with each vertex labeled ``1In'' or ``1Out''.

\textbf{Output:} The number of orientations of $G$ such that each vertex has exactly~1 incoming or exactly~1 outgoing edge as specified by its label.
\vspace*{\problemSpace}

This problem is denoted by $\holant{[0,1,0]}{\mathcal{F}}$,
where the set $\mathcal{F}$ consists of signatures of the form $f = [0,1,0,\dotsc,0]$ and $g = [0,\dotsc,0,1,0]$.
Once again, it is $\SHARPP$-hard on $k$-regular graphs for $k \ge 3$.
We apply the same transformation as in the above two orientation problems.
The result is $\Holant(\{\hat{f},\hat{g}\})$,
where $\hat{f} = [k, (k-2)i, -(k-4), \dotsc]$ and $\hat{g} = [k, -(k-2)i, -(k-4), \dotsc]$ of arity $k$.
In fact, the entries of $\hat{f}$ satisfy a second order recurrence relation with characteristic polynomial $(x-i)^2$
while the entries of $\hat{g}$ satisfy one with characteristic polynomial $(x+i)^2$.
The hardness follows from Theorem~\ref{thm:main}.
However, the restriction of this problem to planar graphs is tractable by matchgates~\cite{CL11a}.
Alternatively, if we only consider one signature, then either $\Holant(\hat{f})$ or $\Holant(\hat{g})$ is tractable.
The problem $\Holant(\hat{f})$ is equivalent to the problem $\holant{[0,1,0]}{[0,1,0,\dotsc,0]}$,
which is always~$0$ provided $k \ge 3$ by a simple counting argument.
Similarly for $\Holant(\hat{g})$.
Therefore, despite the complicated-looking $\hat{f}$ and $\hat{g}$,
the Holant value for any input graph using only $\hat{f}$ or $\hat{g}$ is always~$0$.
These are what we call vanishing signatures.
This is also an example where combining two vanishing signatures induces $\SHARPP$-hardness.

One sufficient condition for a signature to be vanishing is that its entries satisfy a second order recurrence relation with characteristic polynomial $(x \pm i)^2$.
If the entries of a signature $f$ satisfy a second order recurrence relation with characteristic polynomial $(x-a)^2$ for $a \neq \pm i$,
then there exists an orthogonal holographic transformation such that $f$ is transformed into a weighted matching signature.

\vspace*{\problemSpace}
\textbf{Problem:} \#\textsc{$\lambda$-WeightedMatching}

\textbf{Input:} An undirected graph $G$.

\textbf{Output:} ${\displaystyle \sum_{M \in \mathcal{M}(G)} \lambda^{v(M)}}$,\\
where $\mathcal{M}(G)$ is the set of all matchings in $G$ and $v(M)$ is the number of unmatched vertices in the matching $M$.
\vspace*{\problemSpace}

The Holant expression of this problem is $\Holant(\mathcal{F})$,
where $\mathcal{F}$ consists of signatures of the form $[\lambda, 1, 0, \dotsc, 0]$.
When $\lambda = 0$, this problem counts perfect matchings,
which is $\SHARPP$-hard even for bipartite graphs~\cite{Val79b} but tractable over planar graphs by Kasteleyn's algorithms~\cite{Kas67}.
When $\lambda = 1$, this problem counts general matchings.
Vadhan~\cite{Vad01} proved that counting general matchings is $\SHARPP$-hard over $k$-regular graphs for $k \ge 5$,
but left open the question for $k = 4$.
Theorem~\ref{thm:dic:single} shows that \#\textsc{$\lambda$-WeightedMatching} is $\SHARPP$-hard,
for any weight $\lambda$ and on any $k$-regular graphs for $k \ge 3$.
The power of our dichotomy theorem is such that it gives a sweeping classification for \emph{all} such problems;
the open case for $k = 4$ from~\cite{Vad01} is a single \emph{point} in the problem space.

\section{Vanishing Signatures} \label{sec:vanishing}

Vanishing signatures were first introduced in~\cite{GLV13} in the parity setting to denote signatures for which the Holant value is always~$0$ modulo~$2$.

\begin{definition}
 A set of signatures $\mathcal{F}$ is called \emph{vanishing} if the value $\Holant_\Omega(\mathcal{F})$ is~$0$ for every signature grid $\Omega$.
 A signature $f$ is called \emph{vanishing} if the singleton set $\{f\}$ is vanishing.
\end{definition}

In this section, we characterize all sets of symmetric vanishing signatures.
First we observe that a simple lemma (Lemma~6.2 in~\cite{GLV13}) from the parity setting works over any field $\mathbb{F}$, with the same proof.
It also works for general, not necessarily symmetric, signatures.
Let $f + g$ denote the entry-wise addition of two signatures $f$ and $g$ with the same arity,
i.e.~$(f + g)_\ell = f_\ell + g_\ell$ for any index $\ell$.

\begin{lemma} \label{lem:van:linear}
 Let $\mathcal{F}$ be a vanishing signature set.
 If a signature $f$ can be realized by a gadget using signatures in $\mathcal{F}$,
 then $\mathcal{F} \cup \{f\}$ is also vanishing.
 If $f$ and $g$ are two signatures in $\mathcal{F}$ of the same arity,
 then $\mathcal{F} \cup \{f + g\}$ is vanishing as well.
\end{lemma}

Obviously, the identically zero signature, in which all entries are~0, is vanishing.
This is trivial.
However, we show that the concept of vanishing signatures is not trivial.
Notice that the unary signature $[1,i]$ when connected to another $[1,i]$ has a Holant value~0.
Consider a signature set $\mathcal{F}$ where every signature of arity $n$ is degenerate.
That is, every signature of arity $n$ is a tensor product of unary signatures.
Moreover, for each signature, suppose that more than half of the unary signatures in the tensor product are $[1,i]$.
For any signature grid $\Omega$ with signatures from $\mathcal{F}$,
it can be decomposed into many pairs of unary signatures.
The total Holant value is the product of the Holant on each pair.
Since more than half of the unaries in each signature are $[1,i]$,
more than half of the unaries in $\Omega$ are $[1,i]$.
Then two $[1,i]$'s must be paired up and hence $\Holant_{\Omega} = 0$.
Thus, all such signatures form a vanishing set.
We also observe that this argument holds when $[1,i]$ is replaced by $[1,-i]$.

These signatures described above are generally not symmetric and our present aim is to characterize symmetric vanishing signatures.
To this end, we define the following symmetrization operation.

\begin{definition} \label{def:sym}
 Let $S_n$ be the symmetric group of degree $n$.
 Then for positive integers $t$ and $n$ with $t \le n$ and unary signatures $v, v_1, \dotsc, v_{n-t}$,
 we define
 \[
  \Sym_n^t(v; v_1, \dotsc, v_{n-t}) = \sum_{\pi \in S_n} \bigotimes_{k=1}^n u_{\pi(k)},
 \]
 where the ordered sequence
 $(u_1, u_2, \dotsc, u_n) = ( \underbrace{v, \dotsc, v}_{t \text{ copies}}, v_1, \dotsc, v_{n-t} )$.
\end{definition}

Note that we include redundant permutations of $v$ in the definition.
Equivalent $v_i$'s also induce redundant permutations.
These redundant permutations simply introduce a nonzero constant factor, which does not change the complexity.
However, the allowance of redundant permutations simplifies our calculations.
An illustrative example of Definition~\ref{def:sym} is
\begin{align*}
 \Sym_3^2([1,i]; [a,b])
 &= 2 [a,b] \otimes [1,i] \otimes [1,i] + 2 [1,i] \otimes [a,b] \otimes [1,i] + 2 [1,i] \otimes [1,i] \otimes [a,b]\\
 &= 2 [3 a, 2 i a + b, -a + 2 i b, -3 b].
\end{align*}

\begin{definition}
 A nonzero symmetric signature $f$ of arity $n$ has \emph{positive vanishing degree} $k \ge 1$, which is denoted by $\vd^+(f) = k$,
 if $k \le n$ is the largest positive integer such that there exists $n-k$ unary signatures $v_1, \dotsc, v_{n-k}$ satisfying
 \begin{align*}
  f = \Sym_n^{k}([1,i]; v_1, \dots, v_{n-k}).
 \end{align*}
 If $f$ cannot be expressed as such a symmetrization form, we define $\vd^+(f) = 0$.
 If $f$ is the all zero signature, define $\vd^+(f) = n + 1$.
 
 We define \emph{negative vanishing degree} $\vd^-$ similarly, using $-i$ instead of $i$.
\end{definition}

Notice that it is possible for a signature $f$ to have both $\vd^+(f)$ and $\vd^-(f)$ nonzero.
For example, $f = [1,0,1]$ has $\vd^+(f) = \vd^-(f) = 1$.

By the discussion above and Lemma~\ref{lem:van:linear},
we know that for a signature $f$ of arity $n$,
if $\vd^\sigma(f) > \frac{n}{2}$ for some $\sigma \in \{+,-\}$,
then $f$ is a vanishing signature.
This argument is easily generalized to a set of signatures.

\begin{definition}
 For $\sigma \in \{+, -\}$,
 we define $\mathscr{V}^\sigma = \{f \st 2 \vd^\sigma(f) > \arity(f)\}$.
\end{definition}

\begin{lemma} \label{lem:sym:van}
 Let $\mathcal{F}$ be a set of symmetric signatures.
 If $\mathcal{F} \subseteq \mathscr{V}^{+}$ or $\mathcal{F} \subseteq \mathscr{V}^{-}$,
 then $\mathcal{F}$ is vanishing.
\end{lemma}

In Theorem~\ref{thm:van}, we show that these two sets capture all symmetric vanishing signature sets.

\subsection{Characterizing Vanishing Signatures using Recurrence Relations}

Now we give an equivalent characterization of vanishing signatures.

\begin{definition}
 A symmetric signature $f = [f_0, f_1, \dotsc, f_n]$ of arity $n$ is in $\mathscr{R}^{+}_t$ for a nonnegative integer $t \ge 0$
 if $t > n$ or for any $0 \le k \le n - t$, $f_k, \dotsc, f_{k+t}$ satisfy the recurrence relation
 \begin{align}
  \binom{t}{t} i^t f_{k+t} + \binom{t}{t-1} i^{t-1} f_{k+t-1} + \dotsb + \binom{t}{0} i^0 f_k = 0. \label{eqn:recurrence}
 \end{align}
 We define $\mathscr{R}^{-}_t$ similarly but with $-i$ in place of $i$ in~(\ref{eqn:recurrence}).
\end{definition}

It is easy to see that $\mathscr{R}^+_0 = \mathscr{R}^-_0$ is the set of all zero signatures.
Also, for $\sigma \in \{+,-\}$, we have $\mathscr{R}^\sigma_t \subseteq \mathscr{R}^\sigma_{t'}$ when $t \le t'$.
By definition, if $\arity(f) = n$ then $f \in \mathscr{R}^\sigma_{n+1}$.

Let $f = [f_0, f_1, \dotsc, f_n] \in \mathscr{R}^+_t$ with $0 < t \le n$.
Then the characteristic polynomial of its recurrence relation is $(1 + x i)^t$.
Thus there exists a polynomial $p(x)$ of degree at most $t-1$ such that $f_k = i^k p(k)$, for $0 \le k \le n$.
This statement extends to $\mathscr{R}^+_{n+1}$ since a polynomial of degree $n$ can interpolate any set of $n+1$ values.
Furthermore, such an expression is unique.
If there are two polynomials $p(x)$ and $q(x)$, both of degree at most $n$, such that $f_k = i^k p(k) = i^k q(k)$ for $0 \le k \le n$,
then $p(x)$ and $q(x)$ must be the same polynomial.
Now suppose $f_k = i^k p(k)$ ($0 \le k \le n$) for some polynomial $p$ of degree at most $t-1$, where $0 < t \le n$.
Then $f$ satisfies the recurrence~(\ref{eqn:recurrence}) of order $t$.
Hence $f \in \mathscr{R}^{+}_t$.

Thus $f \in \mathscr{R}^{+}_{t+1}$
iff there exists a polynomials $p(x)$ of degree at most $t$ such that $f_k = i^k p(k)$ ($0 \le k \le n$), for all $0  \le t \le n$.
For $\mathscr{R}^-_{t+1}$, just replace $i$ by $-i$.

\begin{definition}
 For a nonzero symmetric signature $f$ of arity $n$,
 it is of \emph{positive} (resp.~\emph{negative}) \emph{recurrence degree} $t \le n$, denoted by $\rd^+(f) = t$ (resp.~$\rd^-(f) = t$),
 if and only if $f \in \mathscr{R}^+_{t+1} - \mathscr{R}^+_{t}$ (resp.~$f \in \mathscr{R}^-_{t+1} - \mathscr{R}^-_{t}$).
 If $f$ is the all zero signature, we define $\rd^+(f) = \rd^-(f)= -1$.
\end{definition}

Note that although we call it the recurrence degree, it refers to a special kind of recurrence relation.
For any nonzero symmetric signature $f$, by the uniqueness of the representing polynomial $p(x)$,
it follows that $\rd^\sigma(f) = t$ iff $\deg(p) = t$, where $0 \le t \le n$.
We remark that $\rd^\sigma(f)$ is the maximum integer $t$ such that $f$ does \emph{not} belong to $\mathscr{R}^\sigma_t$.
Also, for an arity $n$ signature $f$, $\rd^\sigma(f) = n$
if and only if $f$ does not satisfy any such recurrence relation~(\ref{eqn:recurrence}) of order $t \le n$ for $\sigma \in \{+,-\}$.

\begin{lemma} \label{lem:cha:sym}
 Let $f = [f_0, \dotsc, f_n]$ be a symmetric signature of arity $n$, not identically 0.
 Then for any nonnegative integer $0 \le t < n$ and $\sigma \in \{+, -\}$, the following are equivalent:
 \begin{enumerate}
  \item[(i)] There exist $t$ unary signatures $v_1, \dotsc, v_t$, such that
   \begin{equation}
    f = \Sym_n^{n-t}([1, \sigma i]; v_1, \dots, v_t). \label{eqn:define-f-by-sym}
   \end{equation}
  \item[(ii)] $f \in \mathscr{R}^\sigma_{t+1}$.
 \end{enumerate}
\end{lemma}

\begin{proof}
 We consider $\sigma = +$ since the other case is similar, so let $v = [1,i]$.

 We start with $(i) \implies (ii)$ and proceed via induction on both $t$ and $n$.
 For the first base case of $t = 0$, $\Sym_n^n(v) = [1,i]^{\otimes n} = [1, i, -1, -i, \dotsc, i^n]$, so $f_{k+1} = i f_k$ for all $0 \le k \le n - 1$ and $f \in \mathscr{R}^+_1$.

 The other base case is that $t = n-1$.
 Let $\Sym_n^1(v; v_1, \dotsc, v_t) = [f_0, \dotsc, f_n]$ where $v_i = [a_i, b_i]$ for $1 \le i \le t$, and $S = i^n f_{n}  + \dots + \binom{n}{1} i f_1 + \binom{n}{0} i^0 f_{0}$.
 We need to show that $S = 0$.
 First notice that any entry in $f$ is a linear combination of terms of the form $a_{i_1} a_{i_2} \dotsm a_{i_{n-1-k}} b_{j_1} \dotsm b_{j_k}$,
 where $0 \le k \le n-1$, and $\{i_1, \dotsc, i_{n-1-k}, j_1, \dotsc, j_k\} = \{1, 2, \dotsc, n-1\}$.
 Thus $S$ is a linear combination of such terms as well.
 Now we compute the coefficient of each of these terms in $S$.

 Each term $a_{i_1} a_{i_2} \dotsm a_{i_{n-1-k}} b_{j_1} \dotsm b_{j_k}$ appears twice in $S$, once in $f_k$ and the other time in $f_{k+1}$.
 In $f_k$, the coefficient is $k! (n-k)!$, and in $f_{k+1}$, it is $i (k+1)! (n-k-1)!$.
 Thus, its coefficient in $S$ is
 \[
  \binom{n}{k+1} i^{k+1} i (k+1)! (n-k-1)!
  +
  \binom{n}{k} i^k k! (n-k)! =0.
 \]
 The above computation works for any such term due to the symmetry of $f$, so all coefficients in $S$ are~0, which means that $S = 0$.

 Now assume for any $t' < t$ or for the same $t$ and any $n' < n$, the statement holds.
 For $(n,t)$, where $n > t + 1$, assume that $f = [f_0, \dotsc, f_n] = \Sym_n^{n-t}(v; v_1, \dotsc, v_t)$, $g = \Sym_{n-1}^{n-t-1}(v; v_1, \dotsc, v_t) = [g_0, \dotsc, g_{n-1}]$,
 and for any $1 \le j \le t$, 
 \begin{align*}
   h^{(j)} = \Sym_{n-1}^{n-t}(v; v_1, \dotsc, v_{j-1}, v_{j+1}, \dotsc, v_t) = [h_0^{(j)}, \dotsc, h_{n-1}^{(j)}].
 \end{align*}
 By the induction hypothesis, $g$ satisfies the recurrence relation of order $t+1$, namely $g \in \mathscr{R}^+_{t+1}$.
 Also for any $j$, $h^{(j)}$ satisfies the recurrence relation of order $t$, namely $h^{(j)} \in \mathscr{R}^+_t \subseteq \mathscr{R}^+_{t+1}$.

 We have the recurrence relation
 \begin{align}
  \Sym_n^{n-t}(v; v_1, \dotsc, v_t) \label{eqn:sym:rec}
  =            (n-t) v \otimes &\Sym_{n-1}^{n-t-1}(v; v_1, \dotsc, v_t) \\
    + \sum_{j=1}^t v_j \otimes &\Sym_{n-1}^{n-t}(v; v_1, \dotsc, v_{j-1}, v_{j+1}, \dotsc, v_t). \notag
 \end{align}

 By~(\ref{eqn:sym:rec}), the entry of weight $k$ in $f$ for any $k > 0$ is
 \[f_k = (n-t) i g_{k-1} + \sum_{j=1}^t b_j h_{k-1}^{(j)}.\]
 We know that $\{g_i\}$ and $\{h_i^{(j)}\}$ satisfy the recurrence relation~(\ref{eqn:recurrence}) of order $t+1$.
 Thus, their linear combination $\{f_i\}$ also satisfies the recurrence relation~(\ref{eqn:recurrence}) starting from $i = k > 0$.

 We also observe that by~(\ref{eqn:sym:rec}), the entry of weight $k$ in $f$ for any $k < n$ is
 \[f_k = (n-t) g_{k} + \sum_{j=1}^t a_j h_k^{(j)}.\]
 Since $t < n-1$, by the same argument, the recurrence relation~(\ref{eqn:recurrence}) holds for $f$ when $k = 0$ as well.

 Now we show $(ii) \implies (i)$.
 Notice that we only need to find unary signatures $\{v_i\}$ for $1 \le i \le t$ such that $\Sym_n^{n-t}(v; v_1, \dotsc, v_t)$ matches the first $t+1$ entries of $f$.
 The theorem follows from this since we have shown that $\Sym_n^{n-t}(v; v_1, \dotsc, v_t)$ satisfies the recurrence relation of order $t+1$ and any such signature is determined by the first $t+1$ entries.

 We show that there exist $v_i = [a_i, b_i]$ ($1 \le i \le t$) satisfying the above requirement.
 Since $f$ is not identically 0, by~(\ref{eqn:recurrence}), some nonzero term occurs among $\{f_0, \dotsc, f_{t} \}$.
 Let $f_{s} \neq 0$, for $0 \le s \le t$, be the first nonzero term.
 By a nonzero constant multiplier, we may normalize $f_s = s! (n-s)!$,
 and set  $v_j = [0, 1]$, for $1 \le j \le s$ (which is vacuous if $s=0$),
 and set $v_{s+j} = [1, b_{s+j}]$, for $1 \le j \le t-s$ (which is vacuous if $s=t$).
 We will set up a system of polynomial equations with $b_{s+j}$'s as variables.
 Solving it will give us desired $v_{s+j}$'s.

 Let $F$ be the function defined in~(\ref{eqn:define-f-by-sym}).
 Then $F_k = f_k = 0$ for $0 \le k < s$ (which is vacuous if $s=0$).
 By expanding the symmetrization function, for $s \le k \le t$, we get
 \[F_k = k! (n-k)! \sum_{j=0}^{k-s} \binom{n-t}{k-s-j} \Delta_j i^{k-s-j},\]
 where $\Delta_j$ is the elementary symmetric polynomial in $\{b_{s+1}, \dots, b_{t}\}$ of degree $j$ for $0 \le j \le t-s$.
 By definition, $\Delta_0 = 1$ and $F_s = f_s$.
 Setting $F_k = f_k$ for $s+1  \le k \le t$, this is a linear equation system on $\Delta_j$ ($1  \le j \le t-s$), with a triangular matrix and nonzero diagonals.
 From this, we know that all $\Delta_j$'s are uniquely determined by $\{f_{s+1}, \dots, f_t\}$.
 Moreover, $\{b_{s+1}, \dots, b_{t}\}$ are the roots of the equation $\sum_{j=0}^{t-s} (-1)^j \Delta_j x^{t-s-j} = 0$.
 Thus $\{b_{s+1}, \dots, b_{t}\}$ are also uniquely determined by $\{f_{s+1}, \dots, f_t\}$ up to a permutation.
\end{proof}

\begin{corollary} \label{cor:van:degree}
 If $f$ is a symmetric signature and $\sigma \in \{+,-\}$,
 then $\vd^\sigma(f) + \rd^\sigma(f) = \arity(f)$.
\end{corollary}

Thus we have an equivalent form of $\mathscr{V}^\sigma$ for $\sigma \in \{+,-\}$.
Namely,
\[\mathscr{V}^\sigma = \{f \st 2 \rd^\sigma(f) < \arity(f)\}.\]

\begin{figure}[t]
 \centering
 \begin{tikzpicture}[scale=\scale,transform shape,node distance=\nodeDist,semithick]
  \node[internal] (0)                   {};
  \node[external] (1) [above left of=0] {};
  \node[external] (2) [below left of=0] {};
  \node[external] (3) [left       of=1] {};
  \node[external] (4) [left       of=2] {};
   \path let
          \p1 = (3),
          \p2 = (4)
         in
          node[external] (5) at (\x1, \y1 / 2 + \y2 / 2) {};
  \node[external]  (6) [right       of=0] {};
  \node[internal]  (7) [right       of=6] {};
  \node[external]  (8) [above right of=7] {};
  \node[external]  (9) [below right of=7] {};
  \node[external] (10) [right       of=8] {};
  \node[external] (11) [right       of=9] {};
   \path let
          \p1 = (10),
          \p2 = (11)
         in
          node[external] (12) at (\x1, \y1 / 2 + \y2 / 2) {};
  \path (0) edge[out= 135, in=0]    (3)
            edge[out=-135, in=0]    (4)
            edge                    (5)
            edge[bend left]         (7)
            edge[bend right]        (7)
            edge[out= 70, in= 110]  (7)
            edge[out=-70, in=-110]  (7)
        (7) edge[out= 45, in= 180] (10)
            edge[out=-45, in= 180] (11)
            edge                   (12);
  \begin{pgfonlayer}{background}
   \node[inner sep=0pt,transform shape=false,draw=\borderColor,thick,rounded corners,fit = (1) (2) (8) (9)] {};
  \end{pgfonlayer}
 \end{tikzpicture}
 \caption{Example of a gadget used to create a degenerate vanishing signature from some general vanishing signature.
  This example is for a signature of arity 7 and recurrence degree 2, which is assigned to both vertices.}
 \label{fig:gadget:vanishing_to_degenerate}
\end{figure}
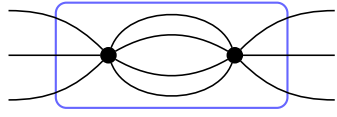

\subsection{Characterizing Vanishing Signature Sets}

Now we show that $\mathscr{V}^{+}$ and $\mathscr{V}^{-}$ capture all symmetric vanishing signature sets.
To begin, we show that a vanishing signature set cannot contain both types of nontrivial vanishing signatures.

\begin{lemma} \label{lem:van:mix}
 Let $f_+ \in \mathscr{V}^{+}$ and $f_- \in \mathscr{V}^{-}$.
 If neither $f_+$ nor $f_-$ is the all zero signature,
 then the signature set $\{f_+, f_-\}$ is not vanishing.
\end{lemma}

\begin{proof}
 Let $\arity(f_+) = n$ and $\rd^+(f_+) = t$, so $2 t < n$.
 Consider the gadget with two vertices and $2 t$ edges between two copies of $f_+$.
 (See Figure~\ref{fig:gadget:vanishing_to_degenerate} for an example of this gadget.)
 View $f_+$ in the symmetrized form.
 Since $\vd^+(f_+) = n-t$, in each term, there are $n-t$ many $[1,i]$'s and $t$ many unary signatures not equal to (a multiple of) $[1,i]$.
 This is a superposition of many degenerate signatures.
 Then the only non-vanishing contributions come from the cases where the $n-2t$ dangling edges on both sides are all assigned $[1,i]$,
 while inside, the $t$ copies of $[1,i]$ pair up with $t$ unary signatures not equal to $[1,i]$ from the other side perfectly.
 Notice that for any such contribution,
 the Holant value of the inside part is always the same constant and this constant is not~$0$ because $[1,i]$ paired up with any unary signature other than (a multiple of) $[1,i]$ is not~$0$.
 Then the superposition of all of the permutations is a degenerate signature $[1,i]^{\otimes 2 (n - 2 t)}$ up to a nonzero constant factor.

 Similarly, we can do this for $f_-$ of arity $n'$ and $\rd^-(f_-) = t'$, where $2 t' < n'$, and get a degenerate signature $[1,-i]^{\otimes 2 (n' - 2 t')}$, up to a nonzero constant factor.
 Then form a bipartite signature grid with $(n' - 2 t')$ vertices on one side,
 each assigned $[1,i]^{\otimes 2 (n - 2 t)}$, and $(n - 2 t)$ vertices on the other side,
 each assigned $[1,-i]^{\otimes 2 (n' - 2 t')}$.
 Connect edges between the two sides arbitrarily as long as it is a 1-1 correspondence.
 The resulting Holant is a power of $2$, which is not vanishing.
\end{proof}

\begin{lemma} \label{lem:cha:van}
 Every symmetric vanishing signature is in $\mathscr{V}^{+} \cup \mathscr{V}^{-}$.
\end{lemma}

\begin{proof}
 Let $f$ be a symmetric vanishing signature.
 We prove this by induction on $n$, the arity of $f$.
 For $n=1$, by connecting $f = [f_0, f_1]$ to itself, we have $f_0^2 + f_1^2 = 0$.
 Then up to a constant factor, we have either $f = [1,i]$ or $f = [1,-i]$.
 The lemma holds.

 For $n = 2$, first we do a self loop.
 The Holant is $f_0 + f_2$.
 Also, we can connect two copies of $f$, in which case the Holant is $f_0^2 + 2 f_1^2 + f_2^2$.
 Since $f$ is vanishing, both $f_0 + f_2 = 0$ and $f_0^2 + 2 f_1^2 + f_2^2 = 0$.
 Solving them, we get $f = [1,i,-1] = [1,i]^{\otimes 2}$ or $[1,-i,-1] = [1,-i]^{\otimes 2}$ up to a constant factor.

 Now assume $n > 2$ and the lemma holds for any signature of arity $k < n$.
 Let $f = [f_0, f_1, \dotsc, f_n]$ be a vanishing signature.
 A self loop on $f$ gives $f' = [f'_0, f'_1, \dotsc, f'_{n-2}]$, where $f'_j = f_j + f_{j+2}$ for $0 \leq j \leq n-2$.
 Since $f$ is vanishing, $f'$ is vanishing as well.
 By the induction hypothesis, $f' \in \mathscr{V}^{+} \union \mathscr{V}^{-}$.

 If $f'$ is an all zero signature,
 then we have $f_j + f_{j+2} = 0$ for $0 \leq j \leq n-2$.
 This means that the $f_j$'s satisfy a recurrence relation with characteristic polynomial $x^2 + 1$, so we have $f_j = a i^j + b(-i)^j$ for some $a$ and $b$.
 Then we perform a holographic transformation with $Z = \frac{1}{\sqrt{2}} \left[\begin{smallmatrix} 1 & 1 \\ i & -i \end{smallmatrix}\right]$,
 \begin{align*}
  \holant{{=}_2}{f}
  &\equiv_T \holant{[1,0,1] Z^{\otimes 2}}{(Z^{-1})^{\otimes n} f}\\
  &\equiv_T \holant{[0,1,0]}{\hat{f}},
 \end{align*}
 where $\hat{f}=[a, 0, \dotsc, 0, b]$.
 The problem $\holant{[0,1,0]}{\hat{f}}$ is a weighted version of testing if a graph is bipartite.
 Now consider a graph with only two vertices,
 both assigned $f$,
 and $n$ edges between them.
 The Holant of this graph is $2 a b$.
 However,
 we know that it must be vanishing,
 so $a b = 0$.
 If $a = 0$,
 then $f \in \mathscr{V}^-$.
 Otherwise,
 $b = 0$ and $f \in \mathscr{V}^+$.

 Now suppose that $f'$ is in $\mathscr{V}^{+} \union \mathscr{V}^{-}$ but is not an all zero signature.
 We consider $f' \in \mathscr{V}^{+}$ since the other case is similar.
 Then $\rd^+(f') = t$, so $2 t < n - 2$.
 Consider the gadget which has only two vertices, both assigned $f'$, and has $2 t$ edges between them.
 (See Figure~\ref{fig:gadget:vanishing_to_degenerate} for an example of this gadget.)
 It forms a signature of degree $d = 2 (n - 2 - 2 t)$.
 This gadget is valid because $n - 2 > 2 t$.
 By the combinatorial view as in the proof of Lemma~\ref{lem:van:mix}, this signature is $[1,i]^{\otimes d}$.

 Moreover, $\rd^+(f') = t$ implies that the entries of $f'$ satisfy a recurrence of order $t+1$.
 Replacing $f'_j$ by $f_j + f_{j+2}$, we get a recurrence relation for the entries of $f$ with characteristic polynomial $(x^2 + 1) (x - i)^{t+1} = (x + i) (x - i)^{t+2}$.
 Thus, $f_j = i^j p(j) + c (-i)^j$ for some polynomial $p(x)$ of degree at most $t + 1$ and some constant $c$.
 It suffices to show that $c = 0$ since $2 (t + 1) < n$ as $2 t < n - 2$.

 Consider the signature $h = [h_0, \dotsc, h_{n-1}]$ created by connecting $f$ with a single unary signature $[1,i]$.
 For any ($n-1$)-regular graph $G = (V,E)$ with $h$ assigned to every vertex, we can define a duplicate graph of $(d+1)|V|$ vertices as follows.
 First for each $v \in V$, define vertices $v'$, $v_1, \dots, v_{d}$.
 For each $i$, $1 \le i \le d$, we make a copy of $G$ on $\{v_i \mid v \in V\}$, i.e., for each edge $(u,v) \in E$, include the edge $(u_i, v_i)$ in the new graph.
 Next for each $v \in V$, we introduce edges between $v'$ and $v_i$ for all $1 \leq i \leq d$.
 For each $v \in V$, assign the degenerate signature $[1,i]^{\otimes d}$ that we just constructed to the vertices $v'$; assign $f$ to all the vertices $v_1, \dots, v_{d}$.
 Assume the Holant of the original graph $G$ with $h$ assigned to every vertex is $H$.
 Then for the new graph with the given signature assignments, the Holant is $H^d$.
 By our assumption, $f$ is vanishing, so $H^d = 0$.
 Thus, $H = 0$.
 This holds for any graph $G$, so $h$ is vanishing.

 Notice that $h_k = f_k + if_{k+1}$ for any $0 \leq k \leq n-1$.
 If $h$ is identically zero,
 then $f_k + i f_{k+1} = 0$ for any $0 \leq k \leq n-1$,
 which means $f = [1,i]^{\otimes n}$ up to a constant factor and we are done.
 Otherwise,
 suppose that $h$ is not identically zero.
 By the inductive hypothesis, $h \in \mathscr{V}^{+} \cup \mathscr{V}^{-}$.
 We claim $h$ cannot be from $\mathscr{V}^{-}$.
 This is because, although we do not directly construct $h$ from $f$, we can always realize it by the method depicted in the previous paragraph.
 Therefore the set $\{f',h\}$ is vanishing.
 As both $f'$ and $h$ are nonzero, and $f' \in \mathscr{V}^{+}$, we have $h \not \in \mathscr{V}^{-}$, by Lemma~\ref{lem:van:mix}.

 Hence  $h$ is in $\mathscr{V}^{+}$.
 Then there exists a polynomial $q(x)$ of degree at most $t' = \floor{\frac{n-1}{2}}$ such that $h_k = i^k q(k)$, for any $0 \leq k \leq n - 1$.
 Since $2 t < n - 2$, we have $t \leq t'$.
 On the other hand, $h_k = f_k + i f_{k+1}$ for any $0 \leq k \leq n - 1$, so we have
 \begin{align*}
  h_k
  &= f_k + i f_{k+1}\\
  &= i^k p(k) + c(-i)^k + i \left(i^{k+1} p(k+1)+ c(-i)^{k+1}\right)\\
  &= i^k \left(p(k) - p(k+1)\right) + 2 c (-i)^k\\
  &= i^k r(k) + 2 c (-i)^k\\
  &= i^k q(k),
 \end{align*}
 where $r(x) = p(x) - p(x+1)$ is another polynomial of degree at most $t$.
 Then we have
 \begin{equation*}
  q(k) - r(k) = 2 c (-1)^k,
 \end{equation*}
 which holds for all $0 \leq k \leq n-1$.
 Notice that the left hand side is a polynomial of degree at most $t'$, call it $s(x)$.
 However, for all even $k \in \{0, \dots, n-1\}$, $s(k) = 2 c$.
 There are exactly $\ceil{\frac{n}{2}} > \floor{\frac{n-1}{2}} = t'$ many even $k$ within the range $\{0, \dots, n-1\}$.
 Thus $s(x) = 2 c$ for any $x$.
 Now we pick $k=1$, so $s(1) = -2 c = 2 c$, which implies $c = 0$.
 This completes the proof.
\end{proof}

Combining Lemma~\ref{lem:sym:van}, Lemma~\ref{lem:van:mix}, and Lemma~\ref{lem:cha:van},
we obtain the following theorem that characterizes all symmetric vanishing signature sets.

\begin{theorem} \label{thm:van}
 Let $\mathcal{F}$ be a set of symmetric signatures.
 Then $\mathcal{F}$ is vanishing
 if and only if
 $\mathcal{F} \subseteq \mathscr{V}^{+}$ or $\mathcal{F} \subseteq \mathscr{V}^{-}$.
\end{theorem}

\noindent
We note that some particular categories of tractable cases in previous dichotomies
(case~\ref{case:arity3:exceptional} of Theorem~\ref{thm:arity3:singleton},
case~\ref{case:holant_star:exceptional} of Theorem~\ref{thm:holant_star},
and case~\ref{case:holant_c:exceptional} of Theorem~\ref{thm:holant_c})
are in $\mathscr{R}^\pm_2$.

To finish this subsection,
we prove some useful properties regarding vanishing and recurrence degrees in the construction of signatures.
For two symmetric signatures $f$ and $g$ such that $\arity(f) \ge \arity(g)$,
let $\langle f, g \rangle = \langle g, f \rangle$ denote the signature that results after connecting all edges of $g$ to $f$.
(If $\arity(f) = \arity(g)$,
then $\langle f, g \rangle$ is a constant,
which can be viewed as a signature of arity 0.)

\begin{lemma} \label{lem:van:con}
 For $\sigma \in \{+,-\}$, suppose symmetric signatures $f$ and $g$ satisfy $\vd^\sigma(g)$ $ = 0$ and $\arity(f) - \arity(g) \ge \rd^\sigma(f)$.
 Then $\rd^\sigma(\langle f, g\rangle) = \rd^\sigma(f)$.
\end{lemma}

\begin{proof}
 We consider $\sigma = +$ since the case $\sigma = -$ is similar.
 Let $\arity(f) = n$, $\arity(g) = m$, and $\rd^+(f) = t$.
 Denote the signature $\langle f, g\rangle$ by $f'$.

 If $t = -1$, then $f$ is identically 0 and so is $f'$.
 Hence $\rd^+(f') = -1$.

 Suppose $t \ge 0$.
 Then we have $f_k = i^k p(k)$ where $p(x)$ is a polynomial of degree exactly $t$.
 Also $\arity(f') = n- m \ge t$.
 We have
 \begin{align*}
  f_k'
  &=     \sum_{j=0}^{m} \binom{m}{j} f_{k+j}    g_j\\
  &= i^k \sum_{j=0}^{m} \binom{m}{j} p(k+j) i^j g_j\\
  &= i^k q(k),
 \end{align*}
 where $q(k) = \sum_{j=0}^{m} \binom{m}{j} p(k+j) i^j g_j$ is a polynomial in $k$.
 Notice that $\vd^+(g) = 0$.
 Then $\rd^+(g) = m$ and $g \not \in \mathscr{R}^+_{m}$.
 Thus $\sum_{j=0}^m \binom{m}{j} i^j g_j \neq 0$.
 Then the leading coefficient of degree $t$ in the polynomial $q(k)$ is nonzero.
 However, $\arity(f') \ge t$.
 Thus $\rd^+(f') = t$ as well.
\end{proof}

\begin{lemma} \label{lem:van:self}
 For $\sigma \in \{+,-\}$,
 let $f$ be a nonzero symmetric signature and suppose that $f'$ is obtained from $f$ by a self loop.
 If $\vd^\sigma(f) > 0$,
 then $\vd^\sigma(f) - \vd^\sigma(f') = \rd^\sigma(f) - \rd^\sigma(f') = 1$.
\end{lemma}

\begin{proof}
 We may assume $\sigma = +$, $\arity(f) = n$, and $\rd^+(f) = t$.
 Since $f$ is not the all zero signature,
 $t \ge 0$.
 Also since $\vd^+(f) > 0$,
 $t = n - \vd^+(f) < n$.
 By assumption,
 we have $f_k=i^kp(k)$, where $p(x)$ is a polynomial of degree exactly $t$.
 Then we have
 \begin{align*}
  f_k'
  &= f_k+f_{k+2}\\
  &= i^k(p(k)-p(k+2))\\
  &= i^kq(k),
 \end{align*}
 where $q(k)=p(k)-p(k+2)$ is a polynomial in $k$.
 If $t=0$, then $p(x)$ is a constant polynomial and $q(x)$ is identically zero.
 Then $\rd^+(f') = -1$ by definition and $\rd^+(f) - \rd^+(f') = 1$ holds.
 Suppose $t>0$, then in $q(k)$, the term of degree $t$ has a zero coefficient, but the term of degree $t-1$ is nonzero.
 So $q(x)$ has degree exactly $t-1 \le n-2 = \arity(f')$.
 Thus $\rd^+(f') = t-1$.
 Notice that $\arity(f) - \arity(f') = 2$, then $\vd^+(f) - \vd^+(f') = 1$ as well.
\end{proof}

Moreover,
the set of vanishing signatures is closed under orthogonal transformations.
This is because under any orthogonal transformation,
the unary signatures $[1,i]$ and $[1,-i]$ are either invariant or transformed into each other.
Then considering the symmetrized form of any signature,
we have the following lemma.

\begin{lemma} \label{lem:van:orth}
 For a symmetric signature $f$ of arity $n$, $\sigma \in \{+,-\}$, and an orthogonal matrix $T \in \mathbb{C}^{2 \times 2}$,
 either $\vd^\sigma(f) = \vd^\sigma(T^{\otimes n} f)$ or $\vd^\sigma(f) = \vd^{-\sigma}(T^{\otimes n} f)$.
\end{lemma}

\subsection{Characterizing Vanishing Signatures via a Holographic Transformation} \label{subsec:vanishing-by-holographic-Z}

There is another explanation for the vanishing signatures.
Given an $f \in \mathscr{V}^+$ with $\arity(f) = n$ and $\rd^+(f) = d$,
we perform a holographic transformation with $Z = \frac{1}{\sqrt{2}}
\left[\begin{smallmatrix} 1 & 1 \\ i & -i \end{smallmatrix}\right]$,
\begin{align*}
 \holant{{=}_2}{f}
 &\equiv_T \holant{[1,0,1] Z^{\otimes 2}}{(Z^{-1})^{\otimes n} f}\\
 &\equiv_T \holant{[0,1,0]}{\hat{f}},
\end{align*}
where $\hat{f}$ is of the form $[\hat{f_0}, \hat{f}_1, \dots, \hat{f}_{d}, 0, \dots,0]$, and $\hat{f}_d \neq 0$.
To see this,
note that $Z^{-1} = \tfrac{1}{\sqrt{2}} \left[\begin{smallmatrix} 1 & -i \\ 1 & i \end{smallmatrix}\right]$
and $Z^{-1} \left[\begin{smallmatrix} 1 \\ i \end{smallmatrix}\right] = \sqrt{2} \left[\begin{smallmatrix} 1 \\ 0 \end{smallmatrix}\right]$.
We know that $f$ has a symmetrized form,
such as $\Sym_n^{n-d}(\left[\begin{smallmatrix} 1 \\ i \end{smallmatrix}\right]; v_1, \dots, v_{d})$.
Then up to a factor of ${2}^{n/2}$,
we have $\hat{f} = (Z^{-1})^{\otimes n} f = \Sym_n^{n-d}(\left[\begin{smallmatrix} 1 \\ 0 \end{smallmatrix}\right]; u_1, \dots, u_{d})$,
where $u_i = Z^{-1} v_i$ for $1 \le i \le d$ and $u_i$ and $v_i$ are column vectors in $\mathbb{C}^2$.
From this expression for $\hat{f}$, it is clear that all entries of Hamming weight greater than $d$ in $\hat{f}$ are~$0$.
Moreover, if $\hat{f}_d=0$, then one of the $u_i$ has to be a multiple of $[1,0]$.
This contradicts the degree assumption of $f$, namely $\vd^+(f) = n - \rd^+(f) = n -d$ but not any higher.
Formally we have the following.

\begin{lemma} \label{lem:vanishing_form_in_Z_basis}
 Suppose $f$ is a symmetric signature of arity $n$.
 Let $\hat{f} = (Z^{-1})^{\otimes n} f$.
 If $\rd^+(f) = d$, then $\hat{f} = [\hat{f}_0, \hat{f}_1, \dotsc, \hat{f}_d, 0, \dotsc, 0]$ and $\hat{f}_d \ne 0$.
 Also $f \in \mathscr{R}_d^+$ if and only if all nonzero entries of $\hat{f}$ are among the first $d$ entries in its symmetric signature notation.
 
 Similarly, if $\rd^-(f) = d$, then $\hat{f} = [0, \dotsc, 0, \hat{f}_{n-d}, \dotsc, \hat{f}_n]$ and $\hat{f}_{n-d} \ne 0$.
 Also $f\in\mathscr{R}_d^-$ if and only if all nonzero entries of $\hat{f}$ are among the last $d$ entries in its symmetric signature notation.
\end{lemma}

By linearity, Lemma \ref{lem:vanishing_form_in_Z_basis} implies the following fact.
If $f = g + h$ is of arity $n$, where $\rd^+(g) = d$, $\rd^-(h) = d'$, and $d + d' < n$,
then after a holographic transformation by $Z$,
$\hat{f} = (Z^{-1})^{\otimes n} f$ takes the form $[\hat{g}_0, \dotsc,\hat{g}_{d}, 0, \dotsc, 0, \hat{h}_{d'}, \dotsc, \hat{h}_0]$,
with $n - d - d' - 1 \ge 0$ zeros in the middle of the signature.

In any bipartite graph for $\holant{[0,1,0]}{\hat{f}}$,
the binary \textsc{Disequality} $(\neq_2) = [0,1,0]$ on the left imposes the condition that half of the edges must take the value~0
and the other half must take the value~1.
On the right side, by $f \in \mathscr{V}^+$, we have $d < n/2$,
thus $\hat{f}$ requires that less than half of the edges are assigned the value~1.
Therefore the Holant is always~0.
A similar conclusion was reached in~\cite{CLX12} for certain 2-3 bipartite Holant problems with Boolean signatures.
However, the importance was not realized at that time.

Under this transformation, one can observe another interesting phenomenon.
For any $a,b \in \mathbb{C}$,
\[
 \holant{[0,1,0]}{[a,b,1,0,0]}
 \qquad \text{and} \qquad
 \holant{[0,1,0]}{[0,0,1,0,0]}
\]
take exactly the same value on every signature grid.
This is because, to contribute a nonzero term in the Holant, exactly half of the edges must be assigned~1.
Then for the first problem, the signature on the right can never contribute a nonzero value involving $a$ or $b$.
Thus the Holant values of these two problems on any signature grid are always the same.
Nevertheless, there exist $a,b \in \mathbb{C}$ such that there is no holographic transformation between these two problems.
We note that this is the first counterexample involving non-unary signatures in the Boolean domain to the converse of the Holant theorem,
which provides a negative answer to a conjecture made by Xia in~\cite[Conjecture~4.1]{Xia11}.

Moreover, $\holant{[0,1,0]}{[0,0,1,0,0]}$ counts Eulerian orientations in a $4$-regular graph.
This problem was proven $\SHARPP$-hard by Huang and Lu in Theorem~V.10 of~\cite{HL12} and plays an important role in our proof of hardness.
Translating back to the standard setting, the problem of counting Eulerian orientations in a 4-regular graph is $\Holant([3,0,1,0,3])$.
The problem $\holant{[0,1,0]}{[a,b,1,0,0]}$ corresponds to a certain signature $f = Z^{\otimes 4} [a,b,1,0,0]$ of arity~4 with recurrence degree~$2$.
It has a different appearance but induces exactly the same Holant value as the signature for counting Eulerian orientations.
Therefore, all such signatures are $\SHARPP$-hard as well.
We use this fact later.

\section{Main Result, Tractability Proof, and Outline of Hardness Proof}

Using the definitions from the previous section,
we can now formally state our main result.

\begin{theorem} \label{thm:main}
 Let $\mathcal{F}$ be any set of symmetric, complex-valued signatures in Boolean variables.
 Then $\Holant(\mathcal{F})$ is $\SHARPP$-hard unless $\mathcal{F}$ satisfies one of the following conditions,
 in which case the problem is computable in polynomial time:
 \begin{enumerate}
  \item All non-degenerate signatures in $\mathcal{F}$ are of arity at most~2; \label{case:main_tractable:trivial}
  \item $\mathcal{F}$ is $\mathscr{A}$-transformable; \label{case:main_tractable:CSP:A}
  \item $\mathcal{F}$ is $\mathscr{P}$-transformable; \label{case:main_tractable:CSP:P}
  \item $\mathcal{F} \subseteq \mathscr{V}^\sigma \union \{f \in \mathscr{R}_2^\sigma \st \arity(f) = 2\}$ for $\sigma \in \{+,-\}$; \label{case:main_tractable:vanishing_and_binary}
  \item All non-degenerate signatures in $\mathcal{F}$ are in $\mathscr{R}_2^\sigma$ for $\sigma \in \{+,-\}$. \label{case:main_tractable:vanishing_and_unary}
 \end{enumerate}
\end{theorem}

\noindent
Note that any signature in $\mathscr{R}_2^\sigma$ having arity at least~3 is a vanishing signature.
Thus all signatures of arity at least~3 in case~\ref{case:main_tractable:vanishing_and_unary} are vanishing.
While both cases~\ref{case:main_tractable:vanishing_and_binary} and~\ref{case:main_tractable:vanishing_and_unary}
are largely concerned with vanishing signatures, these two cases differ.
In case~\ref{case:main_tractable:vanishing_and_binary}, all signatures in $\mathcal{F}$,
including unary signatures but excluding binary signatures, must be vanishing of a single type $\sigma$;
the binary signatures are only required to be in $\mathscr{R}_2^\sigma$.
In contrast,
case~\ref{case:main_tractable:vanishing_and_unary} has no requirement placed on degenerate signatures which include all unary signatures.
Then all non-degenerate binary signatures are required to be in $\mathscr{R}_2^\sigma$.
Finally all non-degenerate signatures of arity at least~3 are also required to be in $\mathscr{R}_2^\sigma$,
which is a strong form of vanishing; they must have a large vanishing degree of type $\sigma$.

Case~\ref{case:main_tractable:vanishing_and_unary} is actually a known tractable case~\cite{CLX13a, CLX11a}.
Every signature (after replacing all degenerate signatures with corresponding ones) is a generalized Fibonacci signature with $m = \sigma 2 i$,
which means that every signature $[f_0, f_1, \dotsc, f_n] \in \mathcal{F}$ satisfies $f_{k+2} = m f_{k+1} + f_k$ for $0 \le k \le n - 2$.
However, we present a unified proof of tractability based on vanishing signatures.

\subsection{Tractability Proof for Theorem~\ref{thm:main}}

For any signature grid $\Omega$,
$\Holant_\Omega$ is the product of the Holant on each connected component,
so we only need to compute over connected components.

For case~\ref{case:main_tractable:trivial},
after decomposing all degenerate signatures into unary ones,
a connected component of the graph is either a path or a cycle and the Holant can be computed using matrix product and trace.
Cases~\ref{case:main_tractable:CSP:A} and~\ref{case:main_tractable:CSP:P} are tractable because,
after a particular holographic transformation, their instances are tractable instances of $\CSP(\mathcal{F})$ (cf.~\cite{CLX09a}).
For case~\ref{case:main_tractable:vanishing_and_binary},
any binary signature $g \in \mathscr{R}_2^\sigma$ has $\rd^\sigma(g) \le 1$,
and thus $\vd^\sigma(g) \ge 1 = \arity(g) / 2$.
Any signature $f \in \mathscr{V}^\sigma$ has $\vd^\sigma(f) > \arity(f) / 2$.
If $\mathcal{F}$ contains a signature $f$ of arity at least~3, then it must belong to $\mathscr{V}^\sigma$.
Then by the combinatorial view, more than half of the unary signatures are $[1, \sigma i]$, so $\Holant_\Omega$ vanishes.
On the other hand, if the arity of every signature in $\mathcal{F}$ is at most~2,
then we have reduced to case~\ref{case:main_tractable:trivial}.

Now consider case~\ref{case:main_tractable:vanishing_and_unary}.
After decomposing all degenerate signatures into unary ones,
recursively absorb any unary signature into its neighboring signature.
If it is connected to another unary signature, then this produces a global constant factor.
If it is connected to a binary signature, then this creates another unary signature.
We observe that if $f \in \mathscr{R}^\sigma_2$ has $\arity(f) \ge 2$,
then for any unary signature $u$, after connecting $f$ to $u$,
the signature $\langle f, u \rangle$ still belongs to $\mathscr{R}^\sigma_2$.
Hence after recursively absorbing all unary signatures in the above process,
we still have a signature grid where all signatures belong to $\mathscr{R}^\sigma_2$.
Any remaining signature $f$ that has arity at least~3 belongs to $\mathscr{V}^\sigma$
since $\rd^\sigma(f) \le 1$ and thus $\vd^\sigma(f) \ge \arity(f) - 1 > \arity(f) / 2$.
Thus we have reduced to case~\ref{case:main_tractable:vanishing_and_binary}.

\subsection{Outline of Hardness Proof for Theorem~\ref{thm:main}}

\begin{figure}[t]
 \tikzstyle{block} = [rectangle, draw, fill=blue!20, text centered, rounded corners, minimum height=2em]
 \centering
 \begin{tikzpicture}[scale=1, transform shape, node distance=2.5cm, semithick]
  \node [block, text width=5em]                     (3) {Arity 3};
  \node [block, text width=5em, right of=3]         (4) {Arity 4};
  \node [block, text width=4.5em, below right of=4] (v) {Vanishing};
  \node [block, text width=6em, below left of=v]    (s) {Theorem~\ref{thm:dic:single}};
  \node [block, text width=6em, below right of=s]   (m) {Theorem~\ref{thm:main}};
  \node [external, below right of=3] (hidden1) {};
  \node [external, above right of=hidden1] (hidden2) {};
  \node [external, below right of=hidden2] (hidden3) {};
  \node [block, text width=8em, right of=hidden3] (ap) {$\mathscr{A}$-transformable\\ and\\ $\mathscr{P}$-transformable};
  \path  (3) edge[->] (s)
         (4) edge[->] (s)
             edge[->, dashed] (v)
         (v) edge[->] (s)
        (ap) edge[->] (s)
         (s) edge[->] (m)
         (v) edge[->] (m)
        (ap) edge[->] (m);
 \end{tikzpicture}
 \caption{Dependency graph of key hardness results for our main dichotomy, Theorem~\ref{thm:main}.
 The dashed edge indicates a dependency in terms of techniques rather than the result itself.
 ``Arity 3(4)'' stands for the arity 3(4) single signature dichotomy.
 ``Vanishing'' (``$\mathscr{A}$-transformable and $\mathscr{P}$-transformable'') stands for
 the lemmas regarding vanishing ($\mathscr{A}$-transformable and $\mathscr{P}$-transformable) signatures.
 Dependencies on previous dichotomy theorems are not shown.}
 \label{fig:outline}
\end{figure}
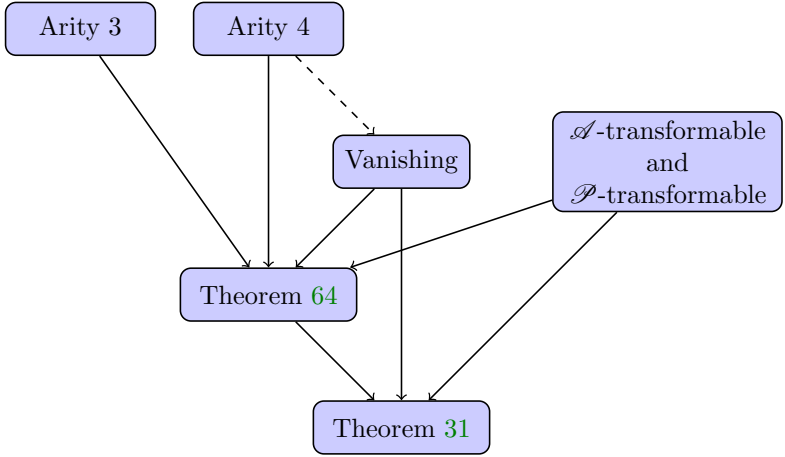

The hardness proof of our main dichotomy is more complicated.
Our first goal is to prove a dichotomy for a single signature, Theorem~\ref{thm:dic:single}.
The proof is by induction on the arity of the signature.
The induction is done by taking a self loop, which causes the arity to go down by~$2$.
Thus, we need two base cases, a dichotomy for an arity~$3$ signature and a dichotomy for an arity~$4$ signature.
The dichotomy for an arity~$3$ signature is known~\cite{CHL12},
while the dichotomy for an arity~$4$ signature is a crucial ingredient in our proof of the full dichotomy.
It is not only a base case of the single signature dichotomy but also utilized several times in the inductive step.

After obtaining the dichotomy for an arity~$4$ signature,
the proof continues by revisiting the vanishing signatures to determine what signatures combine with them to give $\SHARPP$-hardness.
When adding unary or binary signatures,
the only possible combinations that maintain the tractability of the vanishing signatures
are as described in cases~\ref{case:main_tractable:vanishing_and_binary} and~\ref{case:main_tractable:vanishing_and_unary} in Theorem~\ref{thm:main}.
Moreover, combining two vanishing signatures of 
the opposite type of arity at least~$3$ implies $\SHARPP$-hardness.
The proof of this last statement uses techniques that are similar to those in the proof of the arity~$4$ dichotomy.

Another important piece of the proof is to understand the signatures that are $\mathscr{A}$-transformable or $\mathscr{P}$-transformable.
We obtain new explicit characterizations of these signatures.
We use these characterizations to prove dichotomy theorems for any signature set containing an $\mathscr{A}$- or $\mathscr{P}$-transformable signature.
Unless every signature in the set is $\mathscr{A}$- or $\mathscr{P}$-transformable, the problem is $\SHARPP$-hard.
The proofs of these dichotomy theorems utilize the $\CSP^d$ dichotomy in~\cite{HL12}.

The main dichotomy, Theorem~\ref{thm:main},
depends on Theorem~\ref{thm:dic:single} and the results regarding vanishing signatures as well as $\mathscr{A}$- and $\mathscr{P}$-transformable signatures.
Figure~\ref{fig:outline} summarizes the dependencies among these results.

\section{Dichotomy Theorem for an Arity 4 Signature} \label{sec:arity4}
\begin{definition}
 A 4-by-4 matrix is \emph{redundant} if its middle two rows and middle two columns are the same.
 Denote the set of all redundant 4-by-4 matrices over a field $\mathbb{F}$ by $\RM_4(\mathbb{F})$.
\end{definition}

Consider the function $\varphi : \mathbb{C}^{4 \times 4} \to \mathbb{C}^{3 \times 3}$ defined by
\[\varphi(M) = A M B,\]
where
\[
 A =
 \begin{bmatrix}
  1 & 0 & 0 & 0\\
  0 & \frac{1}{2} & \frac{1}{2} & 0\\
  0 & 0 & 0 & 1
 \end{bmatrix}
 \qquad
 \text{and}
 \qquad
 B =
 \begin{bmatrix}
  1 & 0 & 0\\
  0 & 1 & 0\\
  0 & 1 & 0\\
  0 & 0 & 1
 \end{bmatrix}.
\]
Intuitively, the operation $\varphi$ replaces the middle two columns of $M$ with their sum and then the middle two rows of $M$ with their average.
(These two steps commute.)
Conversely, we have the following function $\psi : \mathbb{C}^{3 \times 3} \to \RM_4(\mathbb{C})$ defined by
\[\psi(N) = B N A.\]
Intuitively, the operation $\psi$ duplicates the middle row of $N$ and then splits the middle column evenly into two columns.
Notice that $\varphi(\psi(N)) = N$.
When restricted to $\RM_4(\mathbb{C})$, $\varphi$ is an isomorphism between the semi-group of 4-by-4 redundant matrices and the semi-group of 3-by-3 matrices,
under matrix multiplication, and $\psi$ is its inverse.
To see this, just notice that
\[
 A B =
 \begin{bmatrix}
  1 & 0 & 0\\
  0 & 1 & 0\\
  0 & 0 & 1
 \end{bmatrix}
 \qquad
 \text{and}
 \qquad
  B A =
 \begin{bmatrix}
  1 & 0 & 0 & 0\\
  0 & \frac{1}{2} & \frac{1}{2} & 0\\
  0 & \frac{1}{2} & \frac{1}{2} & 0\\
  0 & 0 & 0 & 1
 \end{bmatrix}
\]
are the identity elements of their respective semi-groups.

An example of a redundant matrix is the signature matrix of a symmetric arity~4 signature.

\begin{definition} \label{def:signature_matrix}
 The \emph{signature matrix} of a symmetric arity~4 signature $f = [f_0, f_1, f_2, f_3, f_4]$ is
 \begin{align*}
  M_f =
  \begin{bmatrix}
   f_0 & f_1 & f_1 & f_2\\
   f_1 & f_2 & f_2 & f_3\\
   f_1 & f_2 & f_2 & f_3\\
   f_2 & f_3 & f_3 & f_4
  \end{bmatrix}.
 \end{align*}
 This definition extends to an asymmetric signature $g$ as
 \begin{align*}
  M_g =
  \begin{bmatrix}
   g^{0000} & g^{0010} & g^{0001} & g^{0011}\\
   g^{0100} & g^{0110} & g^{0101} & g^{0111}\\
   g^{1000} & g^{1010} & g^{1001} & g^{1011}\\
   g^{1100} & g^{1110} & g^{1101} & g^{1111}
  \end{bmatrix}.
 \end{align*}
 When we present $g$ as an $\mathcal{F}$-gate, we order the four external edges ABCD counterclockwise.
 In $M_g$,
 the row index bits are ordered AB and the column index bits are ordered DC,
 in reverse order.
 This is for convenience so that the signature matrix of the linking of two arity~4 $\mathcal{F}$-gates is the matrix product of the signature matrices of the two $\mathcal{F}$-gates.

 If $M_g$ is redundant,
 we also define the \emph{compressed signature matrix} of $g$ as $\widetilde{M_g} = \varphi(M_g)$.
\end{definition}

If all signatures in an $\mathcal{F}$-gate have even arity,
then the $\mathcal{F}$-gate also has even arity.
Knowing that binary signatures alone do not produce $\SHARPP$-hardness,
and with the above constraint in mind,
we would like to interpolate other arity~4 signatures using the given arity~4 signature.
We are particularly interested in the signature $g$ with signature matrix
\begin{equation} \label{eqn:compressed_identity_matrix}
 M_g =
 \begin{bmatrix}
  1 & 0 & 0 & 0\\
  0 & \frac{1}{2} & \frac{1}{2} & 0\\
  0 & \frac{1}{2} & \frac{1}{2} & 0\\
  0 & 0 & 0 & 1
 \end{bmatrix},
\end{equation}
the identity element in the semi-group of redundant matrices.
Thus, $\widetilde{M_g} = I_3$.
Lemma~\ref{lem:arity4:avg_sig_hard} shows that the Holant problem with this signature is $\SHARPP$-hard.
In Lemma~\ref{lem:arity4:get_avg_sig}, we consider when we can interpolate it.

There are three cases in Lemma~\ref{lem:arity4:get_avg_sig} and one of them requires the following technical lemma.

\begin{lemma} \label{lem:two_vandermondes_nonsingular}
 Let $M = [B_0\ B_1\ \cdots\ B_t]$ be an $n$-by-$n$ block matrix such that there exists a $\lambda \in \mathbb{C}$, for all integers $0 \le k \le t$,
 block $B_k$ is an $n$-by-$c_k$ matrix for some integer $c_k \ge 0$, and the entry of $B_k$ at row $r$ and column $c$ is $(B_k)_{r c} = r^{c-1} \lambda^{k r}$, where $r, c \ge 1$.
 If $\lambda$ is nonzero and is not a root of unity, then $M$ is nonsingular.
\end{lemma}

\begin{proof}
 We prove by induction on $n$.
 If $n = 1$, then the sole entry is $\lambda^k$ for some nonnegative integer $k$.
 This is nonzero since $\lambda \ne 0$.
 Assume $n > 1$ and let the left-most nonempty block be $B_j$.
 We divide row $r$ by $\lambda^{j r}$, which is allowed since $\lambda \ne 0$.
 This effectively changes block $B_\ell$ into a block of the form $B_{\ell - j}$.
 Thus, we have another matrix of the same form as $M$ but with a nonempty block $B_0$.
 To simplify notation, we also denote this matrix again by $M$.
 The first column of $B_0$ is all 1's.
 We subtract row $r-1$ from row $r$, for $r$ from $n$ down to 2.
 This gives  us a new matrix $M'$, and $\det M = \det M'$.
 Then $\det M'$ is the determinant of the $(n-1)$-by-$(n-1)$ submatrix $M''$ obtained from $M'$ by removing the first row and column.
 Now we do column operations (on $M''$) to return the blocks to the proper form so that we can invoke the induction hypothesis.

 For any block $B_k'$ different from $B_0'$, we prove by induction on the number of columns in $B_k'$ that $B_k'$ can be repaired.
 In the base case, the $r$th element of the first column is $(B_k')_{r 1} = \lambda^{k r} - \lambda^{k (r-1)} = \lambda^{k (r-1)} (\lambda^{k} - 1)$ for $r \ge 2$.
 We divide this column by $\lambda^k - 1$ to obtain $\lambda^{k (r-1)}$, which is allowed since $\lambda$ is not a root of unity and $k \ne 0$.
 This is now the correct form for the $r$th element of the first column of a block in $M''$.

 Now for the inductive step, assume that the first $d-1$ columns of block $B_k'$ are in the correct form to be a block in $M''$.
 That is, for row index $r \ge 2$, which denotes the $(r-1)$-th row of $M''$, the $r$th element in the first $d-1$ columns of $B_k'$ have the form $(B_k')_{r c} = (r-1)^{c-1} \lambda^{k (r-1)}$.
 The $r$th element in column $d$ of $B_k'$ currently has the form $(B_k')_{r d} = r^{d-1} \lambda^{k r} - (r-1)^{d-1} \lambda^{k (r-1)}$.
 Then we do column operations
 \begin{align*}
  (B_k')_{r d} - \sum_{c=1}^{d-1} \binom{d-1}{c-1} (B_k')_{r c}
  =& r^{d-1} \lambda^{k r} - (r-1)^{d-1} \lambda^{k (r-1)} \\
  &- \sum_{c=1}^{d-1} \binom{d-1}{c-1} (r-1)^{c-1} \lambda^{k (r-1)}\\
  =& r^{d-1} \lambda^{k r} - r^{d-1} \lambda^{k (r-1)}\\
  =& r^{d-1} \lambda^{k (r-1)} (\lambda^k - 1)
 \end{align*}
 and divide by $(\lambda^k - 1)$ to get $r^{d-1} \lambda^{k (r-1)}$.
 Once again, this is allowed since $\lambda$ is not a root of unity and $k \ne 0$.
 Then more (of the same) column operations yield
 \begin{align*}
   & r^{d-1} \lambda^{k (r-1)} - \sum_{c=1}^{d-1} \binom{d-1}{c-1} (r-1)^{c-1} \lambda^{k (r-1)}\\
  =\, & \lambda^{k (r-1)} \left(r^{d-1} + (r-1)^{d-1} - \sum_{c=1}^{d} \binom{d-1}{c-1} (r-1)^{c-1}\right)
 \end{align*}
 and the term in parentheses is precisely $(r-1)^{d-1}$.
 This gives the correct form for the $r$th element in column $d$ of $B_k'$ in $M''$.

 Now we repair the columns in $B_0'$, also by induction on the number of columns.
 In the base case, if $B_0'$ only has one column, then there is nothing to prove, since this block has disappeared in $M''$.
 Otherwise, $(B_0')_{r 2} = r - (r - 1) = 1$, so the second column is already in the correct form to be the first column in $M''$, and there is still nothing to prove.
 For the inductive step, assume that columns 2 to $d-1$ are in the correct form to be the first block in $M''$ for $d \ge 3$.
 That is, the entry at row $r \ge 2$ and column $c$ from 2 through $d-1$ has the form $(B_0')_{r c} = (r-1)^{c-2}$.
 The $r$th element in column $d$ currently has the form $(B_0')_{r d} = r^{d-1} - (r-1)^{d-1}$.
 Then we do the column operations
 \begin{align*}
  (B_0')_{r d} - \sum_{c=2}^{d-1} \binom{d-1}{c-2} (B_0')_{r c}
  &= r^{d-1} - (r-1)^{d-1} - \sum_{c=2}^{d-1} \binom{d-1}{c-2} (r-1)^{c-2}\\
  &= (d-1) (r-1)^{d-2}
 \end{align*}
 and divide by $d-1$, which is nonzero, to get $(r-1)^{d-2}$.
 This is the correct form for the $r$th element in column $d$ of $B_0'$ in $M''$.
 Therefore, we invoke our original induction hypothesis that the $(n-1)$-by-$(n-1)$ matrix $M''$ has a nonzero determinant, which completes the proof.
\end{proof}

\begin{figure}[t]
 \centering
 \captionsetup[subfigure]{labelformat=empty}
 \subfloat[$N_1$]{
  \begin{tikzpicture}[scale=\scale,transform shape,node distance=\nodeDist,semithick]
   \node[external] (0)                    {};
   \node[external] (1) [right       of=0] {};
   \node[internal] (2) [below right of=1] {};
   \node[external] (3) [below left  of=2] {};
   \node[external] (4) [left        of=3] {};
   \node[external] (5) [above right of=2] {};
   \node[external] (6) [right       of=5] {};
   \node[external] (7) [below right of=2] {};
   \node[external] (8) [right       of=7] {};
   \path (0) edge[out=   0, in=135, postaction={decorate, decoration={
                                                           markings,
                                                           mark=at position 0.4   with {\arrow[>=diamond, white] {>}; },
                                                           mark=at position 0.4   with {\arrow[>=open diamond]   {>}; },
                                                           mark=at position 0.999 with {\arrow[>=diamond, white] {>}; },
                                                           mark=at position 1.0   with {\arrow[>=open diamond]   {>}; } } }] (2)
         (2) edge[out=-135, in=  0] (4)
             edge[out=  45, in=180] (6)
             edge[out= -45, in=180] (8);
   \begin{pgfonlayer}{background}
    \node[inner sep=0pt,transform shape=false,draw=\borderColor,thick,rounded corners,fit = (1) (3) (5) (7)] {};
   \end{pgfonlayer}
  \end{tikzpicture}}
 \qquad
 \subfloat[$N_2$]{
  \begin{tikzpicture}[scale=\scale,transform shape,node distance=\nodeDist,semithick]
   \node[external]  (0)                    {};
   \node[external]  (1) [right       of=0] {};
   \node[internal]  (2) [below right of=1] {};
   \node[external]  (3) [below left  of=2] {};
   \node[external]  (4) [left        of=3] {};
   \node[external]  (5) [right       of=2] {};
   \node[internal]  (6) [right       of=5] {};
   \node[external]  (7) [above right of=6] {};
   \node[external]  (8) [right       of=7] {};
   \node[external]  (9) [below right of=6] {};
   \node[external] (10) [right       of=9] {};
   \path (0) edge[out=   0, in=135, postaction={decorate, decoration={
                                                           markings,
                                                           mark=at position 0.4   with {\arrow[>=diamond, white] {>}; },
                                                           mark=at position 0.4   with {\arrow[>=open diamond]   {>}; },
                                                           mark=at position 0.999 with {\arrow[>=diamond, white] {>}; },
                                                           mark=at position 1.0   with {\arrow[>=open diamond]   {>}; } } }] (2)
         (2) edge[out=-135, in=  0]  (4)
             edge[bend left,        postaction={decorate, decoration={
                                                           markings,
                                                           mark=at position 0.999 with {\arrow[>=diamond, white] {>}; },
                                                           mark=at position 1.0   with {\arrow[>=open diamond]   {>}; } } }] (6)
             edge[bend right]        (6)
         (6) edge[out=  45, in=180]  (8)
             edge[out= -45, in=180] (10);
   \begin{pgfonlayer}{background}
    \node[inner sep=0pt,transform shape=false,draw=\borderColor,thick,rounded corners,fit = (1) (3) (7) (9)] {};
   \end{pgfonlayer}
  \end{tikzpicture}}
 \qquad
 \subfloat[$N_{s+1}$]{
  \begin{tikzpicture}[scale=\scale,transform shape,node distance=\nodeDist,semithick]
   \node[external]  (0)                     {};
   \node[external]  (1) [above left  of=0]  {};
   \node[external]  (2) [below left  of=0]  {};
   \node[external]  (3) [below left  of=1]  {};
   \node[external]  (4) [below left  of=3]  {};
   \node[external]  (5) [above left  of=3]  {};
   \node[external]  (6) [left        of=4]  {};
   \node[external]  (7) [left        of=5]  {};
   \node[external]  (8) [right       of=0]  {};
   \node[internal]  (9) [right       of=8]  {};
   \node[external] (10) [above right of=9]  {};
   \node[external] (11) [below right of=9]  {};
   \node[external] (12) [right       of=10] {};
   \node[external] (13) [right       of=11] {};
   \path let
          \p1 = (1),
          \p2 = (2)
         in
          node[external] at (\x1, \y1 / 2 + \y2 / 2) {\Huge $N_s$};
   \path let
          \p1 = (0)
         in
          node[external] (14) at (\x1 + 2, \y1 + 10) {};
   \path let
          \p1 = (0)
         in
          node[external] (15) at (\x1 + 2, \y1 - 10) {};
   \path let
          \p1 = (3)
         in
          node[external] (16) at (\x1 - 2, \y1 + 10) {};
   \path let
          \p1 = (3)
         in
          node[external] (17) at (\x1 - 2, \y1 - 10) {};
   \path (7) edge[out=   0, in=135, postaction={decorate, decoration={
                                                           markings,
                                                           mark=at position 0.43  with {\arrow[>=diamond, white] {>}; },
                                                           mark=at position 0.43  with {\arrow[>=open diamond]   {>}; },
                                                           mark=at position 0.999 with {\arrow[>=diamond, white] {>}; },
                                                           mark=at position 1.0   with {\arrow[>=open diamond]   {>}; } } }] (16)
        (17) edge[out=-135, in=  0]  (6)
        (14) edge[out=  35, in=135, postaction={decorate, decoration={
                                                           markings,
                                                           mark=at position 0.99 with {\arrow[>=diamond, white] {>}; },
                                                           mark=at position 1.0   with {\arrow[>=open diamond]   {>}; } } }] (9)
         (9) edge[out=-135, in=-35] (15)
             edge[out=  45, in=180] (12)
             edge[out= -45, in=180] (13);
   \begin{pgfonlayer}{background}
    \node[inner sep=0pt,transform shape=false,draw=\borderColor,thick,densely dashed,rounded corners,fit = (0) (1.south) (2.north) (3)] {};
    \node[inner sep=0pt,transform shape=false,draw=\borderColor,thick,rounded corners,fit = (4) (5) (10) (11)] {};
   \end{pgfonlayer}
  \end{tikzpicture}}
 \caption{Recursive construction to interpolate $g$. Vertices are assigned $f$.}
 \label{fig:gadget:arity4:interpolate_I3}
\end{figure}
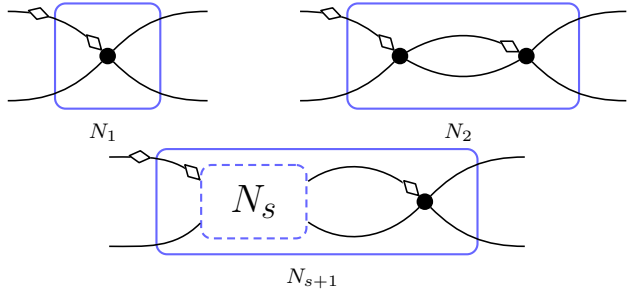

\begin{lemma} \label{lem:arity4:get_avg_sig}
 Let $g$ be the arity~4 signature with $M_g$ given in~(\ref{eqn:compressed_identity_matrix}) and let $f$ be an arity~4 signature with complex weights.
 If $M_f$ is redundant and $\widetilde{M_f}$ is nonsingular,
 then for any set $\mathcal{F}$ containing $f$,
 we have \[\Holant(\mathcal{F} \union \{g\}) \le_T \Holant(\mathcal{F}).\]
\end{lemma}

\begin{proof}
 Consider an instance $\Omega$ of $\Holant(\mathcal{F} \union \{g\})$.
 Suppose that $g$ appears $n$ times in $\Omega$.
 We construct from $\Omega$ a sequence of instances $\Omega_s$ of $\Holant(\mathcal{F})$ indexed by $s \ge 1$.
 We obtain $\Omega_s$ from $\Omega$ by replacing each occurrence of $g$
 with the gadget $N_s$ in Figure~\ref{fig:gadget:arity4:interpolate_I3} with $f$ assigned to all vertices.
 In $\Omega_s$,
 the edge corresponding to the $i$th significant index bit of $N_s$ connects to
 the same location as the edge corresponding to the $i$th significant index bit of $g$ in $\Omega$.

 Now to determine the relationship between $\Holant_\Omega$ and $\Holant_{\Omega_s}$, we use the isomorphism between redundant 4-by-4 matrices and 3-by-3 matrices.
 To obtain $\Omega_s$ from $\Omega$, we effectively replace $M_g$ with $M_{N_s} = (M_f)^s$, the $s$th power of the signature matrix $M_f$.
 By the Jordan normal form of $\widetilde{M_f}$, there exist $T, \Lambda \in \mathbb{C}^{3 \times 3}$ such that
 \begin{align*}
  \widetilde{M_f} = T \Lambda T^{-1} = T \begin{bmatrix} \lambda_1 & b_1 & 0 \\ 0 & \lambda_2 & b_2 \\ 0 & 0 & \lambda_3 \end{bmatrix} T^{-1},
 \end{align*}
 where $b_1, b_2 \in \{0,1\}$.
 Note that $\lambda_1 \lambda_2 \lambda_3 = \det(\widetilde{M_f}) \ne 0$.
 Also since $\widetilde{M_g} = \varphi(M_g) = I_3$, and $T I_3 T^{-1} = I_3$, we have $\psi(T) M_g \psi(T^{-1}) = M_g$.
 We can view our construction of $\Omega_s$ as first replacing each $M_g$ by $\psi(T) M_g \psi(T^{-1})$,
 which does not change the Holant value, and then replacing each new $M_g$ with $\psi(\Lambda^s) = \psi(\Lambda)^s$ to obtain $\Omega_s$.
 Observe that
 \[
  \varphi( \psi(T) \psi(\Lambda^s) \psi(T^{-1}) )
  = T \Lambda^s T^{-1}
  = (\widetilde{M_f})^s
  = (\varphi( M_f))^s
  = \varphi ((M_f)^s)
  = \varphi (M_{N_s}),
 \]
 hence, $\psi(T) \psi(\Lambda^s) \psi(T^{-1}) = M_{N_s}$.
 (Since $M_g = \psi(T) M_g \psi(T^{-1})$ and $M_{N_s} = \psi(T)$ $\psi(\Lambda^s) \psi(T^{-1})$,
 replacing each $M_g$, sandwiched between $\psi(T)$ and $\psi(T^{-1})$, by $\psi(\Lambda^s)$ indeed transforms $\Omega$ to $\Omega_s$.
 We also note that, by the isomorphism, $\psi(T^{-1})$ \emph{is} the multiplicative inverse of $\psi(T)$ within the semi-group of redundant 4-by-4 matrices;
 but we prefer not to write it as $\psi(T)^{-1}$ since it is not the usual matrix inverse as a 4-by-4 matrix.
 Indeed, $\psi(T)$ is not invertible as a 4-by-4 matrix.)

 In the case analysis below, we stratify the assignments in $\Omega_s$ based on the assignment to $\psi(\Lambda^s)$.
 The inputs to $\psi(\Lambda^s)$ are from $\{0, 1\}^2 \times \{0, 1\}^2$.
 However, we can combine the inputs $01$ and $10$, since $\psi(\Lambda^s)$ is redundant.
 Thus we actually stratify the assignments in $\Omega_s$ based on the assignment to $\Lambda^s$, which takes inputs from $\{0, 1, 2\} \times \{0, 1, 2\}$.
 In this compressed form,
 the row and column assignments to $\Lambda^s$ are the Hamming weight of the two actual binary valued inputs to the uncompressed form $\psi(\Lambda^s)$.

 Now we begin the case analysis on the values of $b_1$ and $b_2$.

 \begin{enumerate}
  \item Assume $b_1 = b_2 = 0$.
   In this case,
   \begin{align*}
     \psi(\Lambda^s) =
     \psi\left(      
     \begin{bmatrix}
       \lambda_1^s & 0 & 0 \\
       0 & \lambda_2^s & 0 \\
       0 & 0 & \lambda_3^s
     \end{bmatrix} \right)
     =
     \begin{bmatrix}
       \lambda_1^s & 0 & 0 & 0 \\
       0 & \lambda_2^s/2 & \lambda_2^s/2 & 0 \\
       0 & \lambda_2^s/2 & \lambda_2^s/2 & 0 \\
       0 & 0 & 0 & \lambda_3^s
     \end{bmatrix}.
   \end{align*}
   We only need to consider the assignments to $\Lambda^s$ that assign
   \begin{itemize}
    \item $(0,0)$ $i$ many times,
    \item $(1,1)$ $j$ many times, and
    \item $(2,2)$ $k$ many times
   \end{itemize}
  since any other assignment contributes a factor of 0.
  In particular, the $(1,1)$ case actually corresponds to the middle four entries in $\psi(\Lambda^s)$.
  We collect them together as they contribute the same factor.
  Let $c_{ijk}$ be the sum over all such assignments of the products of evaluations of
  all signatures in $\Omega_s$ except for $\Lambda^s$ (including the contributions from $T$ and $T^{-1}$).
  Note that this quantity is the same in $\Omega$ as in $\Omega_s$.
  In particular it does not depend on $s$.
  Then
  \begin{align*}
   \Holant_{\Omega}
   = \sum_{i + j + k = n} \frac{c_{i j k}}{2^j}.
  \end{align*}
  Note that the factor of $\frac{1}{2^j}$ comes from \eqref{eqn:compressed_identity_matrix}, the definition of $g$.
  The value of the Holant on $\Omega_s$, for $s \ge 1$, is
  \begin{align*}
   \Holant_{\Omega_s}
   = \sum_{i + j + k = n} \left(\lambda_1^i \lambda_2^j \lambda_3^k\right)^s \left(\frac{c_{i j k}}{2^j}\right).
  \end{align*}
  The coefficient matrix is Vandermonde,
  but it may not have full rank because it might be that $\lambda_1^i \lambda_2^j \lambda_3^k = \lambda_1^{i'} \lambda_2^{j'} \lambda_3^{k'}$ for some $(i,j,k) \ne (i',j',k')$,
  where $i + j + k = i' + j' + k' = n$.
  However, this is not a problem since we are only interested in the sum $\sum \frac{c_{i j k}}{2^j}$.
  If two coefficients are the same, we replace their corresponding unknowns $c_{i j k} / 2^j$ and $c_{i' j' k'} / 2^{j'}$ with their sum as a new variable.
  After all such combinations, we have a Vandermonde system of full rank.
  In particular,
  none of the entries are~$0$ since $\lambda_1 \lambda_2 \lambda_3 = \det(\widetilde{M_f}) \ne 0$.
  Therefore, we can solve the linear system and obtain the value of $\Holant_\Omega$.

 \item Assume $b_1 \ne b_2$.
  We can permute the Jordan blocks in $\Lambda$ so that $b_1 = 1$ and $b_2 = 0$, then $\lambda_1 = \lambda_2$, denoted by $\lambda$.
   In this case,
   \begin{align*}
     \psi(\Lambda^s) =
     \psi\left(      
     \begin{bmatrix}
       \lambda^s & s\lambda^{s-1} & 0 \\
       0 & \lambda^s & 0 \\
       0 & 0 & \lambda_3^s
     \end{bmatrix} \right)
     =
     \begin{bmatrix}
       \lambda^s & s\lambda^{s-1}/2 & s\lambda^{s-1}/2 & 0 \\
       0 & \lambda^s/2 & \lambda^s/2 & 0 \\
       0 & \lambda^s/2 & \lambda^s/2 & 0 \\
       0 & 0 & 0 & \lambda_3^s
     \end{bmatrix}.
   \end{align*}  
  We only need to consider the assignments to $\Lambda^s$ that assign
  \begin{itemize}
   \item $(0,0)$ $i$ many times,
   \item $(1,1)$ $j$ many times,
   \item $(2,2)$ $k$ many times, and
   \item $(0,1)$ $\ell$ many times
  \end{itemize}
  since any other assignment contributes a factor of 0.
  Let $c_{i j k \ell}$ be the sum over all such assignments of the products of evaluations of
  all signatures in $\Omega_s$ except for $\Lambda^s$ (including the contributions from $T$ and $T^{-1}$).
  Then
  \begin{align*}
   \Holant_{\Omega}
   = \sum_{i + j + k = n} \frac{c_{i j k 0}}{2^j}
  \end{align*}
  and the value of the Holant on $\Omega_s$, for $s \ge 1$, is
  \begin{align*}
   \Holant_{\Omega_s}
   &= \sum_{i + j + k + \ell = n} \lambda^{(i+j) s} \lambda_3^{k s} \left(s \lambda^{s - 1}\right)^\ell \left(\frac{c_{i j k \ell}}{2^{j+\ell}}\right)\\
   &= \lambda^{n s} \sum_{i + j + k + \ell = n} \left(\frac{\lambda_3}{\lambda}\right)^{k s} s^\ell \left(\frac{c_{i j k \ell}}{\lambda^\ell 2^{j+\ell}}\right).
  \end{align*}

  If $\lambda_3 / \lambda$ is a root of unity, then take a $t$ such that $(\lambda_3 / \lambda)^t = 1$.
  Then
  \begin{align*}
   \Holant_{\Omega_{s t}}
   &= \lambda^{n s t} \sum_{i + j + k + \ell = n} s^\ell \left(\frac{t^\ell c_{i j k \ell}}{\lambda^\ell 2^{j+\ell}}\right).
  \end{align*}
  For $s \ge 1$, this gives a coefficient matrix that is Vandermonde.
  Although this system is not full rank,
  we can replace all the unknowns $c_{i j k \ell} / 2^j$ having $i + j + k = n - \ell$
  by their sum to form new unknowns $c_\ell' = \sum_{i + j + k = n- \ell} \frac{c_{i j k \ell}}{2^j}$,
  where $0 \le \ell \le n$.
  The new unknown $c_0'$ is the Holant of $\Omega$ that we seek.
  The resulting Vandermonde system
  \begin{align*}
   \Holant_{\Omega_{s t}}
   &= \lambda^{n s t} \sum_{\ell = 0}^n s^\ell \left(\frac{t^\ell c_\ell'}{\lambda^\ell 2^\ell}\right)
  \end{align*}
  has full rank, so we can solve for the new unknowns and obtain the value of $\Holant_\Omega = c_0'$.

  If $\lambda_3 / \lambda$ is not a root of unity,
  then we replace all the unknowns $c_{i j k \ell} / (\lambda^\ell 2^{j+\ell})$ having $i + j = m$ with their sum to form new unknowns $c_{m k \ell}'$,
  for any $0 \le m, k, \ell$ and $m + k + \ell = n$.
  The Holant of $\Omega$ is now
  \begin{align*}
   \Holant_{\Omega}
   = \sum_{m + k = n} c_{m k 0}'
  \end{align*}
  and the value of the Holant on $\Omega_s$ is
  \begin{align*}
   \Holant_{\Omega_s}
   &= \lambda^{n s} \sum_{i + j + k + \ell = n} \left(\frac{\lambda_3}{\lambda}\right)^{k s} s^\ell \left(\frac{c_{i j k \ell}}{\lambda^\ell 2^{j+\ell}}\right)\\
   &= \lambda^{n s} \sum_{m + k + \ell = n} \left(\frac{\lambda_3}{\lambda}\right)^{k s} s^\ell c_{m k \ell}'.
  \end{align*}
  After a suitable reordering of the columns, the matrix of coefficients satisfies the hypothesis of Lemma~\ref{lem:two_vandermondes_nonsingular}.
  Therefore, the linear system has full rank.
  We can solve for the unknowns and obtain the value of $\Holant_\Omega$.

 \item Assume $b_1 = b_2 = 1$, and therefore $\lambda_1 = \lambda_2 = \lambda_3$, denoted by $\lambda$.
   In this case,
   \begin{align*}
     \psi(\Lambda^s) & =
     \psi\left(      
     \begin{bmatrix}
       \lambda^s & s\lambda^{s-1} & s(s-1)\lambda^{s-2}/2 \\
       0 & \lambda^s & s\lambda^{s-1} \\
       0 & 0 & \lambda^s
     \end{bmatrix} \right)\\
     & =
     \begin{bmatrix}
       \lambda^s & s\lambda^{s-1}/2 & s\lambda^{s-1}/2 & s(s-1)\lambda^{s-2}/2 \\
       0 & \lambda^s/2 & \lambda^s/2 & s\lambda^{s-1} \\
       0 & \lambda^s/2 & \lambda^s/2 & s\lambda^{s-1} \\
       0 & 0 & 0 & \lambda^s
     \end{bmatrix}.
   \end{align*}    
  We only need to consider the assignments to $\Lambda^s$ that assign
  \begin{itemize}
   \item $(0,0)$ or $(2,2)$ $i$ many times,
   \item $(1,1)$ $j$ many times,
   \item $(0,1)$ $k$ many times,
   \item $(1,2)$ $\ell$ many times, and
   \item $(0,2)$ $m$ many times
   \end{itemize}
  since any other assignment contributes a factor of 0.
  Let $c_{i j k \ell m}$ be the sum over all such assignments of the products of evaluations of
  all signatures in $\Omega_s$ except for $\Lambda^s$ (including the contributions from $T$ and $T^{-1}$).
  Then
  \begin{align*}
   \Holant_{\Omega}
   = \sum_{i + j = n} \frac{c_{i j 0 0 0}}{2^j}
  \end{align*}
  and the value of the Holant on $\Omega_s$, for $s \ge 1$, is
  \begin{align*}
   \Holant_{\Omega_s}
   &= \sum_{i + j + k + \ell + m = n} \lambda^{(i + j) s} \left(s \lambda^{s - 1}\right)^{k + \ell} \left(s (s - 1) \lambda^{s - 2}\right)^m \left(\frac{c_{i j k \ell m}}{2^{j+k+m}}\right)\\
   &= \lambda^{n s} \sum_{i + j + k + \ell + m = n} s^{k + \ell + m} (s - 1)^m \left(\frac{c_{i j k \ell m}}{\lambda^{k+\ell+2m} 2^{j+k+m}}\right).
  \end{align*}
  We replace all the unknowns $c_{i j k \ell m} / (\lambda^{k+\ell+2m} 2^{j+k+m})$ having $i + j = p$ and $k + \ell = q$ with their sum to form new unknowns $c_{p q m}'$,
  for any $0 \le p, q, m$ and $p + q + m = n$.
  The Holant of $\Omega$ is now $c_{n 0 0}'$.
  This new linear system is
  \begin{align*}
   \Holant_{\Omega_s}
   = \lambda^{n s} \sum_{p + q + m = n} s^{q + m} (s - 1)^m c_{p q m}'
  \end{align*}
  but is still rank deficient.
  We now index the columns by $(q, m)$, where $q \ge 0$, $m \ge 0$, and $q + m \le n$.
  Correspondingly, we rename the variables $x_{q,m} = c'_{p q m}$.
  Note that $p = n - q - m$ is determined by $(q, m)$.
  Observe that the column indexed by $(q,m)$ is the  sum of the columns indexed by $(q-1, m)$ and $(q-2, m+1)$ provided $q-2 \ge 0$.
  Namely, $s^{q+m} (s-1)^m = s^{q-1+m} (s-1)^m + s^{q-2+m+1} (s-1)^{m+1}$.
  Of course this is only meaningful if $q \ge 2$, $m \ge 0$ and $q + m \le n$.
  We write the linear system as
  \[\sum_{q \ge 0,\ m \ge 0,\ q + m \le n} \alpha_{q,m} x_{q,m} = \frac{\Holant_{\Omega_s}}{\lambda^{n s}},\]
  where $\alpha_{q,m} = s^{q+m} (s-1)^m$ are the coefficients.
  Hence $\alpha_{q,m} x_{q,m} = \alpha_{q-1,m} x_{q,m} + \alpha_{q-2,m+1} x_{q,m}$, and we define new variables
  \begin{align*}
   x_{q-1, m}   &\leftarrow x_{q, m} + x_{q-1, m}\\
   x_{q-2, m+1} &\leftarrow x_{q, m} + x_{q-2, m+1}
  \end{align*}
  from $q = n-m$ down to $2$ for every $0\leq m\leq n-2$.

  Observe that in each update, the newly defined variables have a decreased index value for $q$.
  A more crucial observation is that the column indexed by $(0, 0)$ is never updated.
  This is because, in order to be an updated entry, there must be some $q \ge 2$ and $m \ge 0$ such that $(q-1, m) = (0,0)$ or $(q-2, m+1) = (0,0)$, which is clearly impossible.
  Hence $x_{0,0} = c'_{n 0 0}$ is still the Holant value on $\Omega$.
  The $2 n + 1$ unknowns that remain are
  \[x_{0,0},\ x_{1,0},\ x_{0,1},\ x_{1,1},\ x_{0,2},\ x_{1,2},\ \dotsc,\ x_{0,n-1},\ x_{1,n-1},\ x_{0,n}\]
  and their coefficients in row $s$ are
  \[1, s, s (s - 1), s^2 (s - 1), s^2 (s - 1)^2, \dotsc, s^{n-1} (s - 1)^{n-1}, s^n (s - 1)^{n-1}, s^n (s - 1)^n.\]
  It is clear that the $\kappa$-th entry in this row is a monic polynomial in $s$ of degree $\kappa$, where $0 \le \kappa \le 2 n$,
  and thus $s^{\kappa}$ is a linear combination of the first $\kappa$ entries.
  It follows that the coefficient matrix is a product of the standard Vandermonde matrix multiplied to its right by an upper triangular matrix with all 1's on the diagonal.
  Therefore, the linear system has full rank.
  We can solve for these final unknowns and obtain the value of $\Holant_\Omega = x_{0,0} = c_{n00}'$.
 \end{enumerate}
\end{proof}

For an asymmetric signature, we often want to reorder the input bits under a circular permutation.
For a single counterclockwise rotation by $90^\circ$,
the effect on the entries of the signature matrix of an arity~4 signature is given in Figure~\ref{fig:rotate_asymmetric_signature}.

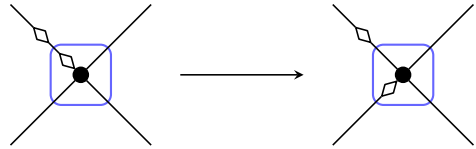
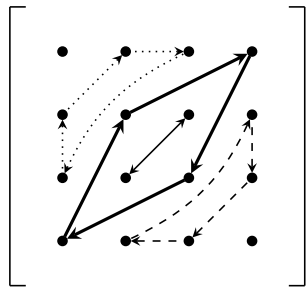
\begin{figure}[b!]
 \centering
 \def\capWidth{6cm}
 \captionsetup[subfigure]{width=\capWidth}
 \tikzstyle{entry} = [internal, inner sep=2pt]
 \subfloat[A counterclockwise rotation]{
  \begin{tikzpicture}[scale=\scale,transform shape,node distance=1.7 * \nodeDist,semithick]
   \node[internal]  (0)                    {};
   \node[external]  (1) [above  left of=0] {};
   \node[external]  (2) [above right of=0] {};
   \node[external]  (3) [below  left of=0] {};
   \node[external]  (4) [below right of=0] {};
   \node[external]  (5) [      right of=0] {};
   \node[external]  (6) [      right of=5] {};
   \node[internal]  (7) [      right of=6] {};
   \node[external]  (8) [above  left of=7] {};
   \node[external]  (9) [above right of=7] {};
   \node[external] (10) [below  left of=7] {};
   \node[external] (11) [below right of=7] {};
   \path (0) edge[postaction={decorate, decoration={
                                         markings,
                                         mark=at position 0.25 with {\arrow[>=diamond,white] {>}; },
                                         mark=at position 0.25 with {\arrow[>=open diamond]  {>}; },
                                         mark=at position 0.65 with {\arrow[>=diamond,white] {>}; },
                                         mark=at position 0.65 with {\arrow[>=open diamond]  {>}; } } }] (1)
             edge (2)
             edge (3)
             edge (4)
    (5.west) edge[->, >=stealth] (6.east)
         (7) edge[postaction={decorate, decoration={
                                         markings,
                                         mark=at position 0.65 with {\arrow[>=diamond,white] {>}; },
                                         mark=at position 0.65 with {\arrow[>=open diamond]  {>}; } } }] (8)
             edge (9)
             edge [postaction={decorate, decoration={
                                         markings,
                                         mark=at position 0.25 with {\arrow[>=diamond,white] {>}; },
                                         mark=at position 0.25 with {\arrow[>=open diamond]  {>}; } } }] (10)
             edge (11);
   \begin{pgfonlayer}{background}
    \node[draw=\borderColor,thick,rounded corners,fit = (0),inner sep=16pt] {};
    \node[draw=\borderColor,thick,rounded corners,fit = (7),inner sep=16pt] {};
   \end{pgfonlayer}
  \end{tikzpicture}}
 \qquad
 \subfloat[Movement of signature matrix entries]{
  \makebox[\capWidth][c]{
   \begin{tikzpicture}[scale=\scale,transform shape,>=stealth,node distance=\nodeDist,semithick]
    \node[entry] (11)               {};
    \node[entry] (12) [right of=11] {};
    \node[entry] (13) [right of=12] {};
    \node[entry] (14) [right of=13] {};
    \node[entry] (21) [below of=11] {};
    \node[entry] (22) [right of=21] {};
    \node[entry] (23) [right of=22] {};
    \node[entry] (24) [right of=23] {};
    \node[entry] (31) [below of=21] {};
    \node[entry] (32) [right of=31] {};
    \node[entry] (33) [right of=32] {};
    \node[entry] (34) [right of=33] {};
    \node[entry] (41) [below of=31] {};
    \node[entry] (42) [right of=41] {};
    \node[entry] (43) [right of=42] {};
    \node[entry] (44) [right of=43] {};
    \node[external] (nw) [above left  of=11] {};
    \node[external] (ne) [above right of=14] {};
    \node[external] (sw) [below left  of=41] {};
    \node[external] (se) [below right of=44] {};
    \path (13) edge[<-, dotted]                (12)
          (12) edge[<-, dotted]                (21)
          (21) edge[<-, dotted]                (31)
          (31) edge[<-, dotted,out=65,in=-155] (13)
          (42) edge[<-, dashed]                (43)
          (43) edge[<-, dashed]                (34)
          (34) edge[<-, dashed]                (24)
          (24) edge[<-, dashed,out=-115,in=25] (42)
          (14) edge[<-, very thick]            (22)
          (22) edge[<-, very thick]            (41)
          (41) edge[<-, very thick]            (33)
          (33) edge[<-, very thick]            (14)
          (23) edge[<->]                      (32);
    \path (nw.west) edge (sw.west)
          (ne.east) edge (se.east)
          (nw.west) edge (nw.east)
          (sw.west) edge (sw.east)
          (ne.west) edge (ne.east)
          (se.west) edge (se.east);
   \end{tikzpicture}}}
 \caption{The movement of the entries in the signature matrix of an arity~$4$ signature under a counterclockwise rotation of the input edges.
  Entires of Hamming weight~$1$ are in the dotted cycle,
  entires of Hamming weight~$2$ are in the two solid cycles (one has length~$4$ and the other one is a swap),
  and entries of Hamming weight~$3$ are in the dashed cycle.}
 \label{fig:rotate_asymmetric_signature}
\end{figure}

We ultimately derive most of our $\SHARPP$-hardness results through Lemma~\ref{lem:arity4:avg_sig_hard}.
This is done by a reduction from the problem of counting Eulerian orientations on 4-regular graphs, 
which is the Holant problem $\holant{[0,1,0]}{[0,0,1,0,0]}$.
Recall (from Section~\ref{sec:intro}) that under a holographic transformation by $\left[\begin{smallmatrix} 1 & 1 \\ i & -i \end{smallmatrix} \right]$,
this bipartite Holant problem becomes the Holant problem $\Holant([1,0,\tfrac{1}{3},0,1])$ up to a nonzero constant factor.

\begin{theorem}[Theorem~V.10 in~\cite{HL12}] \label{thm:4reg_EO_hard}
 \textsc{Counting-Eulerian-Orientations} is $\SHARPP$-hard for 4-regular graphs.
\end{theorem}

\begin{lemma} \label{lem:arity4:avg_sig_hard}
 Let $g$ be the arity~4 signature with $M_g$ given in~(\ref{eqn:compressed_identity_matrix}) so that $\widetilde{M_g} = I_3$.
 Then $\Holant(g)$ is $\SHARPP$-hard.
\end{lemma}

\begin{proof}
 We reduce from the Eulerian orientation problem $\Holant(\mathscr{O})$, where $\mathscr{O} = [1,0,\tfrac{1}{3},0,1]$, which is $\SHARPP$-hard by Theorem~\ref{thm:4reg_EO_hard}.
 We achieve this via an arbitrarily close approximation using the recursive construction in Figure~\ref{fig:gadget:approximation} with $g$ assigned to every vertex.

 We claim that the signature matrix $ M_{N_k}$ of Gadget $N_k$ is
 \begin{align*}
  M_{N_k} =
  \begin{bmatrix}
     1 &       0 &       0 & a_k\\
     0 & a_{k+1} & a_{k+1} &   0\\
     0 & a_{k+1} & a_{k+1} &   0\\
   a_k &       0 &       0 &   1
  \end{bmatrix},
 \end{align*}
 where $a_k = \tfrac{1}{3} - \tfrac{1}{3} \left(-\tfrac{1}{2}\right)^k$.
 This is true for $N_0$.
 Inductively assume $M_{N_k}$ has this form.
 Then the rotated form of the signature matrix for $N_{k}$, as described in Figure~\ref{fig:rotate_asymmetric_signature}, is
 \begin{equation} \label{eqn:rotated-N_k}
  \begin{bmatrix}
   1       & 0       & 0       & a_{k+1}\\
   0       & a_{k}   & a_{k+1} & 0\\
   0       & a_{k+1} & a_{k}   & 0\\
   a_{k+1} & 0       & 0       & 1
  \end{bmatrix}.
 \end{equation}
 The action of $g$ on the far right side of $N_{k+1}$ is to replace each of the middle two entries in the middle two rows of the matrix in~(\ref{eqn:rotated-N_k}) with their average,
 $(a_{k} + a_{k+1})/2 = a_{k+2}$.
 This gives $M_{N_{k+1}}$.

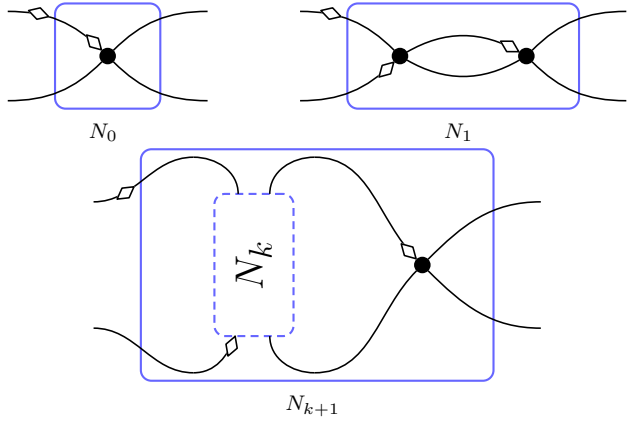
\begin{figure}[t]
 \centering
 \captionsetup[subfigure]{labelformat=empty}
 \subfloat[$N_0$]{
  \begin{tikzpicture}[scale=\scale,transform shape,node distance=\nodeDist,semithick]
   \node[external] (0)                    {};
   \node[external] (1) [right       of=0] {};
   \node[internal] (2) [below right of=1] {};
   \node[external] (3) [below left  of=2] {};
   \node[external] (4) [left        of=3] {};
   \node[external] (5) [above right of=2] {};
   \node[external] (6) [right       of=5] {};
   \node[external] (7) [below right of=2] {};
   \node[external] (8) [right       of=7] {};
    \path (0) edge[out=   0, in=135, postaction={decorate, decoration={
                                                            markings,
                                                            mark=at position 0.4   with {\arrow[>=diamond,white] {>}; },
                                                            mark=at position 0.4   with {\arrow[>=open diamond]  {>}; },
                                                            mark=at position 0.999 with {\arrow[>=diamond,white] {>}; },
                                                            mark=at position 1     with {\arrow[>=open diamond]  {>}; } } }] (2)
          (2) edge[out=-135, in=  0] (4)
              edge[out=  45, in=180] (6)
              edge[out= -45, in=180] (8);
   \begin{pgfonlayer}{background}
    \node[inner sep=0pt,transform shape=false,draw=\borderColor,thick,rounded corners,fit = (1) (3) (5) (7)] {};
   \end{pgfonlayer}
  \end{tikzpicture}}
 \qquad
 \subfloat[$N_1$]{
  \begin{tikzpicture}[scale=\scale,transform shape,node distance=\nodeDist,semithick]
   \node[external]  (0)                    {};
   \node[external]  (1) [right of=0]       {};
   \node[internal]  (2) [below right of=1] {};
   \node[external]  (3) [below left  of=2] {};
   \node[external]  (4) [left        of=3] {};
   \node[external]  (5) [right       of=2] {};
   \node[internal]  (6) [right       of=5] {};
   \node[external]  (7) [above right of=6] {};
   \node[external]  (8) [right       of=7] {};
   \node[external]  (9) [below right of=6] {};
   \node[external] (10) [right       of=9] {};
   \path (0) edge[out=  0, in= 135, postaction={decorate, decoration={
                                                           markings,
                                                           mark=at position 0.4   with {\arrow[>=diamond,white] {>}; },
                                                           mark=at position 0.4   with {\arrow[>=open diamond]  {>}; } } }] (2)
         (4) edge[out=  0, in=-135, postaction={decorate, decoration={
                                                           markings,
                                                           mark=at position 0.999 with {\arrow[>=diamond,white] {>}; },
                                                           mark=at position 1     with {\arrow[>=open diamond]  {>}; } } }] (2)
         (2) edge[bend left,        postaction={decorate, decoration={
                                                           markings,
                                                           mark=at position 0.999 with {\arrow[>=diamond,white] {>}; },
                                                           mark=at position 1     with {\arrow[>=open diamond]  {>}; } } }] (6)
             edge[bend right]        (6)
         (6) edge[out= 45, in= 180]  (8)
             edge[out=-45, in= 180] (10);
   \begin{pgfonlayer}{background}
    \node[inner sep=0pt,transform shape=false,draw=\borderColor,thick,rounded corners,fit = (1) (3) (7) (9)] {};
   \end{pgfonlayer}
  \end{tikzpicture}}
 \qquad
 \subfloat[$N_{k+1}$]{
  \begin{tikzpicture}[scale=\scale,transform shape,node distance=\nodeDist,semithick]
   \node[external] (0)              {};
   \node[external] (1) [above of=0] {};
   \node[external] (2) [below of=0] {};
   \node[external] (3) [right of=0] {};
   \node[external] (4) [above of=3] {};
   \node[external] (5) [below of=3] {};
   \path let
          \p1 = (0),
          \p2 = (3)
         in
          node[external] at (\x1 / 2 + \x2 / 2, \y1) {\Huge \begin{sideways}$N_k$\end{sideways}};
   \path let
          \p1 = (1),
          \p2 = (4)
         in
          node[external] (6) at (3 * \x1 / 4 + \x2 / 4, \y1) {};
   \path let
          \p1 = (1),
          \p2 = (4)
         in
          node[external] (7) at (\x1 / 4 + 3 * \x2 / 4, \y1) {};
   \path let
          \p1 = (2),
          \p2 = (5)
         in
          node[external] (8) at (3 * \x1 / 4 + \x2 / 4, \y1) {};
   \path let
          \p1 = (2),
          \p2 = (5)
         in
          node[external] (9) at (\x1 / 4 + 3 * \x2 / 4, \y1) {};
   \node[external] (10) [above left  of=6]  {};
   \node[external] (11) [below left  of=8]  {};
   \node[external] (12) [below left  of=10] {};
   \node[external] (13) [above left  of=11] {};
   \node[external] (14) [left        of=12] {};
   \node[external] (15) [left        of=13] {};
   \node[external] (16) [above right of=7]  {};
   \node[external] (17) [below right of=9]  {};
   \node[external] (18) [below right of=16] {};
   \node[external] (19) [above right of=17] {};
   \node[external] (n1) [right       of=18] {};
   \path let
          \p1 = (n1),
          \p2 = (18),
          \p3 = (19)
         in
          node[internal] (20) at (\x1, \y2 / 2 + \y3 / 2) {};
   \node[external] (n2) [right of=20] {};
   \node[external] (n3) [right of=n2] {};
   \path let
          \p1 = (n3),
          \p2 = (14)
         in
          node[external] (21) at (\x1, \y2) {};
   \path let
          \p1 = (n3),
          \p2 = (15)
         in
          node[external] (22) at (\x1, \y2) {};
   \path (6)         edge[out=  90, in=   0]     (10.center)
         (11.center) edge[out=   0, in=-110, postaction={decorate, decoration={
                                                                    markings,
                                                                    mark=at position 0.999 with {\arrow[>=diamond,white] {>}; },
                                                                    mark=at position 1     with {\arrow[>=open diamond]  {>}; } } }] (8)
         (14)        edge[out=   0, in= 180, postaction={decorate, decoration={
                                                                    markings,
                                                                    mark=at position 0.4   with {\arrow[>=diamond,white] {>}; },
                                                                    mark=at position 0.4   with {\arrow[>=open diamond]  {>}; } } }] (10.center)
         (11.center) edge[out= 180, in=   0]     (15)
         (7)         edge[out=  90, in= 180]     (16.center)
         (9)         edge[out= -90, in= 180]     (17.center)
         (16.center) edge[out=   0, in= 135, postaction={decorate, decoration={
                                                                    markings,
                                                                    mark=at position 0.999 with {\arrow[>=diamond,white] {>}; },
                                                                    mark=at position 0.999 with {\arrow[>=open diamond]  {>}; } } }] (20)
         (17.center) edge[out=   0, in=-135]     (20)
         (20)        edge[out=  45, in= 180]     (21)
                     edge[out= -45, in= 180]     (22);
   \begin{pgfonlayer}{background}
    \node[inner sep=0pt,transform shape=false,draw=\borderColor,thick,densely dashed,rounded corners,fit = (1) (2) (4) (5)] {};
    \node[inner sep=0pt,transform shape=false,draw=\borderColor,thick,rounded corners,fit = (10) (11) (12) (13) (16) (17) (n2)] {};
   \end{pgfonlayer}
  \end{tikzpicture}}
 \caption{Recursive construction to approximate $[1,0,\tfrac{1}{3},0,1]$. Vertices are assigned $g$.}
 \label{fig:gadget:approximation}
\end{figure}

 Let $G$ be a graph with $n$ vertices and $H_{\mathscr{O}}$ (resp.~$H_{N_k}$) be the Holant value on $G$ with all vertices assigned $\mathscr{O}$ (resp.~$N_k$).
 Since each signature entry in $\mathscr{O}$ can be expressed as a rational number with denominator 3,
 each term in the sum of $H_{\mathscr{O}}$ can be expressed as a rational number with denominator $3^n$,
 and $H_{\mathscr{O}}$ itself is a sum of $2^{2n}$ such terms, where $2 n$ is the number of edges in $G$.
 If the error $|H_{N_k} - H_{\mathscr{O}}|$ is at most $1 / 3^{n+1}$,
 then we can recover $H_{\mathscr{O}}$ from $H_{N_k}$ by selecting the nearest rational number to $H_{N_k}$ with denominator $3^n$.

 For each signature entry $x$ in $M_{\mathscr{O}}$, its corresponding entry $\tilde{x}$ in $M_{N_k}$ satisfies $|\tilde{x} - x| \le x/2^k$.
 Then for each term $t$ in the Holant sum $H_{\mathscr{O}}$,
 its corresponding term $\tilde{t}$ in the sum $H_{N_k}$ satisfies $t (1 - 1/2^k)^n \le \tilde{t} \le t (1 + 1/2^k)^n$,
 thus $-t (1 - (1 - 1/2^k)^n) \le \tilde{t} - t \le t ((1 + 1/2^k)^n - 1)$.
 Since $1 - (1 - 1/2^k)^n \le (1 + 1/2^k)^n - 1$, we get $|\tilde{t} - t| \le t ((1 + 1/2^k)^n - 1)$.
 Also each term $t \le 1$.
 Hence \[|H_{N_k} - H_{\mathscr{O}}| \le 2^{2n} ((1 + 1/2^k)^n - 1) < 1/3^{n+1},\] if we take $k = 4n$.
\end{proof}

We summarize our progress with the following corollary,
which combines Lemmas~\ref{lem:arity4:get_avg_sig} and~\ref{lem:arity4:avg_sig_hard}.

\begin{corollary} \label{cor:arity4:nonsingular_compressed_hard}
 Let $f$ be an arity~4 signature with complex weights.
 If $M_f$ is redundant and $\widetilde{M_f}$ is nonsingular,
 then $\Holant(f)$ is $\SHARPP$-hard.
\end{corollary}

In order to make Corollary~\ref{cor:arity4:nonsingular_compressed_hard} more applicable,
we show that for an arity~4 signature $f$,
the redundancy of $M_f$ and the nonsingularity of $\widetilde{M_f}$ are invariant under an invertible holographic transformation.

\begin{lemma} \label{lem:arity4:nonsingular_compressed_invariant}
 Let $f$ be an arity~4 signature with complex weights, $T \in \mathbb{C}^{2 \times 2}$ a matrix, and $\hat{f} = T^{\otimes 4} f$.
 If $M_f$ is redundant,
 then $M_{\hat{f}}$ is also redundant and $\det(\varphi(M_{\hat{f}})) = \det(\varphi(M_f)) \det(T)^6$.
\end{lemma}

\begin{proof}
 Since $\hat{f} = T^{\otimes 4} f$,
 we can express $M_{\hat{f}}$ in terms of $M_f$ and $T$ as
 \begin{equation} \label{eqn:arity4:holo_trans_on_sig_matrix}
  M_{\hat{f}} = T^{\otimes 2} M_f \left(T^\intercal\right)^{\otimes 2}.
 \end{equation}
 This can be directly checked.
 Alternatively, this relation is known (and can also be directly checked) had we not introduced the flip of the middle two columns,
 i.e., if the columns were ordered $00, 01, 10, 11$ by the last two bits in $f$ and $\hat{f}$.
 Instead, the columns are ordered by $00, 10, 01, 11$ in $M_f$ and $M_{\hat{f}}$.
 Let $T = (t^i_j)$, where row index $i$ and column index $j$ range from $\{0, 1\}$.
 Then $T^{\otimes 2} = ( t^i_j t^{i'}_{j'})$, with row index $ii'$ and column index $j j'$.
 Let
 \[
  \mathcal{E} =
  \begin{bmatrix}
   1 & 0 & 0 & 0\\
   0 & 0 & 1 & 0\\
   0 & 1 & 0 & 0\\
   0 & 0 & 0 & 1
  \end{bmatrix},
 \]
 then $\mathcal{E} T^{\otimes 2} \mathcal{E} = T^{\otimes 2}$,
 i.e.,
 a simultaneous row flip $i i' \leftrightarrow i' i$ and column flip $j j' \leftrightarrow j' j$ keep $T^{\otimes 2}$ unchanged.
 Then the known relations $M_{\hat{f}} \mathcal{E} = T^{\otimes 2} M_f \mathcal{E} \left(T^\intercal\right)^{\otimes 2}$
 and $\mathcal{E} \left(T^\intercal\right)^{\otimes 2} \mathcal{E} = \left(T^\intercal\right)^{\otimes 2}$ imply~(\ref{eqn:arity4:holo_trans_on_sig_matrix}).

 Now $X \in \RM_4(\mathbb{C})$ iff $\mathcal{E} X = X = X \mathcal{E}$.
 Then it follows that $M_{\hat{f}} \in \RM_4(\mathbb{C})$ if $M_f \in \RM_4(\mathbb{C})$.
 For the two matrices $A$ and $B$ in the definition of $\varphi$,
 we note that $B A = M_g$, where $M_g$ given in~(\ref{eqn:compressed_identity_matrix}) is the identity element of the semi-group $\RM_4(\mathbb{C})$.
 Since $M_f \in \RM_4(\mathbb{C})$, we have $BA M_f = M_f = M_f BA$.
 Then we have
 \begin{align}
  \varphi(M_{\hat{f}})
  &= A M_{\hat{f}} B
   = A \left(T^{\otimes 2} M_f \left(T^\intercal\right)^{\otimes 2}\right) B \notag\\
  &= (A T^{\otimes 2} B) (A M_f B) (A \left(T^\intercal\right)^{\otimes 2} B) \label{eqn:rdt-trans}\\
  &= \varphi(T^{\otimes 2})\varphi(M_f)\varphi(\left(T^\intercal\right)^{\otimes 2}). \notag
 \end{align}
 Another direct calculation shows that
 \[\det(\varphi(T^{\otimes 2})) = \det(T)^3 = \det(\varphi(\left(T^\intercal\right)^{\otimes 2})).\]
 Thus, by applying determinant to both sides of~(\ref{eqn:rdt-trans}), we have
 \[\det(\varphi(M_{\hat{f}})) = \det(\varphi(M_f)) \det(T)^6\]
 as claimed.
\end{proof}

In particular,
for a nonsingular matrix $T \in \mathbb{C}^{2 \times 2}$, $M_f$ is redundant and $\widetilde{M_{f}}$ is nonsingular
iff $M_{\hat{f}}$ is redundant and $\widetilde{M_{\hat{f}}}$ is nonsingular.
From Corollary~\ref{cor:arity4:nonsingular_compressed_hard} and Lemma~\ref{lem:arity4:nonsingular_compressed_invariant},
we have the following corollary.
\begin{corollary} \label{cor:arity4:nonsingular_compressed_hard_trans}
 Let $f$ be an arity~$4$ signature with complex weights.
 If there exists a nonsingular matrix $T \in \mathbb{C}^{2 \times 2}$ such that $\hat{f} = T^{\otimes 4} f$,
 where $M_{\hat{f}}$ is redundant and $\widetilde{M_{\hat{f}}}$ is nonsingular,
 then $\Holant(f)$ is $\SHARPP$-hard.
\end{corollary}

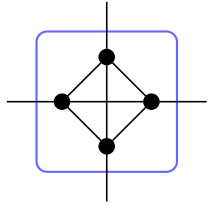
\begin{figure}[t]
 \centering
 \begin{tikzpicture}[scale=\scale,transform shape,node distance=\nodeDist,semithick]
  \node[external] (0)                    {};
  \node[internal] (1) [right       of=0] {};
  \node[internal] (2) [above right of=1] {};
  \node[external] (3) [above       of=2] {};
  \node[internal] (4) [below right of=1] {};
  \node[external] (5) [below       of=4] {};
  \node[internal] (6) [below right of=2] {};
  \node[external] (7) [right       of=6] {};
  \path (1) edge (2)
            edge (4)
            edge (6)
        (2) edge (4)
            edge (6)
        (4) edge (6);
  \path (0) edge node[near end]   (e1) {} (1)
        (2) edge node[near start] (e2) {} (3)
        (4) edge node[near start] (e3) {} (5)
        (6) edge node[near start] (e4) {} (7);
  \begin{pgfonlayer}{background}
   \node[inner sep=0pt,transform shape=false,draw=\borderColor,thick,rounded corners,fit = (1) (2) (4) (6) (e1) (e2) (e3) (e4)] {};
  \end{pgfonlayer}
 \end{tikzpicture}
 \caption{The tetrahedron gadget. Each vertex is assigned $\hat{f} = [t, 1, 0, 0, 0]$.}
 \label{fig:gadget:tetrahedron}
\end{figure}

The next lemma applies Corollary~\ref{cor:arity4:nonsingular_compressed_hard}.

\begin{lemma} \label{lem:arity4:double_root}
 Let $f_k = c k \alpha^{k-1} + d \alpha^k$, where $c \ne 0$ and $0 \le k \le 4$.
 Then the problem $\Holant([f_0, f_1, f_2, f_3, f_4])$ is $\SHARPP$-hard unless $\alpha = \pm i$,
 in which case the Holant vanishes.
\end{lemma}

\begin{proof}
 If $\alpha = \pm i$, then $\rd^{\pm}(f) = 1$, $\vd^{\pm}(f) = 3$, and so $f = [f_0, f_1, f_2, f_3, f_4]$ is vanishing by Theorem~\ref{thm:van}.
 Otherwise, a holographic transformation with orthogonal basis $T = \tfrac{1}{\sqrt{1 + \alpha^2}} \left[\begin{smallmatrix} 1 & \alpha \\ \alpha & -1 \end{smallmatrix}\right]$
 transforms $f$ to $\hat{f} = [t,1,0,0,0]$ for some $t \in \mathbb{C}$ after normalizing the second entry.
 (See Appendix~\ref{sec:orthogonal} for details.)
 Using the tetrahedron gadget in Figure~\ref{fig:gadget:tetrahedron} with $\hat{f}$ assigned to each vertex,
 we have a gadget with signature
 \[h = [t^4 + 6 t^2 + 3, t^3 + 3 t, t^2 + 1, t, 1].\]
 Since the determinant of $\widetilde{M_h}$ is~4,
 the compressed signature matrix of this gadget is nonsingular,
 so we are done by Corollary~\ref{cor:arity4:nonsingular_compressed_hard}.
\end{proof}

Now we are ready to prove a dichotomy for a single arity~4 signature.
\begin{theorem} \label{thm:arity4:singleton}
 If $f$ is a non-degenerate, symmetric, complex-valued signature of arity~4 in Boolean variables,
 then $\Holant(f)$ is $\SHARPP$-hard unless $f$ is $\mathscr{A}$-transformable, $\mathscr{P}$-transformable, or vanishing,
 in which case the problem is computable in polynomial time.
\end{theorem}

\begin{proof}
 Let $f = [f_0, f_1, f_2, f_3, f_4]$.
 If the compressed signature matrix $\widetilde{M_f}$ is nonsingular,
 then $\Holant(f)$ is $\SHARPP$-hard by Corollary~\ref{cor:arity4:nonsingular_compressed_hard},
 so assume that the rank of $\widetilde{M_f}$ is at most~2.
 Hence the rows of $\widetilde{M_f}$ are linearly dependent.
 There exist some $a,b,c \in \mathbb{C}$ that are not all~$0$ such that
 \[
      a \begin{pmatrix} f_0 \\ f_1 \\ f_2 \end{pmatrix}
  + 2 b \begin{pmatrix} f_1 \\ f_2 \\ f_3 \end{pmatrix}
  +   c \begin{pmatrix} f_2 \\ f_3 \\ f_4 \end{pmatrix}
  =
        \begin{pmatrix} 0   \\   0 \\   0 \end{pmatrix}.
 \]
 If $a = c = 0$, then $b \ne 0$, so $f_1 = f_2 = f_3 = 0$.
 In this case, $f \in \mathscr{P}$ is a generalized equality signature, so $f$ is $\mathscr{P}$-transformable.
 Now suppose $a$ and $c$ are not both 0.
 If $b^2 - 4 a c \ne 0$,
 then $f_k = \alpha_1^{4-k} \alpha_2^k + \beta_1^{4-k} \beta_2^k$,
 where $\alpha_1 \beta_2 - \alpha_2 \beta_1 \ne 0$.
 A holographic transformation by $\left[\begin{smallmatrix} \alpha_1 & \beta_1 \\ \alpha_2 & \beta_2 \end{smallmatrix}\right]$ transforms $f$ to $=_4$
 and we can use Theorem~\ref{thm:k-reg_homomorphism}$'$ to show that $f$ is either $\mathscr{A}$- or $\mathscr{P}$-transformable
 unless $\Holant(f)$ is $\SHARPP$-hard.
 Otherwise, $b^2 - 4 a c = 0$ and there are two cases.
 In the first,  for any $0 \le k \le 4$, $f_k = c k \alpha^{k-1} + d \alpha^k$, where $c \ne 0$.
 In the second, for any $0 \le k \le 4$, $f_k = c (4-k) \alpha^{3-k} + d \alpha^{4-k}$, where $c \ne 0$.
 These cases map between each other under a holographic transformation by $\left[\begin{smallmatrix} 0 & 1 \\ 1 & 0 \end{smallmatrix}\right]$,
 so assume that we are in the first case.
 Then we are done by Lemma~\ref{lem:arity4:double_root}.
\end{proof}

The next lemma is related to vanishing signatures.
It appears here because its proof uses similar techniques to those in this section.

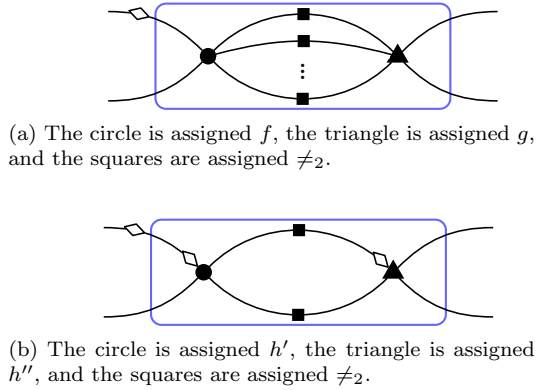
\begin{figure}[t]
 \centering
 \def\capWidth{7cm}
 \captionsetup[subfigure]{width=\capWidth}
 \subfloat[The circle is assigned $f$, the triangle is assigned $g$, and the squares are assigned $\neq_2$.]{
  \makebox[\capWidth][c]{
   \begin{tikzpicture}[scale=\scale,transform shape,node distance=\nodeDist,semithick]
    \node[internal]  (0)                     {};
    \node[external]  (1) [above left  of=0]  {};
    \node[external]  (2) [below left  of=0]  {};
    \node[external]  (3) [left        of=1]  {};
    \node[external]  (4) [left        of=2]  {};
    \node[external]  (5) [above left  of=3]  {};
    \node[external]  (8) [right       of=0]  {};
    \node[external]  (9) [right       of=8]  {};
    \node[triangle] (10) [right       of=9]  {};
    \node[external] (11) [above right of=10] {};
    \node[external] (12) [below right of=10] {};
    \node[external] (13) [right       of=11] {};
    \node[external] (14) [right       of=12] {};
    \path (3) edge[out=   0, in= 135, postaction={decorate, decoration={
                                                             markings,
                                                             mark=at position 0.4 with {\arrow[>=diamond, white] {>}; },
                                                             mark=at position 0.4 with {\arrow[>=open diamond]   {>}; } } }] (0)
          (0) edge[out=-135, in=   0]  (4)
         (10) edge[out=  45, in= 180] (13)
               edge[out= -45, in= 180] (14);
    \path (0) edge[out=  45, in= 135]        node[square] {} (10)
              edge[out=  15, in= 165]        node[square] {} (10)
              edge[out= -10, in=-170, white] node[black]  {\Huge $\vdots$} (10)
              edge[out= -45, in=-135]        node[square] {} (10);
    \begin{pgfonlayer}{background}
     \node[inner sep=0pt,transform shape=false,draw=\borderColor,thick,rounded corners,fit = (1) (2) (11) (12)] {};
    \end{pgfonlayer}
   \end{tikzpicture} \label{fig:gadget:special_vanishing_case:many_glue}}}
 \qquad
 \subfloat[The circle is assigned $h'$, the triangle is assigned $h''$, and the squares are assigned $\neq_2$.]{
  \makebox[\capWidth][c]{
   \begin{tikzpicture}[scale=\scale,transform shape,node distance=\nodeDist,semithick]
    \node[internal]  (0)                     {};
    \node[external]  (1) [above left  of=0]  {};
    \node[external]  (2) [below left  of=0]  {};
    \node[external]  (3) [left        of=1]  {};
    \node[external]  (4) [left        of=2]  {};
    \node[external]  (5) [above left  of=3]  {};
    \node[external]  (8) [right       of=0]  {};
    \node[external]  (9) [right       of=8]  {};
    \node[triangle] (10) [right       of=9]  {};
    \node[external] (11) [above right of=10] {};
    \node[external] (12) [below right of=10] {};
    \node[external] (13) [right       of=11] {};
    \node[external] (14) [right       of=12] {};
    \path (3) edge[out=   0, in= 135, postaction={decorate, decoration={
                                                             markings,
                                                             mark=at position 0.4   with {\arrow[>=diamond,white] {>}; },
                                                             mark=at position 0.4   with {\arrow[>=open diamond]  {>}; },
                                                             mark=at position 0.999 with {\arrow[>=diamond,white] {>}; },
                                                             mark=at position 1     with {\arrow[>=open diamond]  {>}; } } }] (0)
          (0) edge[out=-135, in=   0]  (4)
         (10) edge[out=  45, in= 180] (13)
              edge[out= -45, in= 180] (14);
    \path (0) edge[out=  45, in= 135, postaction={decorate, decoration={
                                                             markings,
                                                             mark=at position 0.999 with {\arrow[>=diamond,white] {>}; },
                                                             mark=at position 0.999 with {\arrow[>=open diamond]  {>}; } } }] node[square] {} (10)
              edge[out= -45, in=-135] node[square] {} (10);
    \begin{pgfonlayer}{background}
     \node[inner sep=0pt,transform shape=false,draw=\borderColor,thick,rounded corners,fit = (1) (2) (11) (12)] {};
    \end{pgfonlayer}
   \end{tikzpicture}
  }
  \label{fig:gadget:special_vanishing_case:2glue}
 }
 \caption{Gadget constructions used to obtain a hard and redundant arity~$4$ signature.}
 \label{fig:gadget:special_vanishing_case:glue}
\end{figure}

\begin{lemma} \label{lem:arity4:special_vanishing_case}
 If $f = [0,1,0,\dotsc,0]$ and $g = [0,\dotsc,0,1,0]$ are both of arity $d \ge 3$,
 then the problem $\holant{[0,1,0]}{\{f, g\}}$ is $\SHARPP$-hard.
\end{lemma}

\begin{proof}
 Our goal is to obtain a signature that satisfies the hypothesis of Corollary~\ref{cor:arity4:nonsingular_compressed_hard_trans}.

 The gadget in Figure~\ref{fig:gadget:special_vanishing_case:many_glue}, with $f$ assigned to the circle vertex,
 $g$ assigned to the triangle vertex, and $\neq_2$ assigned to the square vertices, has signature $h$ with signature matrix
 \[M_h = \begin{bmatrix} 0 & 0 & 0 & v \\ 0 & 1 & 1 & 0 \\ 0 & 1 & 1 & 0 \\ 0 & 0 & 0 & 0 \end{bmatrix},\]
 where $v = d-2$ is positive since $d \ge 3$.
 Although this signature matrix is redundant, its compressed form is singular.
 Rotating this gadget $90^\circ$ clockwise and $90^\circ$ counterclockwise yield signatures $h'$ and $h''$ respectively, with signature matrices
 \[
  M_{h'} = \begin{bmatrix} 0 & 0 & 0 & 1 \\ 0 & v & 1 & 0 \\ 0 & 1 & 0 & 0 \\ 1 & 0 & 0 & 0 \end{bmatrix}
  \qquad \text{and} \qquad
  M_{h''} = \begin{bmatrix} 0 & 0 & 0 & 1 \\ 0 & 0 & 1 & 0 \\ 0 & 1 & v & 0 \\ 1 & 0 & 0 & 0 \end{bmatrix}.
 \]
 The gadget in Figure~\ref{fig:gadget:special_vanishing_case:2glue}, with $h'$ assigned to the circle vertex,
 $h''$ assigned to the triangle vertex, and $\neq_2$ assigned to the square vertices, has a signature $r$ with signature matrix
 \[
  M_r
  = M_{h'} \begin{bmatrix} 0 & 0 & 0 & 1 \\ 0 & 0 & 1 & 0 \\ 0 & 1 & 0 & 0 \\ 1 & 0 & 0 & 0 \end{bmatrix} M_{h''}
  = \begin{bmatrix} 0 & 0 & 0 & 1 \\ 0 & v & v^2 + 1 & 0 \\ 0 & 1 & v & 0 \\ 1 & 0 & 0 & 0 \end{bmatrix}.
 \]
 Note that the effect of the $\neq_2$ signatures is to reverse all four rows of $M_{h''}$ before multiplying it to the right of $M_{h'}$.
 Although this signature matrix is not redundant,
 every entry of Hamming weight~$2$ is nonzero since $v$ is positive.

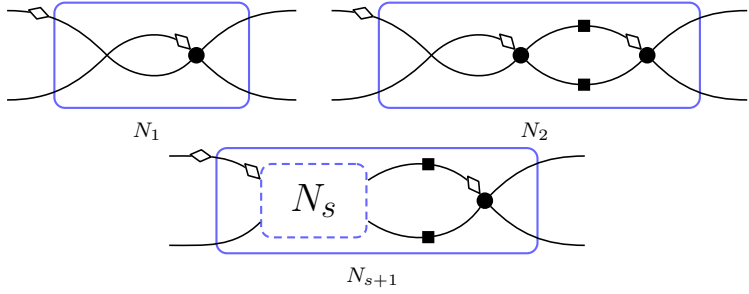
\begin{figure}[t]
 \centering
 \captionsetup[subfigure]{labelformat=empty}
 \subfloat[$N_1$]{
  \begin{tikzpicture}[scale=\scale,transform shape,node distance=\nodeDist,semithick]
   \node[internal]  (0)                    {};
   \node[external]  (1) [above left  of=0] {};
   \node[external]  (2) [below left  of=0] {};
   \node[external]  (3) [below left  of=1] {};
   \node[external]  (4) [below left  of=3] {};
   \node[external]  (5) [above left  of=3] {};
   \node[external]  (6) [left        of=4] {};
   \node[external]  (7) [left        of=5] {};
   \node[external]  (8) [below right of=0] {};
   \node[external]  (9) [above right of=0] {};
   \node[external] (10) [right       of=8] {};
   \node[external] (11) [right       of=9] {};
   \path (3.center) edge[out= 45, in= 135, postaction={decorate, decoration={
                                                                  markings,
                                                                  mark=at position 0.999 with {\arrow[>=diamond,white] {>}; },
                                                                  mark=at position 0.999 with {\arrow[>=open diamond]  {>}; } } }] (0)
                    edge[out=-45, in=-135]  (0)
                (0) edge[out= 45, in= 180] (11)
                    edge[out=-45, in= 180] (10)
                (7) edge[out=  0, in= 135, postaction={decorate, decoration={
                                                                  markings,
                                                                  mark=at position 0.38  with {\arrow[>=diamond,white] {>}; },
                                                                  mark=at position 0.38  with {\arrow[>=open diamond]  {>}; } } }] (3.center)
         (3.center) edge[out=-135, in=  0] (6);
   \begin{pgfonlayer}{background}
    \node[inner sep=0pt,transform shape=false,draw=\borderColor,thick,rounded corners,fit = (4) (5) (8) (9)] {};
   \end{pgfonlayer}
  \end{tikzpicture}}
 \hspace{-5pt}
 \subfloat[$N_2$]{
  \begin{tikzpicture}[scale=\scale,transform shape,node distance=\nodeDist,semithick]
   \node[internal]  (0)                     {};
   \node[external]  (1) [above left  of=0]  {};
   \node[external]  (2) [below left  of=0]  {};
   \node[external]  (3) [below left  of=1]  {};
   \node[external]  (4) [below left  of=3]  {};
   \node[external]  (5) [above left  of=3]  {};
   \node[external]  (6) [left        of=4]  {};
   \node[external]  (7) [left        of=5]  {};
   \node[external]  (8) [right       of=0]  {};
   \node[internal]  (9) [right       of=8]  {};
   \node[external] (10) [above right of=9]  {};
   \node[external] (11) [below right of=9]  {};
   \node[external] (12) [right       of=10] {};
   \node[external] (13) [right       of=11] {};
   \path (3.center) edge[out= 45, in= 135, postaction={decorate, decoration={
                                                           markings,
                                                           mark=at position 0.999 with {\arrow[>=diamond,white] {>};
                                                                                        \arrow[>=open diamond]  {>}; } } }] (0)
                    edge[out= -45, in=-135] (0)
                (7) edge[out=   0, in=135, postaction={decorate, decoration={
                                                                  markings,
                                                                  mark=at position 0.38  with {\arrow[>=diamond,white] {>};
                                                                                               \arrow[>=open diamond]  {>}; } } }] (3.center)
         (3.center) edge[out=-135, in=  0]  (6)
                (9) edge[out=  45, in=180] (12)
                    edge[out= -45, in=180] (13);
   \path (0) edge[out=-45, in=-135] node[square] {} (9)
             edge[out= 45, in= 135, postaction={decorate, decoration={
                                                    markings,
                                                    mark=at position 0.999 with {\arrow[>=diamond,white] {>};
                                                                                 \arrow[>=open diamond]  {>}; } } }] node[square] {} (9);
   \begin{pgfonlayer}{background}
    \node[inner sep=0pt,transform shape=false,draw=\borderColor,thick,rounded corners,fit = (4) (5) (10) (11)] {};
   \end{pgfonlayer}
  \end{tikzpicture}}
 \hspace{-5pt}
 \subfloat[$N_{s+1}$]{
  \begin{tikzpicture}[scale=\scale,transform shape,node distance=\nodeDist,semithick]
   \node[external]  (0)                     {};
   \node[external]  (1) [above left  of=0]  {};
   \node[external]  (2) [below left  of=0]  {};
   \node[external]  (3) [below left  of=1]  {};
   \node[external]  (4) [below left  of=3]  {};
   \node[external]  (5) [above left  of=3]  {};
   \node[external]  (6) [left        of=4]  {};
   \node[external]  (7) [left        of=5]  {};
   \node[external]  (8) [right       of=0]  {};
   \node[internal]  (9) [right       of=8]  {};
   \node[external] (10) [above right of=9]  {};
   \node[external] (11) [below right of=9]  {};
   \node[external] (12) [right       of=10] {};
   \node[external] (13) [right       of=11] {};
   \path let
          \p1 = (1),
          \p2 = (2)
         in
          node[external] at (\x1, \y1 / 2 + \y2 / 2) {\Huge $N_s$};
   \path let
          \p1 = (0)
         in
          node[external] (14) at (\x1 + 2, \y1 + 10) {};
   \path let
          \p1 = (0)
         in
          node[external] (15) at (\x1 + 2, \y1 - 10) {};
   \path let
          \p1 = (3)
         in
          node[external] (16) at (\x1 - 2, \y1 + 10) {};
   \path let
          \p1 = (3)
         in
          node[external] (17) at (\x1 - 2, \y1 - 10) {};
   \path  (7) edge[out=   0, in= 135, postaction={decorate, decoration={
                                                             markings,
                                                             mark=at position 0.44  with {\arrow[>=diamond,white] {>}; },
                                                             mark=at position 0.44  with {\arrow[>=open diamond]  {>}; },
                                                             mark=at position 0.999 with {\arrow[>=diamond,white] {>}; },
                                                             mark=at position 1     with {\arrow[>=open diamond]  {>}; } } }] (16)
         (17) edge[out=-135, in=   0]  (6)
          (9) edge[out=  45, in= 180] (12)
              edge[out= -45, in= 180] (13)
         (14) edge[out=  35, in= 125, postaction={decorate, decoration={
                                                             markings,
                                                             mark=at position 0.999 with {\arrow[>=diamond,white] {>}; },
                                                             mark=at position 0.999 with {\arrow[>=open diamond]  {>}; } } }] node[square] {} (9)
         (15) edge[out= -35, in=-125] node[square] {} (9);
   \begin{pgfonlayer}{background}
    \node[inner sep=0pt,transform shape=false,draw=\borderColor,thick,densely dashed,rounded corners,fit = (0) (1.south) (2.north) (3)] {};
    \node[inner sep=0pt,transform shape=false,draw=\borderColor,thick,rounded corners,fit = (4) (5) (10) (11)] {};
   \end{pgfonlayer}
  \end{tikzpicture}}
 \caption{Recursive construction to interpolate a signature $r'$ that is only a rotation away from having a redundant signature matrix and nonsingular compressed matrix.
  The circles are assigned $r$ and the squares are assigned $\neq_2$.}
 \label{fig:gadget:special_vanishing_case:recursive_construction}
\end{figure}

 Now we claim that we can use $r$ to interpolate the following signature $r'$,
 for any nonzero value $t \in \mathbb{C}$, via the construction in Figure~\ref{fig:gadget:special_vanishing_case:recursive_construction}.
 Define $p^\pm = (v \pm \sqrt{v^2 + 4}) / 2$, $P = \left[\begin{smallmatrix} 1 & 1  \\p^+ & p^-  \end{smallmatrix}\right]$,
 and  $T = P \left[\begin{smallmatrix} t & 0 \\ 0 & t^{-1} \end{smallmatrix}\right] P^{-1}$ where $t \in \mathbb{C}$ is any nonzero value.
 The signature matrix of the target signature $r'$ is
 \begin{equation} \label{eqn:Mrprime_sig_matrix}
  M_{r'} = \begin{bmatrix} 0 & 0 & 0 & 1 \\ 0 & \multicolumn{2}{c}{\multirow{2}{*}{$T$}} & 0 \\ 0 & & & 0 \\ 1 & 0 & 0 & 0 \end{bmatrix}.
 \end{equation}

 Consider an instance $\Omega$ of $\holant{{\neq}_2}{\mathcal{F} \union \{r'\}}$ with $r \in \mathcal{F}$.
 Suppose that $r'$ appears $n$ times in $\Omega$.
 We construct from $\Omega$ a sequence of instances $\Omega_s$ of $\holant{{\neq}_2}{\mathcal{F}}$ indexed by $s \ge 1$.
 We obtain $\Omega_s$ from $\Omega$ by replacing each occurrence of $r'$
 with the gadget $N_s$ in Figure~\ref{fig:gadget:special_vanishing_case:recursive_construction}
 with $r$ assigned to the circle vertices and $\neq_2$ assigned to the square vertices.
 In $\Omega_s$,
 the edge corresponding to the $i$th significant index bit of $N_s$ connects to
 the same location as the edge corresponding to the $i$th significant index bit of $r'$ in $\Omega$.

 The signature matrix of $N_s$ is the $s$th power of the matrix obtained from $M_r$ after reversing all rows,
 and then switching the first and last rows of the final product,
 namely
 \[
  \begin{bmatrix} 0 & 0 & 0 & 1 \\ 0 & 1 & 0 & 0 \\ 0 & 0 &       1 & 0 \\ 1 & 0 & 0 & 0 \end{bmatrix}
  \begin{bmatrix} 1 & 0 & 0 & 0 \\ 0 & 1 & v & 0 \\ 0 & v & v^2 + 1 & 0 \\ 0 & 0 & 0 & 1 \end{bmatrix}^s
  =
  \begin{bmatrix} 0 & 0 & 0 & 1 \\ 0 & 1 & v & 0 \\ 0 & v & v^2 + 1 & 0 \\ 1 & 0 & 0 & 0 \end{bmatrix}
  \begin{bmatrix} 1 & 0 & 0 & 0 \\ 0 & 1 & v & 0 \\ 0 & v & v^2 + 1 & 0 \\ 0 & 0 & 0 & 1 \end{bmatrix}^{s-1}.
 \]
 The twist of the two input edges on the left side for the first copy of $M_r$ switches the middle two rows,
 which is equivalent to a total reversal of all rows, followed by the switching of the first and last rows.
 The total reversals of rows for all subsequent $s - 1$ copies of $M_r$ are due to the presence of $\neq_2$ signatures.

 After such reversals of rows, it is clear that the matrix is a direct sum of block matrices indexed by $\{00, 11\} \times \{00, 11\}$ and $\{01, 10\} \times \{10, 01\}$.
 Furthermore, in the final product, the block indexed by $\{00, 11\} \times \{00, 11\}$ is $\left[\begin{smallmatrix} 0 & 1 \\ 1 & 0 \end{smallmatrix}\right]$.
 Thus in the gadget $N_s$, the only entries of $M_{N_s}$ that vary with $s$ are the four entries in the middle.
 These middle four entries of $M_{N_s}$ form the 2-by-2 matrix $\left[\begin{smallmatrix} 1 & v \\ v & v^2 + 1 \end{smallmatrix}\right]^s$.
 Since $\left[\begin{smallmatrix} 1 & v \\ v & v^2 + 1 \end{smallmatrix}\right] = P \left[\begin{smallmatrix} \lambda_+ & 0 \\ 0 & \lambda_- \end{smallmatrix}\right] P^{-1}$,
 where $\lambda_{\pm} = (v^2 + 2 \pm v \sqrt{v^2 + 4})/2$ are the eigenvalues,
 we have
 \[
  \begin{bmatrix}
   1 & v \\
   v & v^2 + 1
  \end{bmatrix}^s
  =
  P
  \begin{bmatrix}
   \lambda_{+}^s & 0 \\
   0             & \lambda_{-}^s
  \end{bmatrix}
  P^{-1}.
 \]
 The determinant is $\lambda_{+} \lambda_{-} = 1$, so the eigenvalues are nonzero.
 Since $v$ is positive, the ratio of the eigenvalues $\lambda_{+} / \lambda_{-}$ is not a root of unity, so neither $\lambda_+$ nor $\lambda_-$ is a root of unity.

 Now we determine the relationship between $\Holant_\Omega$ and $\Holant_{\Omega_s}$.
 We can view our construction of $\Omega_s$ as first replacing $M_{r'}$ with
 \[
  \begin{bmatrix} 1 & 0 & 0 & 0 \\ 0 & \multicolumn{2}{c}{\multirow{2}{*}{$P$}}      & 0 \\ 0 &   &        & 0 \\ 0 & 0 & 0 & 1 \end{bmatrix}
  \begin{bmatrix} 0 & 0 & 0 & 1 \\ 0 & t & 0                                         & 0 \\ 0 & 0 & t^{-1} & 0 \\ 1 & 0 & 0 & 0 \end{bmatrix}
  \begin{bmatrix} 1 & 0 & 0 & 0 \\ 0 & \multicolumn{2}{c}{\multirow{2}{*}{$P^{-1}$}} & 0 \\ 0 &   &        & 0 \\ 0 & 0 & 0 & 1 \end{bmatrix},
 \]
 which does not change the Holant value, and then replacing the new signature matrix in the middle with the signature matrix
 \[
  \begin{bmatrix} 0 & 0 & 0 & 1 \\ 0 & \lambda_+^s & 0 & 0 \\ 0 & 0 & \lambda_-^s & 0 \\ 1 & 0 & 0 & 0 \end{bmatrix}.
 \]

 We stratify the assignments in $\Omega_s$ based on the assignments to the $n$ occurrences of the signature matrix
 \begin{equation} \label{eqn:rPrime_jnf_signature_matrix}
  \begin{bmatrix} 0 & 0 & 0 & 1 \\ 0 & t & 0 & 0 \\ 0 & 0 & t^{-1} & 0 \\ 1 & 0 & 0 & 0 \end{bmatrix}.
 \end{equation}
 The inputs to this matrix are from $\{0, 1\}^2 \times \{0, 1\}^2$,
 which correspond to the four input bits.
 Recall the way rows and columns of a signature matrix are ordered from Definition~\ref{def:signature_matrix}.
 Thus, e.g., the entry $t$ corresponds to the cyclic input bit pattern $0110$ in counterclockwise order.
 We only need to consider the assignments that assign
 \begin{itemize}
  \item $i$ many times the bit pattern $0110$,
  \item $j$ many times the bit pattern $1001$, and
  \item $k$ many times the bit patterns $0011$ or $1100$,
 \end{itemize}
 since any other assignment contributes a factor of 0.
 Let $c_{ijk}$ be the sum over all such assignments of the products of evaluations of all signatures in $\Omega_s$ except for~(\ref{eqn:rPrime_jnf_signature_matrix}).
 Then
 \[\Holant_\Omega = \sum_{i + j + k = n} t^{i-j} c_{ijk}\]
 and the value of the Holant on $\Omega_s$, for $s \ge 1$, is
 \[
  \Holant_{\Omega_s}
  = \sum_{i + j + k = n} \lambda_+^{s i} \lambda_-^{s j} c_{ijk}
  = \sum_{i + j + k = n} \lambda_+^{s(i-j)} c_{ijk}.
 \]
 This Vandermonde system does not have full rank.
 However, we  can define for $-n \le \ell \le n$,
 \[c_\ell' = \sum_{\substack{i-j= \ell\\i+j+k=n}} c_{ijk}.\]
 Then the Holant of $\Omega$ is
 \[\Holant_\Omega = \sum_{-n \le \ell \le n} t^\ell c_\ell'\]
 and the Holant of $\Omega_s$ is
 \[\Holant_{\Omega_s} = \sum_{-n \le \ell \le n} \lambda_+^{s \ell} c_\ell'.\]
 Now this Vandermonde has full rank because $\lambda_+$ is neither~$0$ nor a root of unity.
 Therefore, we can solve for the unknowns $c_\ell'$ and obtain the value of $\Holant_\Omega$.
 This completes our claim that we can interpolate the signature $r'$ in~(\ref{eqn:Mrprime_sig_matrix}), for any nonzero $t \in \mathbb{C}$.

 Let $t = (\sqrt{v^2 + 8} + \sqrt{v^2 + 4})/2$ so $t^{-1} = (\sqrt{v^2 + 8} - \sqrt{v^2 + 4})/2$.
 Let $a = (\sqrt{v^2 + 8} -v)/2$ and $b = (\sqrt{v^2 + 8} + v)/2$, so $a b = 2 \ne 0$.
 One can verify that
 \[
  P \begin{bmatrix} t & 0 \\ 0 & t^{-1} \end{bmatrix} P^{-1}
  = \begin{bmatrix} a & 1 \\ 1 & b \end{bmatrix}.
 \]
 Thus, the signature matrix for $r'$ is
 \[M_{r'} = \begin{bmatrix} 0 & 0 & 0 & 1 \\ 0 & a & 1 & 0 \\ 0 & 1 & b & 0 \\ 1 & 0 & 0 & 0 \end{bmatrix}.\]
 After a counterclockwise rotation of $90^\circ$ on the edges of $r'$,
 we have a signature $r''$ with a redundant signature matrix
 \[M_{r''} = \begin{bmatrix} 0 & 0 & 0 & a \\ 0 & 1 & 1 & 0 \\ 0 & 1 & 1 & 0 \\ b & 0 & 0 & 0 \end{bmatrix}.\]
 Its compressed signature matrix
 \[\widetilde{M_{r''}} = \begin{bmatrix} 0 & 0 & a \\ 0 & 2 & 0 \\ b & 0 & 0 \end{bmatrix}\]
 is nonsingular.
 After a holographic transformation by $Z^{-1}$,
 where $Z = \tfrac{1}{\sqrt{2}} \left[\begin{smallmatrix} 1 & 1 \\ i & -i \end{smallmatrix}\right]$,
 the binary disequality $(\neq_2) = [0,1,0]$ is transformed to the binary equality $(=_2) = [1,0,1]$.
 Thus the problem $\holant{[0,1,0]}{r''}$ is transformed to $\holant{{=}_2}{Z^{\otimes 4} r''}$,
 which is the same as $\Holant(Z^{\otimes 4} r'')$.
 We conclude that this Holant problem is \#P-hard by Corollary~\ref{cor:arity4:nonsingular_compressed_hard_trans}.
\end{proof}

\section{Vanishing Signatures Revisited}

With Corollary~\ref{cor:arity4:nonsingular_compressed_hard},
Corollary~\ref{cor:arity4:nonsingular_compressed_hard_trans},
and Lemma~\ref{lem:arity4:special_vanishing_case} in hand,
we revisit the vanishing signatures to determine what signatures combine with them to give $\SHARPP$-hardness.
We begin with unary signatures and their tensor powers.

\begin{lemma} \label{lem:van:deg}
 Let $f \in \mathscr{V}^\sigma$ be a symmetric signature of arity $n$ with $\rd^\sigma(f)=d\ge 2$ where $\sigma \in \{+,-\}$.
 Suppose $v = u^{\otimes m}$ is a symmetric degenerate signature for some unary signature $u$ and some integer $m \ge 1$.
 If $u$ is not a multiple of $[1, \sigma i]$,
 then $\Holant(f,v)$ is \numP-hard.
\end{lemma}
\begin{proof}
  We consider $\sigma = +$ since the other case is similar.
  Since $f \in \mathscr{V}^+$,
  we have $n > 2d \ge 4$.
  Under a holographic transformation by $Z$,
  we have
  \[
   \Holant(f,v) \equiv \holant{{\neq}_2}{\hat{f},[a,b]^{\otimes m}},
  \]
  where $\hat{f} = \left(Z^{-1}\right)^{\otimes n} f$ and $[a,b]^{\otimes m} = \left(Z^{-1}\right)^{\otimes m} v$ with $b \neq 0$
  since $u$ is not a multiple of $[1,i]$.
  Moreover,
  $\hat{f} = [\hat{f}_0, \hat{f}_1, \dotsc, \hat{f}_d, 0, \dotsc, 0]$ with 
  $\hat{f}_d\neq 0$ by Lemma~\ref{lem:vanishing_form_in_Z_basis}.

  We get $\widehat{f'} = [\hat{f}_{d-2}, \hat{f}_{d-1}, \hat{f}_d, 0, \dotsc, 0]$ of arity $n-2d+4$
  by $d-2$ self-loops via $\neq_2$ on $\hat{f}$.
  This is a signature on the right side in $\holant{\cdot}{\cdot}$ notation.
  With two more self-loops,
  we get $[1,0]^{\otimes n-2d}$,
  also on the right.

  We claim that we can use $[1,0]^{\otimes n-2d}$ and $[a,b]^{\otimes m}$
  to create $[a,b]^{\otimes n-2d}$.
  Let $t=\gcd(m,n-2d)$.
  If $n-2d>m$,
  then we connect $[a,b]^{\otimes m}$ to $[1,0]^{\otimes n-2d}$ via $\neq_2$
  to get $[1,0]^{\otimes n-2d-m}$ up to a nonzero factor $b \neq 0$.
  We repeat this process until we get a tensor power $[1,0]^{\otimes \ell}$ for some $\ell \le m$.
  We can do a similar construction if $m>n-2d$.
  Repeat this process, which is a subtractive Euclidean algorithm.
  Halt upon getting both $[1,0]^{\otimes t}$ and $[a,b]^{\otimes t}$.
  Then we combine $\tfrac{n-2d}{t}$ copies of $[a,b]^{\otimes t}$ to get $[a,b]^{\otimes n-2d}$.

  Now connecting $[a,b]^{\otimes n-2d}$ back to $\widehat{f'}$ via $\neq_2$,
  gives $\widehat{f''}=[\widehat{f''}_0, \widehat{f''}_1, \widehat{f''}_2, 0, 0]$ of arity~$4$.
  Moreover,
  $\widehat{f''}_2 = b^{n-2d} \widehat{f}_d \neq 0$.
  Notice that $\Holant({\neq}_2 \mid [\widehat{f''}_{0}, \widehat{f''}_{1}, \widehat{f''}_{2}, 0, 0]) \equiv \Holant({\neq}_2\mid[0,0,1,0,0])$,
  the Eulerian Orientation problem over planar $4$-regular graphs,
  (see Section~\ref{subsec:vanishing-by-holographic-Z})
  which is $\SHARPP$-hard by Theorem~\ref{thm:4reg_EO_hard}.
  Thus, $\Holant(f,v)$ is \numP-hard.
\end{proof}

Next we consider binary signatures.

\begin{lemma} \label{lem:van:bin}
 Let $f \in \mathscr{V}^\sigma$ be a symmetric non-degenerate signature where $\sigma \in \{+,-\}$.
 Suppose $g$ is a non-degenerate binary signature.
 If $g \not\in \mathscr{R}^\sigma_2$,
 then $\Holant(f,g)$ is $\SHARPP$-hard.
\end{lemma}

\begin{proof}
  We consider $\sigma = +$ since the other case is similar.
  A unary signature is degenerate.
  If $f$ is binary, then $\vd^+(f) > 1$.
  Hence $\vd^+(f) \ge 2$, and so $f$ is degenerate.
  Since $f$ is non-degenerate, $\arity(f) \ge 3$.
  Under a $Z$ transformation,
  \begin{align*}
    \Holant(f,g) & \equiv \holant{{\neq}_2}{\hat{f},\hat{g}},
  \end{align*}
  where $\hat{f}=\left(Z^{-1}\right)^{\otimes n}f$ and $\hat{g}=\left(Z^{-1}\right)^{\otimes 2}g$.
  Since $g\not\in\mathscr{R}_2^+$,
  we may assume that $\hat{g}=[a,b,1]$ by Lemma~\ref{lem:vanishing_form_in_Z_basis} with a nonzero $\hat{g}_2$ entry.
  Moreover since $g$ is non-degenerate,
  so is $\hat{g}$, and $b^2\neq a$.

 We prove the lemma by induction on the arity of $f$.
 There are two base cases,
 $\arity(f) = 3$ and $\arity(f) = 4$.
 However,
 the arity~$3$ case is easily reduced to the arity~$4$ case.
 We show this first,
 and then show that the lemma holds in the arity~$4$ case.

 Assume $\arity(f) = 3$.  
 Since $f \in \mathscr{V}^+$, we have $\rd^+(f) < 3/2$, thus $f \in \mathscr{R}^+_2$.
 However $f$ is non-degenerate, $\rd^+(f)>0$, and so $\rd^+(f)=1$ and $\vd^+(f)=2$.
 By Lemma~\ref{lem:vanishing_form_in_Z_basis}, $\hat{f}=[t,1,0,0]$ for some $t$.


\begin{figure}[t]
 \centering
 \begin{tikzpicture}[scale=\scale,transform shape,node distance=\nodeDist,semithick]
  \node[internal]  (0)                    {};
  \node[external]  (1) [above left  of=0] {};
  \node[external]  (2) [below left  of=0] {};
  \node[external]  (3) [left        of=1] {};
  \node[external]  (4) [left        of=2] {};
  \node[square]    (5) [right       of=0] {};
  \node[internal]  (6) [right       of=5] {};
  \node[external]  (7) [above right of=6] {};
  \node[external]  (8) [below right of=6] {};
  \node[external]  (9) [right       of=7] {};
  \node[external] (10) [right       of=8] {};
  \path (0) edge[out= 135, in=   0]  (3)
            edge[out=-135, in=   0]  (4)
            edge                     (6)
        (6) edge[out=  45, in= 180]  (9)
            edge[out= -45, in= 180] (10);
  \begin{pgfonlayer}{background}
   \node[inner sep=0pt,transform shape=false,draw=\borderColor,thick,rounded corners,fit = (1) (2) (7) (8)] {};
  \end{pgfonlayer}
 \end{tikzpicture}
 \caption{The circles are assigned $[t,1,0,0]$ and the square is assigned $\neq_2$.}
 \label{fig:gadget:arity3_to_arity4}
\end{figure}
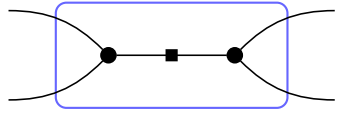

 We connect two copies of $f$ together by one edge to get an arity~$4$ signature $f'$.
 By construction,  it may not appear that $f'$ is a symmetric signature.
 However,  we show that $f'$ is in fact \emph{symmetric}, non-degenerate, and vanishing.
 It is clearly a vanishing signature, since $f$ is vanishing.
 Equivalently this is to connect two $\hat{f}=[t,1,0,0]$ together via a $\neq_2$.
 It is the gadget in Figure~\ref{fig:gadget:arity3_to_arity4}.
 One can verify that the resulting signature is $\hat{f'} = [2t,1,0,0,0]$.
 The crucial observation is that it takes the same value 0 on inputs $1010$ and $1100$,
 where the left two bits are input to one copy of $f$ and the right two bits are for another.
 The corresponding signature $f'$ is non-degenerate with $\rd^+(f') = 1$ and vanishing.
 Thus we reduce to the arity $4$ case.

 Next we consider the base case of $\arity(f) = 4$.
 Since $f \in \mathscr{V}^+$,
 we have $\vd^+(f) > 2$ and $\rd^+(f) < 2$.
 Since $f$ is non-degenerate,
 we have $\rd^+(f) \neq -1, 0$.
 Hence $\rd^+(f) = 1$ and $\vd^+(f) = 3$.
 By Lemma~\ref{lem:vanishing_form_in_Z_basis}, $\hat{f}=[t,1,0,0,0]$ for some $t$.
 We will work in the $Z$ basis henceforth.
 

 Our next goal is to show that we can realize a signature of the form $[c,0,1]$ with $c \neq 0$.
 If $b=0$,
 then $\hat{g}$ is what we want since in this case $a = a - b^2 \ne 0$.
 
 Now we assume $b \neq 0$.
 By connecting $\hat{g}$ to $\hat{f}$ via $\neq_2$,
 we get $[t + 2 b, 1, 0]$.
 If $t \neq -2b$,
 then by Lemma~\ref{lem:simple_interpolation:van:bin},
 we can interpolate any binary signature of the form $[v,1,0]$.
 Otherwise $t = - 2 b$.
 Then we connect two copies of $\hat{g}$ via $\neq_2$,
 and get $\widehat{g}' = [2 a b, a + b^2, 2 b]$.
 By connecting this $\widehat{g}'$ to $\hat{f}$ via $\neq_2$,
 we get $[2 (a - b^2), 2 b, 0]$,
 using $t = -2 b$.
 Since $a \neq b^2$ and $b \neq 0$,
 we can once again interpolate any $[v,1,0]$ by Lemma~\ref{lem:simple_interpolation:van:bin}.

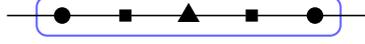
\begin{figure}[t]
 \centering
 \begin{tikzpicture}[scale=\scale,transform shape,node distance=\nodeDist,semithick]
  \node[external] (0)                    {};
  \node[internal] (1) [right       of=0] {};
  \node[square]   (2) [right       of=1] {};
  \node[triangle] (3) [right       of=2] {};
  \node[square]   (4) [right       of=3] {};
  \node[internal] (5) [right       of=4] {};
  \node[external] (6) [right       of=5] {};
  \node[external] (7) [above right of=3] {};
  \node[external] (8) [below right of=3] {};
  \path (0) edge            node[near end]   (e1) {} (1)
        (1) edge                                     (2)
        (2) edge                                     (3)
        (3) edge                                     (4)
        (4) edge                                     (5)
        (5) edge            node[near start] (e2) {} (6);
  \path (3) edge[opacity=0] node[near start] (e3) {} (7)
            edge[opacity=0] node[near start] (e4) {} (8);
  \begin{pgfonlayer}{background}
   \node[inner sep=0pt,transform shape=false,draw=\borderColor,thick,rounded corners,fit = (e1) (e2) (e3) (e4)] {};
  \end{pgfonlayer}
 \end{tikzpicture}
 \caption{A sequence of binary gadgets that forms another binary gadget.
 The circles are assigned $[v, 1, 0]$,
 the squares are assigned $\neq_2$,
 and the triangle is assigned $[a,b,1]$.}
 \label{fig:gadget:long_edge}
\end{figure}

 Hence, we have the signature $[v,1,0]$,
 where $v \in \mathbb{C}$ is for us to choose.
 We construct the gadget in Figure~\ref{fig:gadget:long_edge} with the circles assigned $[v,1,0]$,
 the squares assigned $\neq_2$,
 and the triangle assigned $[a,b,1]$.
 The resulting gadget has signature $[a + 2 b v + v^2, b + v, 1]$,
 which can be verified by the matrix product
 \[
  \begin{bmatrix} v & 1 \\ 1 & 0 \end{bmatrix}
  \begin{bmatrix} 0 & 1 \\ 1 & 0 \end{bmatrix}
  \begin{bmatrix} a & b \\ b & 1 \end{bmatrix}
  \begin{bmatrix} 0 & 1 \\ 1 & 0 \end{bmatrix}
  \begin{bmatrix} v & 1 \\ 1 & 0 \end{bmatrix}
  =
  \begin{bmatrix} a + 2 b v + v^2 & b + v \\ b + v & 1 \end{bmatrix}.
 \]
By setting $v = -b$,
we get $[c,0,1]$,
where $c = a - b^2 \neq 0$.

\begin{figure}[t]
 \centering
 \def\capWidth{5cm}
 \captionsetup[subfigure]{width=\capWidth}
 \subfloat[The tetrahedron gadget with edge signatures given in~\protect\subref{subfig:gadget:triangle_edge}.]
 {
  \makebox[\capWidth][c]
  {
   \begin{tikzpicture}[scale=\scale,transform shape,node distance=\nodeDist,semithick]
    \node[external]  (0)                     {};
    \node[internal]  (1) [right       of=0]  {};
    \node[triangle]  (2) [above right of=1]  {};
    \node[triangle]  (3) [right       of=1]  {};
    \node[triangle]  (4) [below right of=1]  {};
    \node[internal]  (5) [above right of=2]  {};
    \node[external]  (6) [above       of=5]  {};
    \node[triangle]  (7) [below       of=5]  {};
    \node[internal]  (8) [below right of=4]  {};
    \node[external]  (9) [below       of=8]  {};
    \node[triangle] (10) [below right of=5]  {};
    \node[triangle] (11) [above right of=8]  {};
    \node[internal] (12) [above right of=11] {};
    \node[external] (13) [right       of=12] {};
    \path  (1) edge  (2)
               edge  (3)
               edge  (4)
           (2) edge  (5)
           (3) edge (12)
           (4) edge  (8)
           (5) edge  (7)
               edge (10)
           (7) edge  (8)
           (8) edge (11)
          (10) edge (12)
          (11) edge (12);
    \path  (0) edge node[near end]   (e1) {}  (1)
           (5) edge node[near start] (e2) {}  (6)
           (8) edge node[near start] (e3) {}  (9)
          (12) edge node[near start] (e4) {} (13);
    \begin{pgfonlayer}{background}
     \node[inner sep=0pt,transform shape=false,draw=\borderColor,thick,rounded corners,fit = (1) (5) (8) (12) (e1) (e2) (e3) (e4)] {};
    \end{pgfonlayer}
   \end{tikzpicture}
   } \label{subfig:gadget:tetra_main}
 }
 \qquad
 \subfloat[The gadget representing an edge labeled by a triangle in~\protect\subref{subfig:gadget:tetra_main}.]{
  \makebox[\capWidth][c]{
   \begin{tikzpicture}[scale=\scale,transform shape,node distance=\nodeDist,semithick]
    \node[external] (0)                    {};
    \node[square]   (1) [right of=0]       {};
    \node[internal] (2) [right of=1]       {};
    \node[square]   (3) [right of=2]       {};
    \node[external] (4) [right of=3]       {};
    \node[external] (5) [below right of=0] {};
    \node[external] (6) [below right of=5] {};
    \node[external] (7) [below       of=6] {};
    \path (0) edge (1)
          (1) edge (2)
          (2) edge (3)
          (3) edge (4);
   \end{tikzpicture}} \label{subfig:gadget:triangle_edge}}
   \caption{The tetrahedron gadget with each triangle replaced by the edge in~\protect\subref{subfig:gadget:triangle_edge},
   where the circle is assigned $[c,0,1]$ and the squares are assigned $\neq_2$.
   The four circles in~\protect\subref{subfig:gadget:tetra_main} are assigned $[t,1,0,0,0]$.}
 \label{fig:gadget:complicate_tetrahedron}
\end{figure}
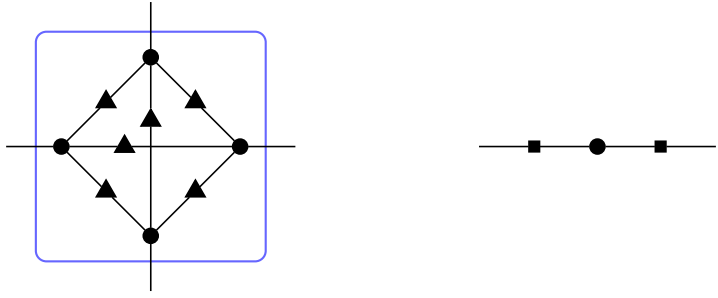

 With this signature $[c,0,1]$,
 we construct the gadget in Figure~\ref{fig:gadget:complicate_tetrahedron},
 where $[c,0,1]$ is assigned to the circle vertex of arity two in Figure~\ref{subfig:gadget:triangle_edge}
 and $\hat{f}$ is assigned to the four circle vertices of arity four in Figure~\ref{subfig:gadget:tetra_main}.
 We get a signature
 \[
  \hat{h} = [3 c^2 + 6 c t^2 + t^4, 3 c t + t^3, c + t^2, t, 1].
 \]
 We note that this computation is reminiscent of matchgate signatures.
 The internal edge function $[1,0,c]$
 (which is a flip from $[c,0,1]$ since both sides are connected to $\neq_2$)
 is a generalized equality signature,
 and the signature $\hat{f}$ on the four circle vertices is a weighted version of the matching function \textsc{At-Most-One}.
 Also note that this computation generalizes a very similar one in Lemma \ref{lem:arity4:double_root}, in which $c=1$.

 The compressed signature matrix of $\hat{h}$ is
 \[
  \widetilde{M_{\hat{h}}}
  =
  \begin{bmatrix}
   3 c^2 + 6 c t^2 + t^4 & 2 (3 c t + t^3) & c + t^2\\
   3 c t + t^3           & 2 (c + t^2)     & t\\
   c + t^2               & 2 t             & 1
  \end{bmatrix}
 \]
 and its determinant is $4 c^3 \neq 0$.
 Thus $\widetilde{M_{\hat{h}}}$ is nonsingular.
 After a holographic transformation by $Z^{-1}$,
 where $Z = \tfrac{1}{\sqrt{2}} \left[\begin{smallmatrix} 1 & 1 \\ i & -i \end{smallmatrix}\right]$,
 the binary disequality $(\neq_2) = [0,1,0]$ is transformed to the binary equality $(=_2) = [1,0,1]$.
 Thus $\Holant([0,1,0] \mid \hat{h})$ is transformed to $\Holant({=}_2 \mid Z^{\otimes 4} \hat{h})$,
 which is the same as $\Holant(Z^{\otimes 4} \hat{h})$.
 Then we are done by Corollary~\ref{cor:arity4:nonsingular_compressed_hard_trans}.

 Now we do the induction step.
 Assume $f$ is of arity $n \geq 5$.
 Since $f$ is non-degenerate,
 $\rd^+(f) \ge 1$.
 If $\rd^+(f) = 1$,
 then we connect the binary $g$ to $f$ to get $f' = \langle f,g \rangle$.
 We have noted that $\rd^+(g) = 2$,
 so $\vd^+(g) = 0$.
 By Lemma~\ref{lem:van:con},
 we have $\rd^+(f')=1$ and $\arity(f') = n-2 \ge 3$.
 Thus $f'$ is vanishing.
 Also $f'$ is non-degenerate,
 for otherwise let $f' = [a,b]^{\otimes (n-2)}$.
 If $[a,b]$ is a multiple of $[1, i]$,
 then $\rd^+(f') \le 0$,
 which is false.
 If $[a,b]$ is not a multiple of $[1, i]$,
 then it can be directly checked that $f' \not \in \mathscr{R}^+_{n-2}$,
 and $\rd^+(f') = n-2 > 1$,
 which is also false.
 Hence $f'$ is a non-degenerate vanishing signature of arity $n-2$,
 so we are done by induction hypothesis.

 Now suppose $\rd^+(f) = t \ge 2$.
 Since $f$ is non-degenerate,
 it is certainly nonzero.
 Since it is vanishing,
 certainly $\vd^+(f) > 0$.
 Hence we can apply Lemma~\ref{lem:van:self}.
 Let $f'$ be obtained from $f$ by a self loop.
 Then $\rd^+(f') = t-1 \ge 1$ and $\arity(f') = n-2$.
 Clearly $f'$ is still vanishing.
 We claim that $f'$ is non-degenerate.
 This follows using the same argument as above.
 If $f'$ were degenerate,
 then either $\rd^+(f') \le 0$ or $\rd^+(f') = \arity(f')$,
 which would contradict $f'$ being a vanishing signature.
 Therefore, we can apply the induction hypothesis.
\end{proof}

Finally, we consider a pair of vanishing signatures of opposite type,
both of arity at least~$3$.
We show that opposite types of vanishing signatures cannot mix.
In other words, vanishing signatures of opposite types, when put together,
lead to $\SHARPP$-hardness.

\begin{lemma} \label{lem:van:plus_and_minus}
 If $f \in \mathscr{V}^+$ and $g \in \mathscr{V}^-$ be two symmetric non-degenerate signatures of arities $\ge 3$,
 then $\Holant(f,g)$ is \numP-hard.
\end{lemma}
\begin{proof}
 Suppose $\rd^+(f) = d$, $\rd^-(g) = d'$, $\arity(f) = n$ and $\arity(g) = n'$, then $2 d < n$ and $2 d' < n'$.
 Under a holographic transformation by $Z = \trans{1}{1}{i}{-i}$,
 we have that
 \begin{align*}
   \holant{{=}_2}{f,g} \equiv_T \holant{{\neq}_2}{\hat{f},\hat{g}},
 \end{align*}
 where $\hat{f}:=(Z^{-1})^{\otimes n} f = [\hat{f}_0, \dots,\hat{f}_d, 0, \dotsc, 0]$
 and $\hat{g}:=(Z^{-1})^{\otimes n'} g = [0, \dotsc, 0, \hat{g}_{d'},\dots,\hat{g}_0]$
 due to Lemma~\ref{lem:vanishing_form_in_Z_basis}.
 Moreover $\hat{f}_d \not = 0$ and $\hat{g}_{d'} \neq 0$.

 If $d\ge 2$, we can do $d'$ many self-loops of $\neq_2$ on $\hat{g}$,
 getting $\hat{g}':=[0,\dots,0,\hat{g}_{d'}]$ of arity $n'-2d'\ge 1$.
 Thus $g':=Z^{\otimes (n'-2d')}\hat{g}'=[1,-i]^{\otimes (n'-2d')}$ up to a nonzero constant.
 We apply Lemma \ref{lem:van:deg} to derive that $\Holant(f,g)$ is \numP-hard.
 If $d'\ge 2$, we can similarly get $[1,i]^{\otimes (n-2d)}$ and apply Lemma~\ref{lem:van:deg}.
 Thus we can assume that $d = d' = 1$.

 So up to nonzero constants, we have $\hat{f} = [a, 1, 0, \dotsc, 0]$ and $\hat{g} = [0, \dotsc, 0, 1, b]$ for some $a, b \in \mathbb{C}$.
 If $a=b=0$, then we apply Lemma \ref{lem:arity4:special_vanishing_case} to conclude \numP-hardness.

 We may thus assume that $a \neq 0$.
 The case of $b\neq 0$ is similar.
 We show that it is always possible to get two such signatures of the same arity $\min\{n, n'\}$.
 Suppose $n > n'$.
 We will construct $[0, 1]^{\otimes (n-n')}$.
 Connecting it back to $\hat{f}$ via $\neq_2$,
 we get a symmetric signature of arity $n'$ consisting of the first $n' + 1$ entries of $\hat{f}$.
 A similar proof works when $n' > n$.
 
 We form a loop from $\hat{f}$ via $\neq_2$.
 It is easy to see that this signature is the degenerate signature $2 [1,0]^{\otimes (n-2)}$.
 Similarly, we can form a loop from $\hat{g}$ and can get $2 [0, 1]^{\otimes (n'-2)}$.
 Thus we have both $[1,0]^{\otimes (n-2)}$ and $[0, 1]^{\otimes (n'-2)}$.
 We can connect all $n'-2$ edges of the second to the first, connected by $\neq_2$.
 This gives us $[1,0]^{\otimes (n-n')}$.
 Similarly to this connection, connect $(n-n')$ many $[0,1]^{\otimes (n'-2)}$ to $n'-3$ many $[1,0]^{\otimes (n-n')}$.
 As $(n-n')(n'-2)-(n-n')(n'-3)=n-n'$,
 the resulting signature is $[0, 1]^{\otimes (n-n')}$.

 Thus we may assume $n = n'$.
 Connecting $[0, 1]^{\otimes (n-2)}$ to $\hat{f} = [a, 1, 0, \dotsc, 0]$ via $\neq_2$ we get $\hat{h} = [a,1,0]$.
 Recall that $a\neq 0$.
 Translating this back by $Z$,
 we have a binary signature $h\notin\mathscr{R}^-_2$. 
 Since $g \in \mathscr{V}^-$, by Lemma~\ref{lem:van:bin}, $\Holant(g,h)$ is \numP-hard.
 Hence $\Holant(f,g)$ is also \numP-hard.
\end{proof}

\section{\texorpdfstring{$\mathscr{A}$}{A}- and \texorpdfstring{$\mathscr{P}$}{P}-transformable Signatures}

In this section,
we investigate the properties of $\mathscr{A}$- and $\mathscr{P}$-transformable signatures.
Throughout,
we define $\alpha = \frac{1+i}{\sqrt{2}} = \sqrt{i} = e^{\frac{\pi i}{4}}$
and use $\mathbf{O}_2(\mathbb{C})$ to denote the group of $2$-by-$2$ orthogonal matrices over $\mathbb{C}$.
Recall that $\mathscr{F}_{123} = \mathscr{F}_1 \union \mathscr{F}_2 \union \mathscr{F}_3$,
where $\mathscr{F}_1$, $\mathscr{F}_2$, and $\mathscr{F}_3$ are defined in Section~\ref{subsec:CSP-Tractable}.
While the main results in this section assume that the signatures involved are symmetric,
we note that some of the lemmas also hold without this assumption.

\subsection{Characterization of \texorpdfstring{$\mathscr{A}$}{A}- and \texorpdfstring{$\mathscr{P}$}{P}-transformable Signatures}

Recall that by definition,
if a set of signatures $\mathcal{F}$ is $\mathscr{A}$-transformable (resp.~$\mathscr{P}$-transformable),
then the binary equality $=_2$ must be simultaneously transformed into $\mathscr{A}$ (resp.~$\mathscr{P}$) along with $\mathcal{F}$.
We first characterize the possible matrices of such a transformation by just considering the transformation of the binary equality.
While there are many binary signatures in $\mathscr{A} \union \mathscr{P}$,
it turns out that it is sufficient to consider only three signatures.

\begin{proposition} \label{prop:matrix_cha}
 Let $T \in \mathbb{C}^{2 \times 2}$ be a matrix.
 Then the following hold:
 \begin{enumerate}
  \item $[1,0,1] T^{\otimes 2} = [1,0,1]$ iff $T \in \mathbf{O}_2(\mathbb{C})$; \label{prop:matrix_cha:case_101}
  \item $[1,0,1] T^{\otimes 2} = [1,0,i]$ iff there exists an $H \in \mathbf{O}_2(\mathbb{C})$
  such that $T = H \left[\begin{smallmatrix} 1 & 0 \\ 0 & \alpha \end{smallmatrix}\right]$; \label{prop:matrix_cha:case_10a}
  \item $[1,0,1] T^{\otimes 2} = [0,1,0]$ iff there exists an $H \in \mathbf{O}_2(\mathbb{C})$
  such that $T = \frac{1}{\sqrt{2}} H \left[\begin{smallmatrix} 1 & 1 \\ i & -i \end{smallmatrix}\right]$. \label{prop:matrix_cha:case_010}
 \end{enumerate}
\end{proposition}

\begin{proof}
 Case~\ref{prop:matrix_cha:case_101} is clear since
 \[[1,0,1] T^{\otimes 2} = [1,0,1] \iff T^\intercal I_2 T = I_2 \iff T^\intercal T = I_2,\]
 the definition of a (2-by-2) orthogonal matrix.
 Now we use this case to prove the others.

 For $M_2 = \left[\begin{smallmatrix} 1 & 0 \\ 0 & \alpha \end{smallmatrix}\right]$
 and $M_3 = Z = \frac{1}{\sqrt{2}} \left[\begin{smallmatrix} 1 & 1 \\ i & -i \end{smallmatrix}\right]$,
 let $T_j = H M_j$ (for $j=2,3$), where $H \in \mathbf{O}_2(\mathbb{C})$.
 Then
 \[[1,0,1] T_j^{\otimes 2} = [1,0,1](H M_j)^{\otimes 2} = [1,0,1]M_j^{\otimes 2} = f_j,\]
 where $f_j$ is the binary signature in case $j$.

 On the other hand, suppose that $[1,0,1] (T_j)^{\otimes 2} = f_j$.
 Then we have
 \[[1,0,1] (T_j M_j^{-1})^{\otimes 2} = f_j(M_j^{-1})^{\otimes 2} = [1,0,1],\]
 so $H = T_j M_j^{-1} \in \mathbf{O}_2(\mathbb{C})$ by case~\ref{prop:matrix_cha:case_101}.
 Thus $T_j = H M_j$ as desired.
\end{proof}

We also need the following lemma;
the proof is direct.

\begin{lemma} \label{lem:rec}
 If a symmetric signature $f = [f_0, f_1, \dotsc, f_n]$ can be expressed in the form $f = a [1, \lambda]^{\otimes n} + b [1, \mu]^{\otimes n}$,
 for some $a, b, \lambda, \mu \in \mathbb{C}$,
 then the $f_k$'s satisfy the recurrence relation $f_{k+2} = (\lambda + \mu) f_{k+1} - \lambda \mu f_k$ for $0 \le k \le n - 2$.
\end{lemma}

To simplify the proof of the characterization of the $\mathscr{A}$-transformable signatures,
we introduce the left and right stabilizer groups of $\mathscr{A}$:
\begin{align*}
 \LStabA &= \{T \in \mathbf{GL}_2(\mathbb{C}) \st T \mathscr{A} \subseteq \mathscr{A}\};\\
 \RStabA &= \{T \in \mathbf{GL}_2(\mathbb{C}) \st \mathscr{A} T \subseteq \mathscr{A}\}.
\end{align*}
In fact, these two groups are equal and coincide with the group of nonsingular signature matrices of binary affine signatures.
More precisely, for a binary signature $f = (f^{00}, f^{01}, f^{10}, f^{11})$,
we define its signature matrix $M_f$ to be
\[
 M_f
 =
 \begin{bmatrix}
  f^{00} & f^{01}\\
  f^{10} & f^{11}
 \end{bmatrix}.
\]
Let
\[
 \binaryA = \{M_f \st f \in \mathscr{A},~\arity(f) = 2, \text{ and } \det(M_f) \neq 0\}
\]
be the set of nonsingular signature matrices of the binary affine signatures.
It is straightforward to verify that $\binaryA$ is closed under multiplication and inverses.
Therefore $\binaryA$ forms a group.

Let $D = \left[\begin{smallmatrix} 1 & 0 \\ 0 & i \end{smallmatrix}\right]$
and $H_2 = \frac{1}{\sqrt{2}} \left[\begin{smallmatrix} 1 & 1 \\ 1 & -1 \end{smallmatrix}\right]$.
Also let $X = \left[\begin{smallmatrix} 0 & 1 \\ 1 & 0 \end{smallmatrix}\right]$
and $Z = \tfrac{1}{\sqrt{2}} \left[\begin{smallmatrix} 1 & 1 \\ i & -i \end{smallmatrix}\right]$.
Note that $Z = D H_2$ and that $D^2  Z = \frac{1}{\sqrt{2}} \left[\begin{smallmatrix} 1 & 1 \\ -i & i \end{smallmatrix}\right] = Z X$,
hence $X = Z^{-1} D^2  Z$.
Furthermore, $D,H_2,X,Z \in \LStabA \intersect \RStabA \intersect \binaryA$, as well as all nonzero scalar multiples of these matrices.

Not only are the groups $\LStabA$, $\RStabA$, and $\binaryA$ equal,
they are generated by $D$ and $H_2$ with a nonzero scalar multiple.

\begin{lemma} \label{lem:StabA:coincide}
 $\LStabA = \RStabA = \binaryA=\mathbb{C}^* \cdot \langle D, H_2 \rangle$.
\end{lemma}

\begin{proof}
  Let 
  \begin{align*}
    \mathbf{S} := \{ S \in \mathbf{GL}_2(\mathbb{C}) \mid \mathscr{F}_{123} S \subseteq \mathscr{F}_{123}\}
  \end{align*}
  be the right stabilizer group of $\mathscr{F}_{123}$.
  Since $\mathscr{F}_{123}\subset\mathscr{A}$, and symmetric signatures are still symmetric under any transformation,
  we have that $\RStabA \subseteq \mathbf{S}$.
  Moreover, as $\mathscr{A}$ is closed under gadget construction, $\binaryA \subseteq \RStabA$.
  Hence, $\binaryA \subseteq \RStabA \subseteq \mathbf{S}$. 
  Together with the fact that $D, H_2 \in \binaryA$,
  we have $\mathbb{C}^* \cdot \langle D, H_2 \rangle \subseteq \binaryA \subseteq \RStabA \subseteq \mathbf{S}$.
  To finish the proof, we show that $\mathbf{S} \subseteq \mathbb{C}^* \cdot \langle D, H_2 \rangle$.
  For $\LStabA$, the proof is similar.

 Consider some $T \in \mathbf{S}$.
 For $f = (=_3)$, we have $f T^{\otimes 3} \in \mathscr{F}_{123}$.
 Then by the form of $\mathscr{F}_{123}$,
 for some $M \in \langle D, H_2 \rangle$,
 chosen to be either $I$,
 or $H_2^\intercal = H_2$,
 or $Z^\intercal = H_2 D$,
 we have $f (TM^{-1})^{\otimes 3} \in \mathscr{F}_1$,
 which is a generalized equality signature.
 Then either $T M^{-1}$ or $T M^{-1}X$ is a diagonal matrix $T' = \lambda \left[\begin{smallmatrix} 1 & 0 \\ 0 & d \end{smallmatrix}\right]$.
 Furthermore, by applying $T'$ to $=_4$, we conclude that $(=_4) T'^{\otimes 4} \in \mathscr{F}_{1}$,
 since it is in $\mathscr{F}_{123}$ but not in $\mathscr{F}_{2} \union \mathscr{F}_{3}$ because $T'$ is diagonal.
 It follows that $d$ is a power of $i$, and hence $\left[\begin{smallmatrix} 1 & 0 \\ 0 & d \end{smallmatrix}\right]$ is a power of $D$.
 Thus $T \in \mathbb{C}^* \cdot \langle D, H_2 \rangle$.
\end{proof}

Since $\LStabA = \RStabA$, we simply write $\StabA$ for this group.
Of course each $T$ under which $\mathcal{F}$ is $\mathscr{A}$-transformable is just a particular solution that can be extended by any element in $\StabA$.

\begin{lemma} \label{lem:StabA}
 Let $\mathcal{F}$ be a set of signatures.
 Then $\mathcal{F}$ is $\mathscr{A}$-transformable under $T$ iff $\mathcal{F}$ is $\mathscr{A}$-transformable under any $T' \in T \StabA$.
\end{lemma}

\begin{proof}
 Sufficiency is trivial since $I_2 \in \StabA$.
 If $\mathcal{F}$ is $\mathscr{A}$-transformable under $T$,
 then by definition, we have $(=_2) T^{\otimes 2} \in \mathscr{A}$ and $\mathcal{F}'=T^{-1} \mathcal{F} \subseteq \mathscr{A}$.
 Let $T' = T M \in T\StabA$ for any $M \in \StabA$.
 It then follows that $(=_2) T'^{\otimes 2} = (=_2) T^{\otimes 2} M^{\otimes 2} \in \mathscr{A} M = \mathscr{A}$
 and $T'^{-1} \mathcal{F} = M^{-1} \mathcal{F}' \subseteq M^{-1} \mathscr{A} = \mathscr{A}$.
 Therefore $\mathcal{F}$ is $\mathscr{A}$-transformable under any $T'\in T\StabA$.
\end{proof}

After restricting by Proposition~\ref{prop:matrix_cha} and normalizing by Lemma~\ref{lem:StabA},
one only needs to check a small subset of $\mathbf{GL}_2(\mathbb{C})$ to determine if $\mathcal{F}$ is $\mathscr{A}$-transformable.

\begin{lemma} \label{lem:affine:trans}
 Let $\mathcal{F}$ be a set of signatures.
 Then $\mathcal{F}$ is $\mathscr{A}$-transformable iff
 there exists an $H \in \mathbf{O}_2(\mathbb{C})$ such that $\mathcal{F} \subseteq H \mathscr{A}$
 or $\mathcal{F} \subseteq H \left[\begin{smallmatrix} 1 & 0 \\ 0 & \alpha \end{smallmatrix}\right] \mathscr{A}$.
\end{lemma}

\begin{proof}
 Sufficiency is easily verified by checking that $=_2$ is transformed into $\mathscr{A}$ in both cases.
 In particular, $H$ leaves $=_2$ unchanged.
 
 If $\mathcal{F}$ is $\mathscr{A}$-transformable, then by definition,
 there exists a matrix $T$ such that $(=_2) T^{\otimes 2} \in \mathscr{A}$ and $T^{-1} \mathcal{F} \subseteq \mathscr{A}$.
 Since $=_2$ is non-degenerate and symmetric, $(=_2) T^{\otimes 2} \in \mathscr{A}$ is equivalent to 
 $(=_2) T^{\otimes 2} \in \mathscr{F}_{123}$.

 Any signature in $\mathscr{F}_{123}$ is expressible as $c (v_1^{\otimes n} + i^t v_2^{\otimes n})$,
 where $t \in \{0,1,2,3\}$ and $(v_1, v_2)$ is a pair of vectors in the set
 \begin{equation*}
 \left\{
  \left(\begin{bmatrix} 1 \\ 0 \end{bmatrix}, \begin{bmatrix} 0 \\  1 \end{bmatrix}\right),
  \left(\begin{bmatrix} 1 \\ 1 \end{bmatrix}, \begin{bmatrix} 1 \\ -1 \end{bmatrix}\right),
  \left(\begin{bmatrix} 1 \\ i \end{bmatrix}, \begin{bmatrix} 1 \\ -i \end{bmatrix}\right)\right\}.
 \end{equation*}
 We use $\StabA$ to further normalize these three sets by Lemma~\ref{lem:StabA}.
 In particular, $\mathscr{F}_1 = H_2\mathscr{F}_2$ and $\mathscr{F}_1 = (D H_2)^{-1}\mathscr{F}_3$.
 Furthermore, the binary signatures in $\mathscr{F}_1$ are just the four signatures $[1,0,1]$, $[1,0,i]$, $[1,0,-1]$, and $[1,0,-i]$ up to a scalar.
 We also normalize these four as $[1,0,1] = [1,0,-1] D^{\otimes 2}$ and $[1,0,i] = [1,0,-i] D^{\otimes 2}$.
 Hence $\mathcal{F}$ being $\mathscr{A}$-transformable implies that
 there exists a matrix $T$ such that $(=_2) T^{\otimes 2} \in \{[1,0,1],[1,0,i]\}$ and $T^{-1}\mathcal{F} \subseteq \mathscr{A}$.
 Now we apply Proposition~\ref{prop:matrix_cha}.
 \begin{enumerate}
  \item If $(=_2) T^{\otimes 2} = [1,0,1]$,
  then by case~\ref{prop:matrix_cha:case_101} of Proposition~\ref{prop:matrix_cha},
  we have $T\in\mathbf{O}_2(\mathbb{C})$.
  Therefore $\mathcal{F} \subseteq H \mathscr{A}$ where $H = T \in \mathbf{O}_2(\mathbb{C})$.
  \item If $(=_2) T^{\otimes 2} = [1,0,i]$,
  then by case~\ref{prop:matrix_cha:case_10a} of Proposition~\ref{prop:matrix_cha},
  there exists an $H \in \mathbf{O}_2(\mathbb{C})$ such that $T = H \left[\begin{smallmatrix} 1 & 0 \\ 0 & \alpha \end{smallmatrix}\right]$.
  Therefore $\mathcal{F} \subseteq T \mathscr{A} = H \left[\begin{smallmatrix} 1 & 0 \\ 0 & \alpha \end{smallmatrix}\right]\mathscr{A}$
  where $H\in\mathbf{O}_2(\mathbb{C})$.
 \end{enumerate}
 This completes the proof.
\end{proof}

Using these two lemmas,
we can characterize all $\mathscr{A}$-transformable signatures.
We first define the three sets $\mathscr{A}_1$, $\mathscr{A}_2$, and $\mathscr{A}_3$.

\begin{definition} \label{def:A1}
 A symmetric signature $f$ of arity $n$ is in $\mathscr{A}_1$ if
 there exists an $H \in \mathbf{O}_2(\mathbb{C})$ and a nonzero constant $c \in \mathbb{C}$ such that
 $f = c H^{\otimes n} \left(\left[\begin{smallmatrix} 1 \\  1 \end{smallmatrix}\right]^{\otimes n}
                    + \beta \left[\begin{smallmatrix} 1 \\ -1 \end{smallmatrix}\right]^{\otimes n}\right)$,
  where $\beta = \alpha^{tn+2r}$ for some $r \in \{0,1,2,3\}$ and $t \in \{0,1\}$.
\end{definition}

When such an $H$ exists,
we say that $f \in \mathscr{A}_1$ with transformation $H$.
If $f \in \mathscr{A}_1$ with $I_2$,
then we say $f$ is in the canonical form of $\mathscr{A}_1$.
If $f$ is in the canonical form of $\mathscr{A}_1$,
then by Lemma~\ref{lem:rec},
for any $0 \le k \le n-2$,
we have $f_{k+2} = f_k$ and one of the following holds:
\begin{itemize}
 \item $f_0 = 0$, or
 \item $f_1 = 0$, or
 \item $f_1 = \pm i f_0 \neq 0$, or
 \item $n$ is odd and $f_1 = \pm (1 \pm \sqrt{2}) i f_0 \ne 0$ (all four sign choices are permissible).
\end{itemize}
Notice that when $n$ is odd and $t = 1$ in Definition~\ref{def:A1},
it has some complication as described by the factor $\alpha^{t n + 2 r}$.

\begin{definition} \label{def:single:A2}
 A symmetric signature $f$ of arity $n$ is in $\mathscr{A}_2$ if
 there exists an $H \in \mathbf{O}_2(\mathbb{C})$ and a nonzero constant $c \in \mathbb{C}$ such that
 $f = c H^{\otimes n} \left(\left[\begin{smallmatrix} 1 \\  i \end{smallmatrix}\right]^{\otimes n}
                          + \left[\begin{smallmatrix} 1 \\ -i \end{smallmatrix}\right]^{\otimes n}\right)$.
\end{definition}

Similarly,
when such an $H$ exists,
we say that $f \in \mathscr{A}_2$ with transformation $H$.
If $f \in \mathscr{A}_2$ with $I_2$,
then we say $f$ is in the canonical form of $\mathscr{A}_2$.
If $f$ is in the canonical form of $\mathscr{A}_2$,
then by Lemma~\ref{lem:rec},
for any $0 \le k \le n-2$,
we have $f_{k+2} = -f_k$.
Since $f$ is non-degenerate,
$f_1 \not = \pm i f_0$ is implied.

It is worth noting that
$\{\left[\begin{smallmatrix} 1 \\  i \end{smallmatrix}\right],
   \left[\begin{smallmatrix} 1 \\ -i \end{smallmatrix}\right]\}$
is setwise invariant up to scale under any transformation in $\mathbf{O}_2(\mathbb{C})$ up to nonzero constants.
That is,
these vectors are the eigenvectors of orthogonal matrices.
Thus for any $H \in \mathbf{O}_2(\mathbb{C})$,
we can write $\trans{1}{1}{i}{-i}^{-1} H \trans{1}{1}{i}{-i} = D$,
where $D$ is either a diagonal or anti-diagonal matrix.
It is also helpful to view this equation as $H \trans{1}{1}{i}{-i} = \trans{1}{1}{i}{-i} D$.

Using this fact,
the following lemma gives a characterization of $\mathscr{A}_2$.
It says that any signature in $\mathscr{A}_2$ is essentially in canonical form.

\begin{lemma} \label{lem:single:P2}
 Let $f$ be a symmetric signature of arity $n$.
 Then $f \in \mathscr{A}_2$ iff
 $f = c \left(\left[\begin{smallmatrix} 1 \\  i \end{smallmatrix}\right]^{\otimes n}
      + \beta \left[\begin{smallmatrix} 1 \\ -i \end{smallmatrix}\right]^{\otimes n}\right)$
 for some nonzero constants $c, \beta \in \mathbb{C}$.
\end{lemma}

\begin{proof}
 Assume that $f = c \left(\tbcolvec{1}{i}^{\otimes n}
                  + \beta \tbcolvec{1}{-i}^{\otimes n}\right)$ for some $c, \beta \ne 0$.
 Consider the orthogonal transformation $H = \tbmatrix{a}{b}{b}{-a}$,
 where $a = \frac{1}{2} \left(\beta^{\frac{1}{2n}} + \beta^{-\frac{1}{2n}}\right)$ and $b = \frac{1}{2i} \left(\beta^{\frac{1}{2n}} - \beta^{-\frac{1}{2n}}\right)$.
 We pick $a$ and $b$ in this way so that $a + b i = \beta^{\frac{1}{2n}}$, $a - b i = \beta^{-\frac{1}{2n}}$, and $(a + b i) (a - b i) = a^2 + b^2 = 1$.
 Also $\left(\frac{a + b i}{a - b i}\right)^n = \beta$.
 Then
 \begin{align*}
  H^{\otimes n} f
  &= c \left(\begin{bmatrix} a + b i \\ -a i + b \end{bmatrix}^{\otimes n}
     + \beta \begin{bmatrix} a - b i \\  a i + b \end{bmatrix}^{\otimes n}\right)\\
  &= c \left((a + b i)^n       \begin{bmatrix} 1 \\ -i \end{bmatrix}^{\otimes n}
           + (a - b i)^n \beta \begin{bmatrix} 1 \\  i \end{bmatrix}^{\otimes n}\right)\\
  &= c \sqrt{\beta} \left(\begin{bmatrix} 1 \\ -i \end{bmatrix}^{\otimes n}
                        + \begin{bmatrix} 1 \\  i \end{bmatrix}^{\otimes n}\right),
 \end{align*}
 so $f$ can be written as
 \[
  f = c \sqrt{\beta} (H^{-1})^{\otimes n} \left(\begin{bmatrix} 1 \\  i \end{bmatrix}^{\otimes n}
                                              + \begin{bmatrix} 1 \\ -i \end{bmatrix}^{\otimes n}\right).
 \]
 Therefore $f \in \mathscr{A}_2$.
 
 On the other hand,
 the desired form $f = c (\tbcolvec{1}{i}^{\otimes n} + \beta \tbcolvec{1}{-i}^{\otimes n})$
 follows from the fact that $\{\tbcolvec{1}{i}, \tbcolvec{1}{-i}\}$ is fixed setwise under any orthogonal transformation up to nonzero constants.
\end{proof}

\begin{definition}
 A symmetric signature $f$ of arity $n$ is in $\mathscr{A}_3$ if
 there exists an $H \in \mathbf{O}_2(\mathbb{C})$ and a nonzero constant $c \in \mathbb{C}$ such that
 $f = c H^{\otimes n} \left(\left[\begin{smallmatrix} 1 \\  \alpha \end{smallmatrix}\right]^{\otimes n}
                      + i^r \left[\begin{smallmatrix} 1 \\ -\alpha \end{smallmatrix}\right]^{\otimes n}\right)$
 for some $r \in \{0,1,2,3\}$.
\end{definition}

Again,
when such an $H$ exists,
we say that $f\in\mathscr{A}_3$ with transformation $H$.
If $f \in \mathscr{A}_3$ with $I_2$,
then we say $f$ is in the canonical form of $\mathscr{A}_3$.
If $f$ is in the canonical form of $\mathscr{A}_3$,
then by Lemma~\ref{lem:rec},
for any $0 \le k \le n-2$,
we have $f_{k+2} = i f_k$ and one of the following holds:
\begin{itemize}
 \item $f_0 = 0$, or
 \item $f_1 = 0$, or
 \item $f_1 = \pm \alpha i f_0 \neq 0$.
\end{itemize}

\vspace*{\baselineskip}

Now we characterize the $\mathscr{A}$-transformable signatures.

\begin{lemma} \label{lem:cha:affine}
 Let $f$ be a non-degenerate symmetric signature.
 Then $f$ is $\mathscr{A}$-trans-formable iff $f \in \mathscr{A}_1 \union \mathscr{A}_2 \union \mathscr{A}_3$.
\end{lemma}

\begin{proof}
  Assume that $f$ is $\mathscr{A}$-transformable of arity $n$.
  By applying Lemma \ref{lem:affine:trans} to $\{f\}$,
  there exists an $H \in \mathbf{O}_2(\mathbb{C})$ such that $f \in H \mathscr{A}$ or $f\in H\trans{1}{0}{0}{\alpha} \mathscr{A}$.
  This is equivalent to $(H^{-1})^{\otimes n} f \in \mathscr{A}$ or $(H^{-1})^{\otimes n}f \in \trans{1}{0}{0}{\alpha} \mathscr{A}$.
  Since $f$ is non-degenerate and symmetric, we can replace $\mathscr{A}$ in the previous expressions with $\mathscr{F}_{123}$.
  Now we consider all possible cases.
  Let $\hat{f} = (H^{-1})^{\otimes n} f$.
  \begin{enumerate}
    \item If $\hat{f} \in \mathscr{F}_1$, then $T^{\otimes n}\hat{f}$ is in the canonical form of $\mathscr{A}_1$,
      where $T=\frac{1}{\sqrt{2}} \trans{1}{1}{1}{-1}\in \mathbf{O}_2(\mathbb{C})$.
    \item If $\hat{f} \in \mathscr{F}_2$, then $\hat{f}$ is already in the canonical form of $\mathscr{A}_1$.
      Let $T=I_2$ in this case.
    \item If $\hat{f} \in \mathscr{F}_3$, then $\hat{f}$ already has the equivalent form of $\mathscr{A}_2$ given by Lemma~\ref{lem:single:P2}.
      Let $T=I_2$ in this case.      
    \item If $\hat{f} \in \trans{1}{0}{0}{\alpha} \mathscr{F}_1$, then $T^{\otimes n}\hat{f}$ is in the canonical form of $\mathscr{A}_1$,
      where $T=\frac{1}{\sqrt{2}} \trans{1}{1}{1}{-1}\in \mathbf{O}_2(\mathbb{C})$.
    \item If $\hat{f} \in \trans{1}{0}{0}{\alpha} \mathscr{F}_2$, then $\hat{f}$ is already in the canonical form of $\mathscr{A}_3$.
      Let $T=I_2$ in this case.
    \item If $\hat{f} \in \trans{1}{0}{0}{\alpha} \mathscr{F}_3$, 
      then $\hat{f}$ has the form $\tbcolvec{1}{\alpha^3}^{\otimes n} + i^r \tbcolvec{1}{-\alpha^3}^{\otimes n}$,
      and $T^{\otimes n}\hat{f}$ is in the canonical form of $\mathscr{A}_3$,
      where $T=\trans{0}{-1}{1}{0}\in \mathbf{O}_2(\mathbb{C})$.
      To see this,
      \begin{align*}
        \hspace{-0.6cm}
        \begin{bmatrix} 0 & -1 \\ 1 & 0 \end{bmatrix}^{\otimes n} 
        \left(\begin{bmatrix} 1 \\  \alpha^3 \end{bmatrix}^{\otimes n} + i^r \begin{bmatrix} 1 \\ -\alpha^3 \end{bmatrix}^{\otimes n}\right)
        &=    \begin{bmatrix} -\alpha^3 \\ 1 \end{bmatrix}^{\otimes n} + i^r \begin{bmatrix}  \alpha^3 \\ 1 \end{bmatrix}^{\otimes n}\\
        &= \left(-\alpha^{3}\right)^n \left(\begin{bmatrix} 1 \\ -\frac{1}{\alpha^3} \end{bmatrix}^{\otimes n}
                               + (-1)^n i^r \begin{bmatrix} 1 \\  \frac{1}{\alpha^3} \end{bmatrix}^{\otimes n}\right)\\
        &= \left(-\alpha^{3}\right)^n \left(\begin{bmatrix} 1 \\  \alpha \end{bmatrix}^{\otimes n}
                                 + i^{2n+r} \begin{bmatrix} 1 \\ -\alpha \end{bmatrix}^{\otimes n}\right).
      \end{align*}
  \end{enumerate}
  Let $\hat{f}' = T^{\otimes n}\hat{f}$, where $T\in \mathbf{O}_2(\mathbb{C})$ is given in each case.
  Then $\hat{f}'$ is $f$ after an orthogonal transformation $TH^{-1}$.
  As shown above, $\hat{f}'$ is in the canonical form of $\mathscr{A}_1$ or $\mathscr{A}_3$, 
  or is in the equivalent form of $\mathscr{A}_2$ by Lemma \ref{lem:single:P2}.
  Hence $f\in\mathscr{A}_1\cup\mathscr{A}_2\cup\mathscr{A}_3$.

  Conversely, if there exists a matrix $H \in \mathbf{O}_2(\mathbb{C})$ such that
  $H^{\otimes n} f$ is in one of the canonical forms of $\mathscr{A}_1$, $\mathscr{A}_2$, or $\mathscr{A}_3$,
  then one can directly check that $f$ is $\mathscr{A}$-transformable.
  In fact, transformations we applied above are all invertible.
\end{proof}

We also have a similar characterization for $\mathscr{P}$-transformable signatures.
We define the stabilizer group of $\mathscr{P}$ similar to $\StabA$.
It is easy to see the left and right stabilizers coincide, which we denote by $\StabP$.
Furthermore, $\StabP$ is generated by nonzero scalar multiples of matrices of the form $\left[\begin{smallmatrix} 1 & 0 \\ 0 & \nu \end{smallmatrix}\right]$
for any nonzero $\nu \in \mathbb{C}$ and $X = \left[\begin{smallmatrix} 0 & 1 \\ 1 & 0 \end{smallmatrix}\right]$.

\begin{lemma} \label{lem:product:trans}
 Let $\mathcal{F}$ be a set of signatures.
 Then $\mathcal{F}$ is $\mathscr{P}$-transformable iff there exists an $H \in \mathbf{O}_2(\mathbb{C})$ such that $\mathcal{F} \subseteq H \mathscr{P}$
 or $\mathcal{F} \subseteq H \left[\begin{smallmatrix} 1 & 1 \\ i & -i \end{smallmatrix}\right] \mathscr{P}$.
\end{lemma}

\begin{proof}
 Sufficiency is easily verified by checking that $=_2$ is transformed into $\mathscr{P}$ in both cases.
 In particular,
 $H$ leaves $=_2$ unchanged.
 
 If $\mathcal{F}$ is $\mathscr{P}$-transformable,
 then by definition,
 there exists a matrix $T$ such that $(=_2) T^{\otimes 2} \in \mathscr{P}$ and $T^{-1} \mathcal{F} \subseteq \mathscr{P}$.
 The non-degenerate binary signatures in $\mathscr{P}$ are either $[0,1,0]$ or of the form $[1,0,\nu]$,
 up to a scalar.
 However,
 notice that $[1,0,1] = [1,0,\nu]
 \left[\begin{smallmatrix} 1 & 0 \\ 0 & \nu^{-\frac{1}{2}} \end{smallmatrix}\right]^{\otimes 2}$
 and $\left[\begin{smallmatrix} 1 & 0 \\ 0 & \nu^{-\frac{1}{2}} \end{smallmatrix}\right] \in \StabP$.
 Thus,
 we only need to consider $[1,0,1]$ and $[0,1,0]$.
 Now we apply Proposition~\ref{prop:matrix_cha}.
 \begin{enumerate}
  \item If $(=_2) T^{\otimes 2} = [1,0,1]$,
  then by case~\ref{prop:matrix_cha:case_101} of Proposition~\ref{prop:matrix_cha}, we have $T \in \mathbf{O}_2(\mathbb{C})$.
  Therefore $\mathcal{F} \subseteq H \mathscr{P}$ where $H = T \in \mathbf{O}(\mathbb{C})$.
  \item If $(=_2) T^{\otimes 2} = [0,1,0]$, 
  then by case~\ref{prop:matrix_cha:case_010} of Proposition~\ref{prop:matrix_cha}, 
  there exists an $H \in \mathbf{O}_2(\mathbb{C})$ such that $T = \frac{1}{\sqrt{2}} H \left[\begin{smallmatrix} 1 & 1 \\ i & -i \end{smallmatrix}\right]$.
  Therefore $\mathcal{F} \subseteq H \left[\begin{smallmatrix} 1 & 1 \\ i & -i \end{smallmatrix}\right]\mathscr{P}$
  where $H \in \mathbf{O}_2(\mathbb{C})$.
 \end{enumerate}
\end{proof}

We also have similar definitions of the sets $\mathscr{P}_1$ and $\mathscr{P}_2$.

\begin{definition}
 A symmetric signature $f$ of arity $n$ is in $\mathscr{P}_1$ if
 there exists $H \in \mathbf{O}_2(\mathbb{C})$ and a nonzero constant $c \in \mathbb{C}$ such that
 $f = c H^{\otimes n} \left(\left[\begin{smallmatrix} 1 \\  1 \end{smallmatrix}\right]^{\otimes n}
                    + \beta \left[\begin{smallmatrix} 1 \\ -1 \end{smallmatrix}\right]^{\otimes n}\right)$,
 where $\beta \neq 0$.
\end{definition}

When such an $H$ exists, we say that $f \in \mathscr{P}_1$ with transformation $H$.
If $f \in \mathscr{P}_1$ with $I_2$, then we say $f$ is in the canonical form of $\mathscr{P}_1$.
If $f$ is in the canonical form of $\mathscr{P}_1$,
then by Lemma~\ref{lem:rec}, for any $0 \le k \le n-2$, we have $f_{k+2} = f_k$.
Since $f$ is non-degenerate, $f_1 \not = \pm f_0$ is implied.

It is easy to check that $\mathscr{A}_1 \subset \mathscr{P}_1$.
The corresponding definition for $\mathscr{P}_2$ coincides with Definition~\ref{def:single:A2} for $\mathcal{A}_2$.
In other words, we define $\mathscr{P}_2 = \mathcal{A}_2$.

Now we characterize the $\mathscr{P}$-transformable signatures as we did for the $\mathscr{A}$-trans-formable signatures in Lemma~\ref{lem:cha:affine}.

\begin{lemma} \label{lem:cha:product}
 Let $f$ be a non-degenerate symmetric signature.
 Then $f$ is $\mathscr{P}$-trans-formable iff $f \in \mathscr{P}_1 \union \mathscr{P}_2$.
\end{lemma}

\begin{proof}
  Assume that $f$ is $\mathscr{P}$-transformable of arity $n$.
  By applying Lemma~\ref{lem:product:trans} to $\{f\}$,
  there exists an $H \in \mathbf{O}_2(\mathbb{C})$ such that $f \in H \mathscr{P}$
  or $f \in H \trans{1}{1}{i}{-i} \mathscr{P}$.
  This is equivalent to $(H^{-1})^{\otimes n}f \in \mathscr{P}$
  or $(H^{-1})^{\otimes n}f \in \trans{1}{1}{i}{-i} \mathscr{P}$.
  Let $\hat{f}=(H^{-1})^{\otimes n}f$.
  It is sufficient to show that $\hat{f}\in\mathscr{P}_1$ or $\mathscr{P}_2$.
 
  The symmetric signatures in $\mathscr{P}$ take the form $[0,1,0]$, or $[a,0,\dots,0,b] = a [1,0]^{\otimes n} + b [0,1]^{\otimes n}$,
  where $a b \ne 0$ since $f$ is non-degenerate.
  Now we consider all possible cases.
  \begin{enumerate}
    \item If $\hat{f} = [0,1,0]$, then $\hat{f} = \tfrac{1}{2 i} \left(\tbcolvec{1}{i}^{\otimes 2} - \tbcolvec{1}{-i}^{\otimes 2}\right)$,
      which is the equivalent form of $\mathscr{P}_2 = \mathscr{A}_2$ given by Lemma~\ref{lem:single:P2}.
    \item If $\hat{f} = a \tbcolvec{1}{0}^{\otimes n} + b \tbcolvec{0}{1}^{\otimes n}$,
      then a further transformation by $\frac{1}{\sqrt{2}} \trans{1}{1}{1}{-1} \in \mathbf{O}_2(\mathbb{C})$
      puts $\hat{f}$ into the canonical form of $\mathscr{P}_1$.
    \item If $\hat{f} = \trans{1}{1}{i}{-i}^{\otimes 2} [0,1,0]^{\texttt T}
      = 2 [1,0,1] = \tbcolvec{1}{1}^{\otimes 2} + \tbcolvec{1}{-1}^{\otimes 2}$,
      then $\hat{f}$ is already in the canonical form of $\mathscr{P}_1$.
    \item If $\hat{f} = \trans{1}{1}{i}{-i}^{\otimes n} \left(a \tbcolvec{1}{0}^{\otimes n} + b \tbcolvec{0}{1}^{\otimes n}\right)$,
      then $\hat{f}$ is already of the equivalent form of $\mathscr{P}_2 = \mathscr{A}_2$ given by Lemma~\ref{lem:single:P2}.
  \end{enumerate}
 
  Conversely, if there exists a matrix $H \in \mathbf{O}_2(\mathbb{C})$ such that
  $H^{\otimes n} f$ is in one of the canonical forms of $\mathscr{P}_1$ or $\mathscr{P}_2$,
  then one can directly check that $f$ is $\mathscr{P}$-transformable.
  In fact, the transformations that we applied above are all invertible.
\end{proof}

Combining Lemma~\ref{lem:cha:affine} and Lemma~\ref{lem:cha:product},
we have a necessary and sufficient condition for a single non-degenerate signature to be $\mathscr{A}$- or $\mathscr{P}$-transformable.

\begin{corollary} \label{cor:single:AP-trans_by_sets}
 Let $f$ be a non-degenerate signature.
 Then $f$ is $\mathscr{A}$- or $\mathscr{P}$-transformable iff $f \in \mathscr{P}_1 \union \mathscr{P}_2 \union \mathscr{A}_3$.
\end{corollary}

Notice that our definitions of $\mathscr{P}_1$, $\mathscr{P}_2$, and $\mathscr{A}_3$ each involve an orthogonal transformation.
For any single signature $f \in \mathscr{P}_1 \union \mathscr{P}_2 \union \mathscr{A}_3$, $\Holant(f)$ is tractable.
However, this does not imply that $\Holant(\mathscr{P}_1)$, $\Holant(\mathscr{P}_2)$, or $\Holant(\mathscr{A}_3)$ is tractable.
One can check, using Theorem~\ref{thm:main},
that $\Holant(\mathscr{P}_2)$ is tractable while $\Holant(\mathscr{P}_1)$ and $\Holant(\mathscr{A}_3)$ are $\SHARPP$-hard.

\subsection{Dichotomies when \texorpdfstring{$\mathscr{A}$}{A}- or \texorpdfstring{$\mathscr{P}$}{P}-transformable Signatures Appear}

Our characterizations of $\mathscr{A}$-transformable signatures in Lemma~\ref{lem:cha:affine}
and $\mathscr{P}$-transformable signatures in Lemma~\ref{lem:cha:product} are up to transformations in $\mathbf{O}_2(\mathbb{C})$.
Since an orthogonal transformation never changes the complexity of the problem, in the proofs of following lemmas,
we assume any signature in $\mathscr{A}_i$ for $i=1,2,3$, or $\mathscr{P}_j$ for $j=1,2$, 
is already in the canonical form.

\begin{lemma} \label{lem:dic:p1}
 Let $\mathcal{F}$ be a set of symmetric signatures.
 Suppose $\mathcal{F}$ contains a non-degenerate signature $f \in \mathscr{P}_1$ of arity $n \ge 3$.
 Then $\Holant(\mathcal{F})$ is $\SHARPP$-hard unless $\mathcal{F}$ is $\mathscr{P}$-transformable or $\mathscr{A}$-transformable.
\end{lemma}

\begin{proof}
 By assumption, for any $0 \le k \le n-2$, $f_{k+2} = f_k$ and $f_1 \neq \pm f_0$ since $f$ is not degenerate.
 We can express $f$ as
 \[f = a_0 \begin{bmatrix} 1 \\ 1 \end{bmatrix}^{\otimes n} + a_1 \begin{bmatrix} 1 \\ -1 \end{bmatrix}^{\otimes n},\]
 where $a_0 = (f_0 + f_1) / 2$ and $a_1 = (f_0 - f_1) / 2$.
%
 For this $f$,
 we can further perform an orthogonal transformation by
 $H_2 = \frac{1}{\sqrt{2}}\left[\begin{smallmatrix} 1 & 1 \\ 1 & -1 \end{smallmatrix}\right]$
 so that $f$ is transformed into the generalized equality signature 
 $2^{n/2} [a_0, 0, \dotsc, 0, a_1]$ of arity $n$,
 where $a_0 a_1 \neq 0$.
 By Lemma~\ref{lem:simple_interpolation:dic:p1},
 we can obtain $=_4$, the arity~$4$ equality signature.
 With this signature,
 we can realize any equality signature of even arity.
 Thus,
 $\CSP^2(H_2 \mathcal{F}) \le_T \Holant(\mathcal{F})$.
 
 Now we apply Theorem~\ref{thm:CSPd},
 the $\CSP^d$ dichotomy,
 to the set $H_2 \mathcal{F}$.
 If this problem is $\SHARPP$-hard, then $\Holant(\mathcal{F})$ is $\SHARPP$-hard as well.
 Otherwise,
 this problem is $\CSP^2$ tractable.
 Therefore,
 there exists some $T$ of the form 
 $\left[\begin{smallmatrix} 1 & 0 \\ 0 & \alpha^k \end{smallmatrix}\right]$,
 where the integer $k \in \{0, 1, \dotsc, 7\}$,
 such that $T H_2 \mathcal{F}$ is a subset of $\mathscr{A}$ or $\mathscr{P}$.

 If $T H_2 \mathcal{F} \subseteq \mathscr{P}$, then we have $H_2 \mathcal{F} \subseteq T^{-1} \mathscr{P}$.
 Notice that $T\in\StabP$,
 so $T^{-1}\mathscr{P}=\mathscr{P}$.
 Thus, $\mathcal{F}$ is $\mathscr{P}$-transformable under this $H_2$ transformation.
 Otherwise, $T H_2 \mathcal{F} \subseteq \mathscr{A}$.
 It is easy to verify that $(=_2) ((T H_2)^{-1})^{\otimes 2}$ is $ [1,0,i^{-k}] \in \mathscr{A}$.
 Thus, $\mathcal{F}$ is $\mathscr{A}$-transformable under this $T H_2$ transformation.
\end{proof}

\begin{lemma} \label{lem:dic:p2}
 Let $\mathcal{F}$ be a set of symmetric signatures.
 Suppose $\mathcal{F}$ contains a non-degenerate signature $f \in \mathscr{P}_2$ of arity $n \ge 3$.
 Then $\Holant(\mathcal{F})$ is $\SHARPP$-hard unless $\mathcal{F}$ is $\mathscr{P}$-transformable or $\mathscr{A}$-transformable.
\end{lemma}

\begin{proof}
 By assumption, for any $0 \le k \le n-2$, $f_{k+2} = -f_k$ and $f_1 \neq \pm i f_0$ since $f$ is not degenerate.
 We can express $f$ as
 \[f = a_0 \begin{bmatrix} 1 \\ i \end{bmatrix}^{\otimes n} + a_1 \begin{bmatrix} 1 \\ -i \end{bmatrix}^{\otimes n},\]
 where $a_0 = (f_0 + i f_1) / 2$ and $a_1 = (f_0 - i f_1) / 2$, and $a_0 a_1 \neq 0$.
 Then under the holographic transformation
 $Z' = \left[\begin{smallmatrix} a_0^{1/n} & a_1^{1/n} \\ a_0^{1/n} i  & -a_1^{1/n} i \end{smallmatrix}\right]^{-1}$,
 we have
 \[
  Z'^{\otimes n} f
  = (=_n)
  = \begin{bmatrix} 1 \\ 0 \end{bmatrix}^{\otimes n} + \begin{bmatrix} 0 \\ 1 \end{bmatrix}^{\otimes n}
 \]
 and
 \begin{align*}
  \holant{{=}_2}{\mathcal{F} \union \{f\}}
  &\equiv_T \holant{[1,0,1] (Z'^{-1})^{\otimes 2}}{Z' \mathcal{F} \union \{Z'^{\otimes n} f\}}\\
  &\equiv_T \holant{(1-i) a_0^{1/n} a_1^{1/n} [0,1,0]}{Z' \mathcal{F} \union \{{=}_n\}}.
 \end{align*}
 Thus, we have a bipartite graph with $=_n$ on the right and $(\neq_2) = [0,1,0]$ on the left up to a nonzero scalar,
 so all equality signatures of arity a multiple of $n$ are realizable on the right side.
 To see this,
 first notice that we can move equality signatures from the right side to the left side using the binary disequality
 because the binary disequality just reverses signatures (i.e.~exchanges the~$0$ and~$1$ input bits),
 which leaves the equality signatures unchanged.
 Now we do an induction.
 Suppose we can realize the equality $=_{(k-1)n}$ on the right side for some integer $k>1$.
 We create a new signature on the right using one $=_{(k-1)n}$ and two $=_n$ on the right and one $=_n$ on the left.
 Since $n \geq 3$,
 we can connect one wire of the left $=_n$ to each of the three equality signatures on the right.
 The remaining wires of the left $=_n$ can be connected arbitrarily to the signatures on the right.
 The resulting signature is an equality of arity $(k-1) n + 2 n - n = k n$.
 Since we have $=_{kn}$ on both sides for any integer $k \ge 1$,
 $\CSP^n(Z' \mathcal{F}) \le_T \Holant(\mathcal{F})$.
 
 Now we apply Theorem~\ref{thm:CSPd},
 the $\CSP^d$ dichotomy,
 to the set $Z' \mathcal{F}$.
 If this problem is $\SHARPP$-hard, then $\Holant(\mathcal{F})$ is $\SHARPP$-hard as well.
 Otherwise, this problem is $\CSP^n$ tractable.
 Let $\omega$ be a primitive $4 n$-th root of unity.
 Then under the holographic transformation $T = \left[\begin{smallmatrix} 1 & 0 \\ 0 & \omega^k \end{smallmatrix}\right]$ for some integer $k$,
 $T Z' \mathcal{F}$ is a subset of $\mathscr{A}$ or $\mathscr{P}$.
 If $T Z' \mathcal{F} \subseteq \mathscr{P}$,
 then we have $Z' \mathcal{F} \subseteq T^{-1} \mathscr{P}$.
 Notice that $T \in \StabP$,
 so $T^{-1} \mathscr{P} = \mathscr{P}$.
 Thus,
 $\mathcal{F}$ is $\mathscr{P}$-transformable under this $Z'$ transformation.

 Otherwise, $T Z' \mathcal{F} \subseteq \mathscr{A}$.
 It is easy to verify that $(=_2) ((T Z')^{-1})^{\otimes 2}$ is $[0,1,0] \in \mathscr{A}$ up to a scalar.
 Thus, $\mathcal{F}$ is $\mathscr{A}$-transformable under this $T Z'$ transformation.
\end{proof}

\begin{lemma} \label{lem:dic:a3}
 Let $\mathcal{F}$ be a set of symmetric signatures.
 Suppose $\mathcal{F}$ contains a non-degenerate signature $f \in \mathscr{A}_3$ of arity $n \ge 3$.
 Then $\Holant(\mathcal{F})$ is $\SHARPP$-hard unless $\mathcal{F}$ is $\mathscr{A}$-transformable.
\end{lemma}

\begin{proof}
 By assumption, for any $0 \le k \le n-2$, we have $f_{k+2} = i f_k$.
 We can express $f$ as
 \[
  f = \lambda \left(\begin{bmatrix} 1 \\  \alpha \end{bmatrix}^{\otimes n}
              + i^r \begin{bmatrix} 1 \\ -\alpha \end{bmatrix}^{\otimes n}\right),
 \]
 for some integer $r$.

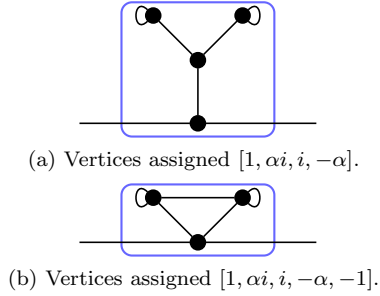
\begin{figure}[t]
 \centering
 \def\10iWidth{6cm}
 \captionsetup[subfigure]{width=\10iWidth}
 \subfloat[Vertices assigned ${[} 1, \alpha i, i, -\alpha {]}$.]{
  \makebox[\10iWidth][c]{
   \begin{tikzpicture}[scale=\scale,transform shape,node distance=\nodeDist,semithick]
    \node[external] (0)                    {};
    \node[external] (1) [right       of=0] {};
    \node[internal] (2) [right       of=1] {};
    \node[internal] (3) [above       of=2] {};
    \node[internal] (4) [above left  of=3] {};
    \node[internal] (5) [above right of=3] {};
    \node[external] (6) [right       of=2] {};
    \node[external] (7) [right       of=6] {};
    \path (0) edge (2)
          (2) edge (3)
              edge (7)
          (3) edge (4)
              edge (5);
    \path (4) edge[out=135, in=-135, looseness=5] coordinate (c1) (4)
          (5) edge[out= 45, in= -45, looseness=5] coordinate (c2) (5);
    \begin{pgfonlayer}{background}
     \node[inner sep=2pt,transform shape=false,draw=\borderColor,thick,rounded corners,fit = (1) (4) (5) (6) (c1) (c2)] {};
    \end{pgfonlayer}
  \end{tikzpicture}} \label{fig:gadget:construct_10i:arity3}}
 \qquad
 \subfloat[Vertices assigned ${[} 1, \alpha i, i, -\alpha, -1 {]}$.]{
  \makebox[\10iWidth][c]{
   \begin{tikzpicture}[scale=\scale,transform shape,node distance=\nodeDist,semithick]
    \node[external] (0)                    {};
    \node[external] (1) [right       of=0] {};
    \node[internal] (2) [right       of=1] {};
    \node[internal] (3) [above left  of=2] {};
    \node[internal] (4) [above right of=2] {};
    \node[external] (5) [right       of=2] {};
    \node[external] (6) [right       of=5] {};
    \path (0) edge (2)
          (2) edge (3)
              edge (4)
              edge (6)
          (3) edge (4);
    \path (3) edge[out=135, in=-135, looseness=5] coordinate (c1) (3)
          (4) edge[out= 45, in= -45, looseness=5] coordinate (c2) (4);
    \begin{pgfonlayer}{background}
     \node[inner sep=2pt,transform shape=false,draw=\borderColor,thick,rounded corners,fit = (1) (3) (4) (5) (c1) (c2)] {};
    \end{pgfonlayer}
  \end{tikzpicture}} \label{fig:gadget:construct_10i:arity4}}
 \caption{Constructions to realize $[1,0,i]$.}
 \label{fig:gadget:construct_10i}
\end{figure}

 A self loop on $f$ yields $f'$,
 where $f'_k = f_k + f_{k+2} = (1+i) f_k$.
 Thus up to the constant $(1+i)$,
 $f'$ is just the first $n-2$ entries of $f$.
 By doing more self loops,
 we eventually obtain a quaternary signature when $n$ is even or a ternary one when $n$ is odd.
 There are eight cases depending on the first two entries of $f$ and the parity of $n$.
 However, for any case, we can realize the signature $[1,0,i]$.
 We list them here.
 (In the calculations below, we omit certain nonzero constant factors without explanation.)
\begin{itemize}
 \item $[0, 1, 0, i]$: Another self loop gives $[0,1]$.
  Connect it back to the ternary to get $[1,0,i]$.
 \item $[1, 0, i, 0]$: Another self loop gives $[1,0]$.
  Connect it back to the ternary to get $[1,0,i]$.
 \item $[1, \alpha i, i, -\alpha]$: Another self loop gives $[1, \alpha i]$.
  Connect two copies of it to the ternary to get $[1, -\alpha]$.
  Then connect this back to the ternary to finally get $[1,0,i]$.
  See Figure~\ref{fig:gadget:construct_10i:arity3}.
 \item $[1, -\alpha i, i, \alpha]$: Same construction as the previous case.
 \item $[0, 1, 0, i, 0]$: Another self loop gives $[0,1,0]$.
  Connect it back to the quaternary to get $[1,0,i]$.
 \item $[1, 0, i, 0, -1]$: Another self loop gives $[1,0,i]$ directly.
 \item $[1, \alpha i, i, -\alpha, -1]$: Another self loop gives $[1, \alpha i, i]$.
  Connect two copies of it together to get $[1, -\alpha, -i]$.
  Connect this back to the quaternary to get $[1,0,i]$.
  See Figure~\ref{fig:gadget:construct_10i:arity4}.
 \item $[1, -\alpha i, i, \alpha, -1]$: Same construction as the previous case.
\end{itemize}

 With $[1,0,i]$ in hand, we can connect three copies to get $[1,0,-i]$ and four copies to get $[1,0,1]$.
 Now we construct a bipartite graph, with $\mathcal{F} \union \{{=}_2\}$ on the right side and $[1,0,-i]$ on the left,
 and do a holographic transformation by $Z = \left[\begin{smallmatrix} \alpha & 1 \\ -\alpha & 1 \end{smallmatrix}\right]$ to get
 \begin{align*}
  &\holant{[1,0,-i]}{\mathcal{F} \union \{f, {=}_2\}}\\
  \equiv_T & \holant{[1,0,-i] (Z^{-1})^{\otimes 2}}{Z \mathcal{F} \union \{Z^{\otimes n} f, Z^{\otimes 2} (=_2)\}}\\
  \equiv_T & \holant{\tfrac{1}{2i}[1,0,1]}{Z \mathcal{F} \union \{[1,0,\dotsc,0,i^k], [1,-i,1]\}}\\
  \equiv_T & \Holant\left(Z \mathcal{F} \union \{[1,0,\dotsc,0,i^k], [1,-i,1]\}\right).
 \end{align*}
 Notice that $f$ becomes $[1,0,\dotsc,0,i^k]$ where $k = r + 2 n$ (after normalizing the first entry) and $=_2$ becomes $[1,-i,1]$.
 On the other side, $[1,0,-i]$ becomes $[1,0,1]$.
 Therefore, we can construct all equality signatures of even arity using the powers of the transformed $f$
 (by using~4 copies of the transformed $f$, and connecting pairs of input wires by $=_2$).
 Thus, $\CSP^2(Z \mathcal{F} \union \{[1,-i,1]\}) \le_T \Holant(\mathcal{F})$.
 
 Now we apply Theorem~\ref{thm:CSPd}, the $\CSP^d$ dichotomy, to the set $Z \mathcal{F} \union \{[1,-i,1]\}$.
 If this problem is $\SHARPP$-hard, then $\Holant(\mathcal{F})$ is $\SHARPP$-hard as well.
 Otherwise, this problem is $\CSP^2$ tractable.
 Therefore, there exists some $T$ of the form $\left[\begin{smallmatrix} 1 & 0 \\ 0 & \alpha^d \end{smallmatrix}\right]$, where the integer $d \in \{0, 1, \dotsc, 7\}$,
 such that $T Z \mathcal{F} \union \{T^{\otimes 2} [1,-i,1]\}$ is a subset of $\mathscr{A}$ or $\mathscr{P}$.

 However, $T^{\otimes 2} [1,-i,1]$ can never be in $\mathscr{P}$.
 Thus $T Z \mathcal{F} \union \{T^{\otimes 2} [1,-i,1]\} \subseteq \mathscr{A}$.
 Further notice that if $d \in \{1,3,5,7\}$ in the expression of $T$, then $T^{\otimes 2} [1,-i,1]$ is not in $\mathscr{A}$.
 Hence, $T$ must be of the form $\left[\begin{smallmatrix} 1 & 0 \\ 0 & i^d \end{smallmatrix}\right]$, where the integer $d \in \{0, 1, 2, 3\}$.
 For such $T$, $T^{\otimes 2}[1,-i,1] \in \mathscr{A}$, and $T^{-1} \mathscr{A} = \mathscr{A}$ as $T\in\StabA$.
 Thus, $T Z \mathcal{F} \union \{T^{\otimes 2} [1,-i,1]\} \subseteq \mathscr{A}$ simply becomes $Z \mathcal{F} \subseteq \mathscr{A}$.
 Moreover, $(=_2) (Z^{-1})^{\otimes 2}$ is $[1,i,1] \in \mathscr{A}$.
 Therefore, $\mathcal{F}$ is $\mathscr{A}$-transformable under this $Z$ transformation.
\end{proof}

\section{The Main Dichotomy}

In this section, we prove our main dichotomy theorem.
We begin with a dichotomy for a single signature, which we prove by induction on its arity.

\begin{theorem} \label{thm:dic:single}
 If $f$ is a non-degenerate symmetric signature of arity at least~$3$ with complex weights in Boolean variables,
 then $\Holant(f)$ is $\SHARPP$-hard unless $f \in \mathscr{P}_1 \cup \mathscr{P}_2 \cup \mathscr{A}_3$ or $f$ is vanishing,
 in which case the problem is computable in polynomial time.
\end{theorem}

Recall that $\mathscr{A}_1 \subset \mathscr{P}_1$ and $\mathscr{A}_2 = \mathscr{P}_2$,
and $f \in \mathscr{P}_1 \cup \mathscr{P}_2 \cup \mathscr{A}_3$ iff $f$ is $\mathscr{A}$-transformable or $\mathscr{P}$-transformable by Corollary~\ref{cor:single:AP-trans_by_sets}.

\begin{proof}
 Let the arity of $f$ be $n$.
 The base cases of $n=3$ and $n=4$ are proved in Theorem~\ref{thm:arity3:singleton} and Theorem~\ref{thm:arity4:singleton} respectively.
 Now assume $n \ge 5$.

 With the signature $f$, we form a self loop to get a signature $f'$ of arity at least 3.
 We consider the cases separately whether $f'$  is degenerate or not.
 \begin{itemize}
  \item Suppose $f' = [a,b]^{\otimes(n-2)}$ is degenerate.
   There are three cases to consider.
   \begin{enumerate}
    \item If $a = b = 0,$ then $f'$ is the all zero signature.
     For $f$, this means $f_{k+2} = -f_k$ for $0 \le k \le n-2$, 
     so $f \in \mathscr{P}_2$ by Lemma~\ref{lem:single:P2}, and therefore $\Holant(f)$ is tractable.
    \item If $a^2 + b^2 \ne 0$, then $f'$ is nonzero and $[a,b]$ is not a constant multiple of either $[1,i]$ or $[1,-i]$.
     We may normalize so that $a^2 + b^2 = 1$.
     Then the orthogonal transformation $\left[\begin{smallmatrix} a & b \\ -b & a \end{smallmatrix}\right]$ transforms the column vector $[a,b]$ to $[1,0]$.
     Let $\hat{f}$ be the transformed signature from $f$, and $\widehat{f'} = [1,0]^{\otimes(n-2)}$ the transformed signature from $f'$.

     Since an orthogonal transformation keeps $=_2$ invariant, this transformation commutes with the operation of taking a self loop, i.e., $\widehat{f'} = (\hat{f})'$.
     Here $(\hat{f})'$ is the function obtained from $\hat{f}$ by taking a self loop.
     So $\hat{f}_0 + \hat{f}_2 = 1$ and for every integer $1 \le k \le n-2$, we have $\hat{f}_k = -\hat{f}_{k+2}$.
     With one or more self loops, we eventually obtain either $[1,0]$ when $n$ is odd or $[1,0,0]$ when $n$ is even.
     In either case, we connect an appropriate number of copies of this
     signature to $\hat{f}$ to get a arity 4 signature $\hat{g} = [\hat{f}_0, \hat{f}_1, \hat{f}_2, -\hat{f}_1, -\hat{f}_2]$.
     We show that $\Holant(\hat{g})$ is $\SHARPP$-hard.
     To see this, we first compute $\det(\widetilde{M_g}) = -2 (\hat{f}_0 + \hat{f}_2) (\hat{f}_1^2 + \hat{f}_2^2)=-2(\hat{f}_1^2 + \hat{f}_2^2)$, since $\hat{f}_0 + \hat{f}_2 = 1$.
     Therefore if $\hat{f}_1^2 + \hat{f}_2^2 \ne 0$, $\Holant(\hat{g})$ is $\SHARPP$-hard by Corollary~\ref{cor:arity4:nonsingular_compressed_hard}.
     Otherwise $\hat{f}_1^2 + \hat{f}_2^2 = 0$, and we consider $\hat{f}_2 = i \hat{f}_1$ since the other case is similar.
     Since $f$ is non-degenerate, $\hat{f}$ is non-degenerate, which implies $\hat{f}_2 \neq 0$.
     We can express $\hat{g}$ as $[1,0]^{\otimes 4} -\hat{f}_2 [1,i]^{\otimes 4}$.
     Under the holographic transformation by $T = \left[\begin{smallmatrix} 1 & (-\hat{f}_2)^{1/4} \\ 0 & i (-\hat{f}_2)^{1/4} \end{smallmatrix}\right]$,
     we have
     \begin{align*}
      \holant{{=}_2}{\hat{g}}
      &\equiv_T \holant{[1,0,1] T^{\otimes 2}}{(T^{-1})^{\otimes 4} \hat{g}}\\
      &\equiv_T \holant{\hat{h}}{{=}_4},
     \end{align*}
     where
     \[\hat{h} = [1,0,1] T^{\otimes 2} = [1,(-\hat{f}_2)^{1/4},0]\]
     and $\hat{g}$ is transformed by $T^{-1}$ into the arity 4 equality $=_4$, since
     \[
      T^{\otimes 4} \left( \begin{bmatrix} 1 \\ 0 \end{bmatrix}^{\otimes 4} + \begin{bmatrix} 0 \\ 1 \end{bmatrix}^{\otimes 4} \right)
      = \begin{bmatrix} 1 \\ 0 \end{bmatrix}^{\otimes 4} - \hat{f}_2 \begin{bmatrix} 1 \\ i \end{bmatrix}^{\otimes 4}
      = \hat{g}.
     \]
     By Theorem~\ref{thm:k-reg_homomorphism}$'$, $\holant{\hat{h}}{{=}_4}$ is $\SHARPP$-hard as $\hat{f}_2 \ne 0$.

    \item If $a^2 + b^2 = 0$ but $(a,b) \ne (0,0)$, then $[a,b]$ is a nonzero multiple of $[1, \pm i]$.
     Ignoring the constant multiple, we have $f' = [1,i]^{\otimes (n-2)}$ or $[1,-i]^{\otimes (n-2)}$.
     We consider the first case since the other case is similar.

     In the first case,
     the characteristic polynomial of the recurrence relation of $f'$ is $x-i$,
     so that of $f$ is $(x-i) (x^2 + 1) = (x-i)^2 (x+i)$.
     Hence there exist $a_0, a_1$, and $c$ such that
     \[
      f_k = (a_0 + a_1 k) i^k + c (-i)^k
     \]
     for every integer $0 \le k \le n$.
     Let $f^+$ and $f^-$ be two signatures of arity $n$ 
     such that $f^+_k=(a_0 + a_1 k) i^k$ and $f^-_k=c (-i)^k$ for every $0\leq k\leq n$.
     Hence $f_k=f^+_k+f^-_k$ and we write $f=f^++f^-$.
     If $a_1 = 0$, then $f'$ is the all zero signature, a contradiction.
     If $c=0$, then $f$ is vanishing, one of the tractable cases.
     Now we assume $a_1 c \neq 0$ and show that $\Holant(f)$ is $\SHARPP$-hard.
     Hence $\rd^+ (f^+) = 1$ and $\rd^- (f^-) = 0$.
     Under the holographic transformation $Z = \frac{1}{\sqrt{2}} \left[\begin{smallmatrix} 1 & 1 \\ i & -i \end{smallmatrix}\right]$,
     we have
     \begin{align*}
      \holant{{=}_2}{f}
      &\equiv_T \holant{[1,0,1] Z^{\otimes 2}}{(Z^{-1})^{\otimes n} f}\\
      &\equiv_T \holant{[0,1,0]}{\hat{f}},
     \end{align*}
     where $\hat{f}$ takes the form $[\hat{f_0}, \hat{f_1}, 0, \dotsc, 0, c']$ with $c' = 2^{n/2} c \ne 0$ and $\hat{f_1} \neq 0$,
     since $\hat{f}$ is the $Z^{-1}$-transformation of the sum of $f^+$ and $f^-$,
     with $\rd^+ (f^+) = 1$ and $\rd^- (f^-) = 0$ respectively.
     On the other side, $(=_2) = [1,0,1]$ is transformed into $(\neq_2) = [0,1,0]$.
     Now consider the gadget in Figure~\ref{sugfig:gadget:van-binary} with $\hat{f}$ assigned to both vertices.
     This gadget has the binary signature $[0, c \hat{f_0}, 2 c \hat{f_1}]$, which is equivalent to $[0, \hat{f_0}, 2 \hat{f_1}]$ since $c \neq 0$.
     Translating back by $Z$ to the original setting, this signature is $g = [\hat{f_0} + \hat{f_1}, -i \hat{f_1}, \hat{f_0} - \hat{f_1}]$.
     This can be verified as
     \[
      \begin{bmatrix} 1 & 1 \\ i & -i \end{bmatrix}
      \begin{bmatrix} 0 & \hat{f_0} \\ \hat{f_0} & 2 \hat{f_1} \end{bmatrix}
      \begin{bmatrix} 1 & 1 \\ i & -i \end{bmatrix}^\intercal
      =
      2
      \begin{bmatrix} \hat{f_0} + \hat{f_1} & -i \hat{f_1} \\ -i \hat{f_1} & \hat{f_0} - \hat{f_1} \end{bmatrix}.
     \]
     Since $\hat{f_1} \neq 0$, it can be directly checked that $g \not\in \mathscr{R}^+_2$.

     If $\hat{f_0} \neq 0$, then $g$ is non-degenerate.
     In this case we construct some function in $\mathscr{V}^+$.
     We connect $f'$ back to $f$, getting a binary signature $p=Z^{\otimes 2}[0,0,c']$.
     Then we connect $p$ to $f$, 
     the resulting signature is $p'=Z^{\otimes n-2}[\hat{f_0}, \hat{f_1},0,\dots,0]$ of arity $n-2\ge 3$ 
     up to the constant factor of $c'\neq 0$.
     Notice that $p'$ is non-degenerate and $p'\in\mathscr{V}^+$.
     By Lemma~\ref{lem:van:bin}, 
     $\Holant(\{p', g\})$ is $\SHARPP$-hard, 
     hence $\Holant(f)$ is also $\SHARPP$-hard.

     Otherwise suppose $\hat{f_0} = 0$.
     Then we have $g = [1,-i]^{\otimes 2}$ after ignoring the nonzero factor $\hat{f_1}$.
     Connecting this degenerate signature to $f$, we get a signature $h = \langle f, g \rangle$.
     We note that $g$ annihilates the signature $f^-=c [1,-i]^{\otimes n}$, and thus $h = \langle f^+, g \rangle$.
     Then $\rd^+(f^+) =1$, $\vd^{+} (g) = 0$,
     and we can apply Lemma~\ref{lem:van:con}.
     It follows that $\rd^+(h) = 1$ and $\arity(h) \ge 3$.
     This implies that $h$ is non-degenerate and $h \in \mathscr{V}^{+}$.

\begin{figure}[t]
 \centering
 \def\capWidth{5.5cm}
 \captionsetup[subfigure]{width=\capWidth}
 \subfloat[The circles are assigned $\hat{f}$ and the squares are assigned $\neq_2$.]{
  \makebox[\capWidth][c]{
   \begin{tikzpicture}[scale=\scale,transform shape,node distance=\nodeDist,semithick]
    \node[external] (0)              {};
    \node[internal] (1) [right of=0] {};
    \node[external] (2) [right of=1] {};
    \node[external] (3) [right of=2] {};
    \node[internal] (4) [right of=3] {};
    \node[external] (5) [right of=4] {};
    \path (0) edge                          node[near end]   (e1) {}               (1)
          (1) edge[out= 45, in= 135]        node[square]     (e2) {}               (4)
              edge[out= 15, in= 165]        node[square]          {}               (4)
              edge[out=-10, in=-170, white] node[black]           {\Huge $\vdots$} (4)
              edge[out=-45, in=-135]        node[square]     (e3) {}               (4)
          (4) edge                          node[near start] (e4) {}               (5);
    \begin{pgfonlayer}{background}
     \node[inner sep=0pt,transform shape=false,draw=\borderColor,thick,rounded corners,fit = (e1) (e2) (e3) (e4)] {};
    \end{pgfonlayer}
  \end{tikzpicture}} \label{sugfig:gadget:van-binary}}
 \qquad
 \subfloat[The circles are assigned $f$.]{
  \makebox[\capWidth][c]{
   \begin{tikzpicture}[scale=\scale,transform shape,node distance=\nodeDist,semithick]
    \node[internal]  (0)                    {};
    \node[external]  (1) [above left  of=0] {};
    \node[external]  (2) [below left  of=0] {};
    \node[external]  (3) [left        of=1] {};
    \node[external]  (4) [left        of=2] {};
    \node[external]  (5) [right       of=0] {};
    \node[external]  (6) [right       of=5] {};
    \node[internal]  (7) [right       of=6] {};
    \node[external]  (8) [above right of=7] {};
    \node[external]  (9) [below right of=7] {};
    \node[external] (10) [right       of=8] {};
    \node[external] (11) [right       of=9] {};
    \path (0) edge[out= 135, in=   0]                                     (3)
              edge[out=-135, in=   0]                                     (4)
              edge[out=  45, in= 135]                                     (7)
              edge[out=  15, in= 165]                                     (7)
              edge[out= -10, in=-170, white] node[black] {\Huge $\vdots$} (7)
              edge[out= -45, in=-135]                                     (7)
          (7) edge[out=  45, in= 180]                                    (10)
              edge[out= -45, in= 180]                                    (11);
    \begin{pgfonlayer}{background}
     \node[inner sep=0pt,transform shape=false,draw=\borderColor,thick,rounded corners,fit = (1) (2) (8) (9)] {};
    \end{pgfonlayer}
   \end{tikzpicture}} \label{subfig:gadget:van-arity4}}
 \caption{Two gadgets used when $f' = [1, \pm i]^{\otimes (n-2)}$.}
\end{figure}
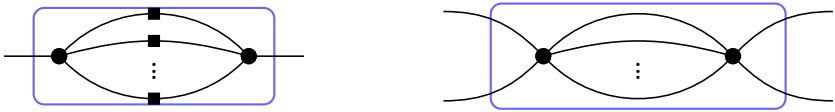

     Moreover, assigning $f$ to both vertices in the gadget of Figure~\ref{subfig:gadget:van-arity4},
     we get a non-degenerate signature $h' \in \mathscr{V}^{-}$ of arity~4.
     To see this, consider this gadget after a holographic transformation by $Z$.
     In this bipartite setting,
     it is the same as assigning $\hat{f} = [0, \hat{f_1}, 0, \dotsc, 0, c]$
     (or equivalently $[0, 1, 0, \dotsc, 0, c'']$, where $c'' = c/\hat{f_1} \neq 0$)
     to both the circle and triangle vertices in the gadget of Figure~\ref{fig:gadget:special_vanishing_case:many_glue}.
     The square vertices there are still assigned $(\neq_2) = [0,1,0]$.
     While it is not apparent from the gadget's geometry, this signature is in fact symmetric.
     In particular, its values on inputs $1010$ and $1100$ are both~$0$.
     The resulting signature is $\widehat{h'} = (Z^{-1})^{\otimes 4} h' = [0,0,0,c'',0]$.
     Hence $\rd^{-}(h') = 1$, and therefore $h'$ is non-degenerate and  $h' \in \mathscr{V}^{-}$.

     By Lemma~\ref{lem:van:plus_and_minus},
     $\Holant(\{h, h'\})$ is $\SHARPP$-hard,
     hence $\Holant(f)$ is also $\SHARPP$-hard.
   \end{enumerate}
  \item Suppose $f'$ is non-degenerate.
   If $f'$ is not in one of the tractable cases, then $\Holant(f')$ is $\SHARPP$-hard and so is $\Holant(f)$.
   We now assume $\Holant(f')$ is not $\SHARPP$-hard.
   Then, by inductive hypothesis, $f' \in \mathscr{P}_1 \cup \mathscr{P}_2 \cup \mathscr{A}_3$ or $f'$ is vanishing.
   If $f' \in \mathscr{P}_1 \cup \mathscr{P}_2 \cup \mathscr{A}_3$,
   then applying Lemma~\ref{lem:dic:p1}, Lemma~\ref{lem:dic:p2}, or Lemma~\ref{lem:dic:a3} to $f'$ and the set $\{f, f'\}$,
   we either have that $\Holant(\{f, f'\})$ is $\SHARPP$-hard, so $\Holant(f)$ is $\SHARPP$-hard as well,
   or that $f$ is $\mathscr{A}$- or $\mathscr{P}$-transformable,
   so by Corollary~\ref{cor:single:AP-trans_by_sets}, $f \in \mathscr{P}_1 \cup \mathscr{P}_2 \cup \mathscr{A}_3$.

   Otherwise, $f'$ is vanishing, so $f' \in \mathscr{V}^\sigma$ for $\sigma \in \{+,-\}$ by Theorem~\ref{thm:van}.
   For simplicity, assume that $f' \in \mathscr{V}^+$.
   The other case is similar.
   Let $\rd^+(f') = d - 1$, where $2 d < n$ and $d \ge 2$ since $f'$ is non-degenerate.
   Then the entries of $f'$ can be expressed as
   \[f_k' = i^k q(k),\]
   where $q(x)$ is a polynomial of degree exactly $d-1$.
   However, notice that if $f'$ satisfies some recurrence relation with characteristic polynomial $t(x)$,
   then $f$ satisfies a recurrence relation with characteristic polynomial $(x^2 + 1) t(x)$.
   In this case, $t(x) = (x-i)^{d}$.
   Then the corresponding characteristic polynomial of $f$ is $(x-i)^{d+1} (x+i)$, and thus the entries of $f$ are
   \[f_k = i^k p(k) + c (-i)^k\]
   for some constant $c$ and a polynomial $p(x)$ of degree at most $d$.
   However, the degree of $p(x)$ is exactly $d$, otherwise the polynomial $q(x)$ for $f'$ would have degree less than $d-1$.
   If $c=0$, then $f\in\mathscr{V}^{+}$ is vanishing, a tractable case.
   Now assume $c \neq 0$, and we want to show the problem is $\SHARPP$-hard.

   Thus, under the transformation $Z = \frac{1}{\sqrt{2}} \left[\begin{smallmatrix} 1 & 1 \\ i & -i \end{smallmatrix}\right]$, we have
   \begin{align*}
    \holant{{=}_2}{f}
    &\equiv_T \holant{[1,0,1] Z^{\otimes 2}}{(Z^{-1})^{\otimes n} f}\\
    &\equiv_T \holant{[0,1,0]}{\hat{f}},
   \end{align*}
   where $\hat{f} = [\hat{f}_0, \hat{f}_1, \dotsc, \hat{f}_d, 0, \dotsc, 0, c]$, with $\hat{f}_d \neq 0$.
   Taking a self loop in the original setting is equivalent to connecting $[0,1,0]$ to a signature after this transformation.
   Thus, doing this once on $\hat{f}$, we can get $\widehat{f'} = [\hat{f}_1, \dotsc, \hat{f}_{d}, 0, \dotsc, 0]$ corresponding to $f'$,
   and doing this $d-2$ times on $\hat{f}$,
   we get a signature $\hat{h} = [\hat{f}_{d-2}, \hat{f}_{d-1}, \hat{f}_{d}, 0, \dotsc, 0, 0 / c]$ of arity $n - 2 (d-2) = n - 2 d + 4$.
   The last entry is $c$ when $d=2$ and is 0 when $d>2$.

   As $n > 2 d$, we may do two more self loops and get $[\hat{f}_{d}, 0, \dotsc, 0]$ of arity $k = n - 2 d$.
   Now connect this signature back to $\hat{f}$ via $[0,1,0]$.
   It is the same as getting the last $n - k + 1 = 2 d + 1$ signature entries of $\hat{f}$.
   We may repeat this operation zero or more times until the arity $k'$ of the resulting signature is less than or equal to $k$.
   We claim that this signature has the form $\hat{g} = [0, \dotsc, 0, c]$.
   In other words, the $k'+1$ entries of $\hat{g}$ consist of the last $c$ and $k'$ many 0's in the signature $\hat{f}$, all appearing after $\hat{f}_{d}$.
   This is because there are $n - d - 1$ many~$0$ entries in the signature $\hat{f}$ after $\hat{f}_{d}$,
   and $n - d - 1 \ge k \ge k'$.

   Translating back by the $Z$ transformation,
   having both $[\hat{f}_{d}, 0, \dotsc, 0]$ of arity $k$ and $\hat{g} = [0, \dotsc, 0, c]$ of arity $k'$ is equivalent to,
   in the original setting,
   having both $[1,i]^{\otimes k}$ and $[1,-i]^{\otimes k'}$.
   If $k > k'$,
   then we can connect $[1,-i]^{\otimes k'}$ to $[1,i]^{\otimes  k}$ and get $[1,i]^{\otimes (k-k')}$.
   Replacing $k$ by $k-k'$,
   we can repeat this process until the new $k \le k'$.
   If the new $k < k'$,
   then we can continue as in the subtractive Euclid algorithm.
   We continue this procedure and eventually we get $[1,i]^{\otimes t}$ and  $[1,-i]^{\otimes  t}$,
   where $t = \gcd(k,k')$, where $k = n - 2 d$ and $k' \le k$,
   as defined in the previous paragraph.
   Now putting $k / t$ many copies of $[1,-i]^{\otimes  t}$ together,
   we get $[1,-i]^{\otimes k}$.

   In the transformed setting,
   $[1,-i]^{\otimes k}$ is $[0, \dotsc, 0, 1]$ of arity $k$.
   Then we connect this back to $\hat{h}$ via $[0,1,0]$.
   Doing this is the same as forcing $k$ connected edges of $\hat{h}$ to be assigned~$0$,
   because $[0,1,0]$ flips the assigned value~1 in $[0, \dotsc, 0, 1]$ to~$0$.
   Thus we get a signature of arity $n - 2 d + 4 - k = 4$,
   which is $[\hat{f}_{d-2}, \hat{f}_{d-1}, \hat{f}_{d}, 0, 0]$.
   Note that the last entry is~$0$ (and not $c$),
   because $k \ge 1$.

   However,
   $\Holant([0,1,0] \: | \: [\hat{f}_{d-2}, \hat{f}_{d-1}, \hat{f}_{d}, 0, 0])$
   is equivalent to $\Holant([0,1,$ $0]\: | \: [0,0,1,0,0])$ when $\hat{f}_{d} \neq 0$,
   which is transformed back by $Z$ to $\Holant([3,0,$ $1,0,3])$.
   This is the Eulerian Orientation problem on $4$-regular graphs
   and is $\SHARPP$-hard by Theorem~\ref{thm:4reg_EO_hard}.
 \end{itemize}
\end{proof}

Now we are ready to prove of our main theorem.

\begin{proof}[Proof of hardness for Theorem~\ref{thm:main}]
 Assume that $\Holant(\mathcal{F})$ is not $\SHARPP$-hard.
 If all of the non-degenerate signatures in $\mathcal{F}$ are of arity at most 2,
 then the problem is tractable case~\ref{case:main_tractable:trivial}.
 Otherwise we have some non-degenerate signatures of arity at least $3$.
 For any such $f$, by Theorem~\ref{thm:dic:single}, $f \in \mathscr{P}_1 \cup \mathscr{P}_2 \cup \mathscr{A}_3$ or $f$ is vanishing.
 If any of them is in $\mathscr{P}_1 \cup \mathscr{P}_2 \cup \mathscr{A}_3$,
 then by Lemma~\ref{lem:dic:p1}, Lemma~\ref{lem:dic:p2}, or Lemma~\ref{lem:dic:a3},
 we have that $\mathcal{F}$ is $\mathscr{A}$- or $\mathscr{P}$-transformable,
 which are tractable cases~\ref{case:main_tractable:CSP:A} and~\ref{case:main_tractable:CSP:P}.

 Now we assume that all non-degenerate signatures of arity at least~$3$ in $\mathcal{F}$ are vanishing,
 and there is a nonempty set of such signatures in $\mathcal{F}$.
 By Lemma~\ref{lem:van:plus_and_minus},
 they must all be in $\mathscr{V}^\sigma$ with the same $\sigma\in\{+,-\}$.
 By Lemma~\ref{lem:van:bin},
 we know that any non-degenerate binary signature in $\mathcal{F}$ has to be in $\mathscr{R}^\sigma_2$.
 Furthermore,
 if $\mathcal{F}$ contains an $f \in \mathscr{V}^\sigma$ such that $\rd^\sigma(f) \ge 2$,
 then by Lemma~\ref{lem:van:deg},
 the only unary signatures allowed in $\mathcal{F}$ are some multiple of $[1,\sigma i]$,
 and all degenerate signatures in $\mathcal{F}$ are a tensor product of some multiple of $[1, \sigma i]$.
 Thus, all non-degenerate signatures of arity at least $3$ as well as all degenerate signatures belong to $\mathscr{V}^\sigma$,
 and all non-degenerate binary signatures belong to $\mathscr{R}^\sigma_2$.
 This is tractable case~\ref{case:main_tractable:vanishing_and_binary}.

 Finally, we have the following:
 (i) all non-degenerate signatures of arity at least $3$ in $\mathcal{F}$ belong to $\mathscr{V}^\sigma$;
 (ii) all signatures $f \in \mathcal{F} \cap \mathscr{V}^\sigma$ have $\rd^\sigma(f) \le 1$, which implies that $f \in \mathscr{R}^\sigma_2$; and
 (iii) all non-degenerate binary signatures in $\mathcal{F}$ belong to $\mathscr{R}^\sigma_2$.
 Hence all non-degenerate signatures in $\mathcal{F}$ belong to $\mathscr{R}^\sigma_2$.
 All unary signatures also belong to $\mathscr{R}^\sigma_2$ by definition.
 This is indeed tractable case~\ref{case:main_tractable:vanishing_and_unary}.
 The proof is complete.
\end{proof}

Furthermore, given a finite signature set $\mathcal{F}$,
the criterion of Theorem~\ref{thm:main} is decidable in polynomial time.
This is reported in \cite{CGW14a}.

\bibliographystyle{siamplain}
\bibliography{bib}

\begin{thebibliography}{10}

\bibitem{BDGJR09}
Andrei Bulatov, Martin Dyer, Leslie~Ann Goldberg, Markus Jalsenius, and David
  Richerby.
\newblock The complexity of weighted {B}oolean \#{CSP} with mixed signs.
\newblock {\em Theor. Comput. Sci.}, 410(38-40):3949--3961, 2009.

\bibitem{BG05}
Andrei Bulatov and Martin Grohe.
\newblock The complexity of partition functions.
\newblock {\em Theor. Comput. Sci.}, 348(2):148--186, 2005.

\bibitem{Bul13}
Andrei~A. Bulatov.
\newblock The complexity of the counting constraint satisfaction problem.
\newblock {\em J. ACM}, 60(5):34:1--34:41, 2013.

\bibitem{BD07}
Andrei~A. Bulatov and V{\'i}ctor Dalmau.
\newblock Towards a dichotomy theorem for the counting constraint satisfaction
  problem.
\newblock {\em Information and Computation}, 205(5):651--678, 2007.

\bibitem{CC10}
Jin-Yi Cai and Xi~Chen.
\newblock A decidable dichotomy theorem on directed graph homomorphisms with
  non-negative weights.
\newblock In {\em FOCS}, pages 437--446. IEEE Computer Society, 2010.

\bibitem{CC12}
Jin-Yi Cai and Xi~Chen.
\newblock Complexity of counting {CSP} with complex weights.
\newblock In {\em STOC}, pages 909--920. ACM, 2012.

\bibitem{CCL11}
Jin-Yi Cai, Xi~Chen, and Pinyan Lu.
\newblock Non-negatively weighted {\#}{CSP}: {A}n effective complexity
  dichotomy.
\newblock In {\em IEEE Conference on Computational Complexity}, pages 45--54.
  IEEE Computer Society, 2011.

\bibitem{CCL13}
Jin-Yi Cai, Xi~Chen, and Pinyan Lu.
\newblock Graph homomorphisms with complex values: {A} dichotomy theorem.
\newblock {\em SIAM J. Comput.}, 42(3):924--1029, 2013.

\bibitem{CGW13}
Jin-Yi Cai, Heng Guo, and Tyson Williams.
\newblock A complete dichotomy rises from the capture of vanishing signatures
  (extended abstract).
\newblock In {\em STOC}, pages 635--644. ACM, 2013.

\bibitem{CGW14a}
Jin-Yi Cai, Heng Guo, and Tyson Williams.
\newblock Holographic algorithms beyond matchgates.
\newblock In {\em ICALP}, pages 271--282. Springer Berlin Heidelberg, 2014.
\newblock \textit{CoRR}, \href{http://arxiv.org/abs/1307.7430}{abs/1307.7430}.

\bibitem{CHL12}
Jin-Yi Cai, Sangxia Huang, and Pinyan Lu.
\newblock From {H}olant to {\#}{CSP} and back: {D}ichotomy for {H}olant$^c$
  problems.
\newblock {\em Algorithmica}, 64(3):511--533, 2012.

\bibitem{CK12}
Jin-Yi Cai and Michael Kowalczyk.
\newblock Spin systems on $k$-regular graphs with complex edge functions.
\newblock {\em Theor. Comput. Sci.}, 461:2--16, 2012.

\bibitem{CK13}
Jin-Yi Cai and Michael Kowalczyk.
\newblock Partition functions on $k$-regular graphs with $\{0,1\}$-vertex
  assignments and real edge functions.
\newblock {\em Theor. Comput. Sci.}, 494(0):63--74, 2013.

\bibitem{CL11a}
Jin-Yi Cai and Pinyan Lu.
\newblock Holographic algorithms: {F}rom art to science.
\newblock {\em J. Comput. Syst. Sci.}, 77(1):41--61, 2011.

\bibitem{CLX09a}
Jin-Yi Cai, Pinyan Lu, and Mingji Xia.
\newblock {H}olant problems and counting {CSP}.
\newblock In {\em STOC}, pages 715--724. ACM, 2009.

\bibitem{CLX10}
Jin-Yi Cai, Pinyan Lu, and Mingji Xia.
\newblock Holographic algorithms with matchgates capture precisely tractable
  planar {\#}{CSP}.
\newblock In {\em FOCS}, pages 427--436. IEEE Computer Society, 2010.

\bibitem{CLX11d}
Jin-Yi Cai, Pinyan Lu, and Mingji Xia.
\newblock Computational complexity of {H}olant problems.
\newblock {\em SIAM J. Comput.}, 40(4):1101--1132, 2011.

\bibitem{CLX11b}
Jin-Yi Cai, Pinyan Lu, and Mingji Xia.
\newblock A computational proof of complexity of some restricted counting
  problems.
\newblock {\em Theor. Comput. Sci.}, 412(23):2468--2485, 2011.

\bibitem{CLX11a}
Jin-Yi Cai, Pinyan Lu, and Mingji Xia.
\newblock Dichotomy for {H}olant* problems of {B}oolean domain.
\newblock In {\em SODA}, pages 1714--1728. SIAM, 2011.

\bibitem{CLX12}
Jin-Yi Cai, Pinyan Lu, and Mingji Xia.
\newblock Holographic reduction, interpolation and hardness.
\newblock {\em Computational Complexity}, 21(4):573--604, 2012.

\bibitem{CLX13a}
Jin-Yi Cai, Pinyan Lu, and Mingji Xia.
\newblock Holographic algorithms by {F}ibonacci gates.
\newblock {\em Linear Algebra and its Applications}, 438(2):690--707, 2013.

\bibitem{CH96}
Nadia Creignou and Miki Hermann.
\newblock Complexity of generalized satisfiability counting problems.
\newblock {\em Inf. Comput.}, 125(1):1--12, 1996.

\bibitem{CKS01}
Nadia Creignou, Sanjeev Khanna, and Madhu Sudan.
\newblock {\em Complexity Classifications of {B}oolean Constraint Satisfaction
  Problems}.
\newblock Society for Industrial and Applied Mathematics, 2001.

\bibitem{DP91}
C.~T.~J. Dodson and T.~Poston.
\newblock {\em Tensor Geometry}, volume 130 of {\em Graduate Texts in
  Mathematics}.
\newblock Springer-Verlag, second edition, 1991.

\bibitem{DGJ09}
Martin Dyer, Leslie~Ann Goldberg, and Mark Jerrum.
\newblock The complexity of weighted {B}oolean {\#}{CSP}.
\newblock {\em SIAM J. Comput.}, 38(5):1970--1986, 2009.

\bibitem{DGP07}
Martin Dyer, Leslie~Ann Goldberg, and Mike Paterson.
\newblock On counting homomorphisms to directed acyclic graphs.
\newblock {\em J. ACM}, 54(6), 2007.

\bibitem{DG00}
Martin Dyer and Catherine Greenhill.
\newblock The complexity of counting graph homomorphisms.
\newblock {\em Random Struct. Algorithms}, 17(3-4):260--289, 2000.

\bibitem{DR10}
Martin Dyer and David Richerby.
\newblock On the complexity of {\#}{CSP}.
\newblock In {\em STOC}, pages 725--734. ACM, 2010.

\bibitem{FV98}
Tom{\'a}s Feder and Moshe~Y. Vardi.
\newblock The computational structure of monotone monadic {SNP} and constraint
  satisfaction: {A} study through {D}atalog and group theory.
\newblock {\em SIAM J. Comput.}, 28(1):57--104, 1998.

\bibitem{GGJT10}
Leslie~Ann Goldberg, Martin Grohe, Mark Jerrum, and Marc Thurley.
\newblock A complexity dichotomy for partition functions with mixed signs.
\newblock {\em SIAM J. Comput.}, 39(7):3336--3402, 2010.

\bibitem{GHLX11}
Heng Guo, Sangxia Huang, Pinyan Lu, and Mingji Xia.
\newblock The complexity of weighted {B}oolean {\#}{CSP} modulo~$k$.
\newblock In {\em STACS}, pages 249--260. Schloss Dagstuhl - Leibniz-Zentrum
  fuer Informatik, 2011.

\bibitem{GLV13}
Heng Guo, Pinyan Lu, and Leslie~G. Valiant.
\newblock The complexity of symmetric {B}oolean parity {H}olant problems.
\newblock {\em SIAM J. Comput.}, 42(1):324--356, 2013.

\bibitem{HN90}
Pavol Hell and Jaroslav Ne{\v{s}}et{\v{r}}il.
\newblock On the complexity of {H}-coloring.
\newblock {\em J. Comb. Theory Ser. B}, 48(1):92--110, 1990.

\bibitem{HL12}
Sangxia Huang and Pinyan Lu.
\newblock A dichotomy for real weighted {H}olant problems.
\newblock In {\em IEEE Conference on Computational Complexity}, pages 96--106.
  IEEE Computer Society, 2012.

\bibitem{HL13}
Sangxia Huang and Pinyan Lu.
\newblock A dichotomy for real weighted {H}olant problems.
\newblock {\em Comput. Complex.}, 2015.
\newblock doi:10.1007/s00037-015-0118-3.

\bibitem{Jos95}
A.~W. Joshi.
\newblock {\em Matrices And Tensors In Physics}.
\newblock New Age International, revised third edition, 1995.

\bibitem{For01}
G.~David~Forney Jr.
\newblock Codes on graphs: normal realizations.
\newblock {\em Information Theory, IEEE Transactions on}, 47(2):520--548, 2001.

\bibitem{Kas67}
P.~W. Kasteleyn.
\newblock Graph theory and crystal physics.
\newblock In F.~Harary, editor, {\em Graph Theory and Theoretical Physics},
  pages 43--110. Academic Press, London, 1967.

\bibitem{Kow09}
Michael Kowalczyk.
\newblock Classification of a class of counting problems using holographic
  reductions.
\newblock In {\em COCOON}, pages 472--485. Springer, 2009.

\bibitem{KC10}
Michael Kowalczyk and Jin-Yi Cai.
\newblock {H}olant problems for regular graphs with complex edge functions.
\newblock In {\em STACS}, pages 525--536. Schloss Dagstuhl - Leibniz-Zentrum
  fuer Informatik, 2010.

\bibitem{Lad75}
Richard~E. Ladner.
\newblock On the structure of polynomial time reducibility.
\newblock {\em J. ACM}, 22(1):155--171, 1975.

\bibitem{Loe04}
Hans-Andrea Loeliger.
\newblock An introduction to factor graphs.
\newblock {\em Signal Processing Magazine, IEEE}, 21(1):28--41, 2004.

\bibitem{Lov67}
L{\'a}szl{\'o} Lov{\'a}sz.
\newblock Operations with structures.
\newblock {\em Acta Math. Hung.}, 18(3-4):321--328, 1967.

\bibitem{MS08}
Igor~L. Markov and Yaoyun Shi.
\newblock Simulating quantum computation by contracting tensor networks.
\newblock {\em SIAM J. Comput.}, 38(3):963--981, 2008.

\bibitem{Sch78}
Thomas~J. Schaefer.
\newblock The complexity of satisfiability problems.
\newblock In {\em STOC}, pages 216--226. ACM, 1978.

\bibitem{Vad01}
Salil~P. Vadhan.
\newblock The complexity of counting in sparse, regular, and planar graphs.
\newblock {\em SIAM J. Comput.}, 31(2):398--427, 2001.

\bibitem{Val79b}
Leslie~G. Valiant.
\newblock The complexity of computing the permanent.
\newblock {\em Theoretical Computer Science}, 8(2):189--201, 1979.

\bibitem{Val06}
Leslie~G. Valiant.
\newblock Accidental algorthims.
\newblock In {\em FOCS}, pages 509--517. IEEE Computer Society, 2006.

\bibitem{Val08}
Leslie~G. Valiant.
\newblock Holographic algorithms.
\newblock {\em SIAM J. Comput.}, 37(5):1565--1594, 2008.

\bibitem{Xia11}
Mingji Xia.
\newblock Holographic reduction: {A} domain changed application and its partial
  converse theorems.
\newblock {\em Int. J. Software and Informatics}, 5(4):567--577, 2011.

\end{thebibliography}

\appendix

\section{Simple Interpolations}
In addition to the two arity~4 interpolations in Section~\ref{sec:arity4},
we also use interpolation in the proofs of two other lemmas.
Compared to our arity~4 interpolations, these binary interpolations are much simpler.

\begin{lemma} \label{lem:simple_interpolation:van:bin}
 Let $x \in \mathbb{C}$.
 If $x \ne 0$, then for any set $\mathcal{F}$ containing $[x,1,0]$, we have
 \[\holant{{\neq}_2}{\mathcal{F} \union \{[v,1,0]\}} \le_T \holant{{\neq}_2}{\mathcal{F}}\]
 for any $v \in \mathbb{C}$.
\end{lemma}

\begin{proof}
 Consider an instance $\Omega$ of $\holant{{\neq}_2}{\mathcal{F} \union \{[v,1,0]\}}$.
 Suppose that $[v,1,0]$ appears $n$ times in $\Omega$.
 We stratify the assignments in $\Omega$ based on the assignments to $[v,1,0]$.
 We only need to consider assignments of Hamming weight~$0$ and~$1$ since an assignment of Hamming weight~$2$ contributes a factor of~$0$.
 Let $i$ be the number of Hamming weight~$0$ assignments to $[v,1,0]$ in $\Omega$.
 Then there are $n-i$ assignments of Hamming weight~$1$ and the Holant on $\Omega$ is
 \[\Holant_\Omega = \sum_{i=0}^n v^i c_i,\]
 where $c_i$ is the sum over all such assignments of the product of evaluations of all other signatures on $\Omega$.

 We construct from $\Omega$ a sequence of instances $\Omega_s$ of $\Holant(\mathcal{F})$ indexed by $s \ge 1$.
 We obtain $\Omega_s$ from $\Omega$ by replacing each occurrence of $[v,1,0]$ with a gadget $g_s$ created from $s$ copies of $[x,1,0]$,
 connected sequentially but with $(\neq_2) = [0,1,0]$ between each sequential pair.
 The signature of $g_s$ is $[s x, 1, 0]$, which can be verified by the matrix product
 \[
  \left(\begin{bmatrix} x & 1 \\ 1 & 0 \end{bmatrix} \begin{bmatrix} 0 & 1 \\ 1 & 0 \end{bmatrix}\right)^{s-1} \begin{bmatrix} x & 1 \\ 1 & 0 \end{bmatrix}
  = \begin{bmatrix} 1 & x \\ 0 & 1 \end{bmatrix}^{s-1} \begin{bmatrix} x & 1 \\ 1 & 0 \end{bmatrix}
  = \begin{bmatrix} 1 & (s-1) x \\ 0 & 1 \end{bmatrix} \begin{bmatrix} x & 1 \\ 1 & 0 \end{bmatrix}
  = \begin{bmatrix} s x & 1 \\ 1 & 0 \end{bmatrix}.
 \]
 The Holant on $\Omega_s$ is
 \[\Holant_{\Omega_s} = \sum_{i=0}^n (s x)^i c_i.\]
 For $s \ge 1$, this gives a coefficient matrix that is Vandermonde.
 Since $x$ is nonzero, $s x$ is distinct for each $s$.
 Therefore, the Vandermonde system has full rank.
 We can solve for the unknowns $c_i$ and obtain the value of $\Holant_\Omega$.
\end{proof}

\begin{lemma} \label{lem:simple_interpolation:dic:p1}
 Let $a, b \in \mathbb{C}$.
 If $a b \ne 0$, then for any set $\mathcal{F}$ of complex-weighted signatures containing $[a, 0, \dotsc, 0, b]$ of arity $r \ge 3$,
 \[\Holant(\mathcal{F} \union \{{=}_4\}) \le_T \Holant(\mathcal{F}).\]
\end{lemma}

\begin{proof}
 Since $a \ne 0$, we can normalize the first entry to get $[1, 0, \dotsc, 0, x]$, where $x \ne 0$.
 First, we show how to obtain an arity~4 generalized equality signature.
 If $r = 3$, then we connect two copies together by a single edge to get an arity~4 signature.
 For larger arities, we form self-loops until realizing a signature of arity~3 or~4.
 By this process, we have a signature $g = [1,0,0,0,y]$, where $y \neq 0$.
 If $y$ is a $p$th root of unity, then we can directly realize $=_4$ by connecting $p$ copies of $g$ together,
 two edges at a time as in Figure~\ref{fig:gadget:arity4:interpolate_I3}.
 Otherwise, $y$ is not a root of unity and we can interpolate $=_4$ as follows.

 Consider an instance $\Omega$ of $\Holant(\mathcal{F} \union \{{=}_4\})$.
 Suppose that $=_4$ appears $n$ times in $\Omega$.
 We stratify the assignments in $\Omega$ based on the assignments to $=_4$.
 We only need to consider the all-zero and all-one assignments since any other assignment contributes a factor of 0.
 Let $i$ be the number of all-one assignments to $=_4$ in $\Omega$.
 Then there are $n-i$ all-zero assignments and the Holant on $\Omega$ is
 \[\Holant_\Omega = \sum_{i=0}^n c_i,\]
 where $c_i$ is the sum over all such assignments of the product
 of evaluations of all other signatures on $\Omega$.

 We construct from $\Omega$ a sequence of instances $\Omega_s$ of $\Holant(\mathcal{F})$ indexed by $s \ge 1$.
 We obtain $\Omega_s$ from $\Omega$ by replacing each occurrence of $=_4$ with a gadget $g_s$ created from $s$ copies of $[1,0,0,0,y]$,
 connecting two edges together at a time as in Figure~\ref{fig:gadget:arity4:interpolate_I3}.
 The Holant on $\Omega_s$ is
 \[\Holant_{\Omega_s} = \sum_{i=0}^n (y^s)^i c_i.\]
 For $s \ge 1$, this gives a coefficient matrix that is Vandermonde.
 Since $y$ is neither~$0$ nor a root of unity,
 $y^s$ is distinct for each $s$.
 Therefore, the Vandermonde system has full rank.
 We can solve for the unknowns $c_i$
 and obtain the value of $\Holant_\Omega$.
\end{proof}

Since the gadget constructions are planar, this lemma holds when restricted to planar graphs.

\section{An Orthogonal Transformation} \label{sec:orthogonal}
Here we give the details of the \emph{orthogonal} transformation used in the proof of Lemma~\ref{lem:arity4:double_root}.
We state the general case for symmetric signatures of arity $n \ge 1$.
Appendix~D of~\cite{CHL12} has the case $n = 3$.

We are given a symmetric signature $f = [f_0, \dotsc, f_n]$ such that $f_k = c k \alpha^{k-1} + d \alpha^k$,
where $c \not = 0$, and $\alpha \neq \pm i$.
Let $S = \left[\begin{smallmatrix} 1 & \frac{d-1}{n} \\ \alpha & c + \frac{d-1}{n} \alpha \end{smallmatrix}\right]$.
Note that $\det S = c \not = 0$.
Then $f$ can be expressed as
\[f = S^{\otimes n} [1, 1, 0, \dotsc, 0],\]
where $[1, 1, 0, \dotsc, 0]$ should be understood as a dimension $2^n$ column vector,
which has a 1 in entries with index weight at most one and 0 elsewhere.
This identity can be verified by observing that
\[[1, 1, 0, \dotsc, 0] = [1,0]^{\otimes n} + \frac{1}{(n-1)!} \Sym_n^{n-1}([1,0]; [0,1])\]
and we apply $S^{\otimes n}$ using properties of tensor product, $S^{\otimes n} [1,0]^{\otimes n} = \left(S [1,0]\right)^{\otimes n}$, etc.
We consider the value at index $0^{n-k} 1^k$, which is the same as the value at any entry of weight $k$.
By considering where the tensor product factor $[0,1]$ is located among the $n$ possible locations, we get
\[\alpha^k + k \left(c + \frac{d-1}{n} \alpha\right) \alpha^{k-1} + (n-k) \frac{d-1}{n} \alpha^{k} =  c k \alpha^{k-1} + d \alpha^k.\]

Let $T = \tfrac{1}{\sqrt{1 + \alpha^2}} \left[\begin{smallmatrix} 1 & \alpha \\ \alpha & -1 \end{smallmatrix}\right]$,
then $T = T^\intercal = T^{-1} \in  \mathbf{O}_2(\mathbb{C})$ is orthogonal,
and $R = T S = \left[\begin{smallmatrix} u & w \\ 0 & v \end{smallmatrix}\right]$ is upper triangular,
where $v, w \in \mathbb{C}$ and $u = \sqrt{1 + \alpha^2} \ne 0$.
However, $\det R = \det T \det S = (-1) c \not = 0$, so we also have $v \not = 0$.
It follows that
\begin{align*}
 T^{\otimes n} f
 &= (T S)^{\otimes n} [1, 1, 0, \dotsc, 0]\\
 &= R^{\otimes n} [1, 1, 0, \dotsc, 0]\\
 &= R^{\otimes n} \left([1,0]^{\otimes n} + \frac{1}{(n-1)!} \Sym_n^{n-1}([1,0]; [0,1])\right)\\
 &= [u,0]^{\otimes n} + \frac{1}{(n-1)!} \Sym_n^{n-1}([u,0]; [w,v])\\
 &= [u^n + n u^{n-1} w, u^{n-1} v, 0, \dotsc, 0].
\end{align*}
Since $u^{n-1} v \not = 0$,
we can normalize to~$1$ the entry of Hamming weight~$1$ by a scalar multiplication.
Thus, we have $[z, 1, 0, \dotsc, 0]$ for some $z \in \mathbb{C}$.

\end{document}